\documentclass{article}

\usepackage{arxiv}

\usepackage[utf8]{inputenc} 
\usepackage[T1]{fontenc}    
\usepackage{hyperref}       
\hypersetup{
    colorlinks=true,
    linkcolor=blue,    
    citecolor=blue,    
    urlcolor=blue      
}
\usepackage{url}            
\usepackage{booktabs}       
\usepackage{amsfonts}       
\usepackage{nicefrac}       
\usepackage{microtype}      
\usepackage{lipsum}		
\usepackage{graphicx}
\usepackage{natbib}
\usepackage{caption}
\usepackage{doi}
\usepackage{algorithm}
\usepackage{algpseudocode}
\usepackage{newfloat}
\usepackage{listings}
\usepackage{booktabs}

\usepackage{amsmath}
\usepackage{amssymb}
\usepackage{amsthm}
\usepackage{mathtools}
\usepackage{xcolor}
\usepackage{cleveref}
\usepackage{tikz}
\usepackage{enumitem}
\usepackage{cancel}
\usepackage{subcaption}
\usepackage{standalone}
\usetikzlibrary{backgrounds,shapes.geometric,positioning,fit, arrows.meta}
\usetikzlibrary{arrows.meta,positioning,calc,bending}

\tikzset{
  sp picture/.style={
    line cap=round,
    line join=round
  },
  sp vertex/.style={
    circle,
    draw,
    fill=white,
    line width=0.6pt,
    inner sep=0pt,
    minimum size=2.6mm
  },
  sp terminal/.style={
    sp vertex,
    minimum size=3.8mm,
    font=\scriptsize
  },
  sp edge/.style={
    draw,
    line width=0.9pt,
    -{Latex[length=2.2mm,width=1.8mm]},
    shorten >=0.4pt
  },
  sp edge label/.style={
    midway,
    fill=white,
    inner xsep=1.6pt,
    inner ysep=0.9pt,
    font=\scriptsize,
    align=center,
    text height=1.6ex,
    text depth=.25ex
  }
}

\newcommand{\E}{\mathbb{E}}

\newcommand{\Oh}[1]{\mathcal{O}(#1)}

\newcommand{\Ctrl}{\mathsf{Ctrl}}
\newcommand{\Ser}{\mathsf{Series}}
\newcommand{\Par}{\mathsf{Par}}
\newcommand{\ParAND}{\mathsf{ParAND}}
\newcommand{\ParOR}{\mathsf{ParOR}}
\newcommand{\one}{\mathbf{1}}
\newcommand{\Ps}{\Psi}
\newcommand{\Psh}{\widehat{\Psi}}
\newcommand{\al}{\alpha}

\newcommand{\sgn}{\sigma}
\newcommand{\Network}{\mathcal{N}}
\newcommand{\frontier}{\mathcal{F}}
\newcommand{\snk}{\mathrm{snk}}

\newcommand{\eqdef}{\vcentcolon=}

\DeclareMathOperator*{\argmax}{arg\,max}

\theoremstyle{plain}
\newtheorem{theorem}{Theorem}

\newtheorem{proposition}{Proposition}

\theoremstyle{definition}

\theoremstyle{remark}
\newtheorem{remark}{Remark}
\newtheorem{example}{Example}
\title{Security Games on Series–Parallel Attack Graphs with Adaptive Attackers}

\author{Russell Kai Min Tan\\
	School of Computing\\
	National University of Singapore\\
	\texttt{russelltankm@u.nus.edu} \\
	\And
	{Hui Han Chin}\\
	DSO National Laboratories\\
	\texttt{chuihan@dso.org.sg} \\
    \And
	{Chun Kai Ling}\\
	School of Computing\\
	National University of Singapore\\
	\texttt{chunkail@nus.edu.sg} \\
}

\renewcommand{\undertitle}{}

\begin{document}
\maketitle

\begin{abstract}
We study security games on attack graphs, where an adaptive attacker seeks to reach a target by sequentially attempting stochastic controls along the current attack frontier, while a defender allocates limited resources across controls to delay compromise. The attacker may choose among exponentially many attack routes and freely pivot between them as successes and failures are observed, yielding an exponentially large space of contingent attack policies.
For any fixed defender allocation, we show that an optimal attacker policy on a two-terminal series–parallel attack graph is an index policy: at each step, the attacker selects an available control with the largest value of an extension of the classical Gittins index. The indices and the resulting attacker best response can be computed in polynomial time, without explicitly enumerating attack paths or contingent policies. To the best of our knowledge, this is the first optimal index characterization for adaptive attackers in security games on general series–parallel attack graphs.
We further develop efficient algorithms for computing the attacker’s exact utility and an exact defender subgradient, enabling deterministic first-order optimization of defensive resource allocations without sampling attack trajectories. Our framework strictly generalizes prior approaches restricted to parallel chains and out-trees, while exploiting the compositional structure of series–parallel graphs to support interpretable attacker policies and parallel computation across independent subgraphs. Experiments demonstrate that the resulting methods scale substantially better than naive explicit-state approaches while producing effective defensive allocations.
\end{abstract}

\section{Introduction}
Security systems commonly deploy multiple complementary controls to
protect high-value assets \citep{jointtaskforce2020,nist_sp80082r3_2023}.
Attack graphs and attack/fault trees provide standard abstractions for
reasoning about how these controls interact, representing a successful
compromise or system failure through alternative routes, sequential
subgoals, and explicit defensive countermeasures
\citep{kordy2011foundations,kordy2014attack,jha2002two,ou2005mulval}.
Such models are widely used for vulnerability analysis, attack-path
evaluation, reliability assessment, and countermeasure selection
\citep{jurgenson2008computing, pietre2010beyond,durkota2015game,ruijters2019ffort}. In particular, a defender may harden selected nodes or
edges subject to a budget, trading off increased attacker effort
against deployment cost and disruption users
\citep{durkota2015game}.

The difficulty is that a strategic attacker need not commit to a single
source--target path. Attacks unfold over time: attempts may succeed or
fail stochastically, repeated failures may change the state of a control
or eventually cause lockout, and the attacker may pivot among exposed
controls as outcomes are observed. Such lateral movement and adaptive
route selection are central features of multistage attacks
\citep{knie2026graph}. The attacker therefore chooses from an exponentially
large space of contingent policies rather than merely from a set of
paths. \cite{ling2026security} obtain an exact Gittins-index solution
when the attack graph consists of independent parallel chains, but this
topology excludes nested alternatives, sequential subgoals, and
reconverging routes. We study two-terminal series--parallel graphs,
which recursively compose attack components in series and parallel and
naturally capture the modular structure present in attack and fault
models, including instances from established fault-tree benchmarks
\citep{ruijters2019ffort,formalmethods2026ffort}.

\section{Related Work}
Our work connects three literatures: security games, attack/fault graphs, and index policies for multi-armed bandits.

\paragraph{Security Games.}
Security games traditionally model a defender allocating scarce resources across potential targets, with applications including airport, transportation, and maritime security; subsequent work incorporates movement and sequential interaction through patrolling and extensive-form models \citep{pita2008armor,sinha2018stackelberg,basilico2009leader,lisy2016counterfactual}. Closest to our setting, \cite{ling2026security} studies a cybersecurity-motivated game in which a defender allocates resources to delay defensive controls, while an adaptive attacker seeks to reach a target quickly. They show that explicitly solving the attacker's best-response MDP is computationally prohibitive and that optimizing against heuristic attackers can produce highly suboptimal allocations. Their tractability results are restricted to parallel chains, where classical Gittins-index arguments apply. We broaden the class of security games admitting efficient attacker best responses and defender optimization to two-terminal series--parallel attack graphs.

\paragraph{Attack and Fault Graphs.}
Attack graphs and attack/fault trees compactly represent multistage routes to compromise or failure, and are widely used for vulnerability analysis, attack-path ranking, reliability evaluation, and countermeasure selection \citep{phillips1998graph,sheyner2002automated,jha2002attackgraphs,ou2006scalable,lallie2020review}. Their compositional AND/OR and series/parallel structure motivates our focus on series--parallel attack graphs. Prior work has also developed game-theoretic attack-graph models for network hardening, deception and honeypot placement, and edge interdiction \citep{durkota2015attackgraphgames,durkota2015game,milani2020harnessing,guo2022practical}. Such models typically have the attacker select a complete attack plan or path, or rely on generic mathematical-programming and approximation methods. In contrast, our attacker observes stochastic outcomes, interleaves attempts across branches, and may pivot freely during an attack. Thus, our analysis captures strategic adaptation within an attack, rather than only the ex ante selection of an attack path.

\paragraph{Gittins Indices and Branching Bandits.}
The classical Gittins-index theorem shows that, in a discounted multi-armed bandit whose independent arms evolve only when selected, an optimal policy always plays an arm with the largest current index \citep{gittins1979bandit,gittins2011multiarmed}.\footnote{Here, ``bandit'' refers to the classical discounted stochastic-control formulation, rather than the regret-minimization formulation commonly associated with UCB or multiplicative weights.} Subsequent work develops efficient index-computation algorithms and polyhedral interpretations, as well as related index methods for restless bandits and exact extensions to branching bandits \citep{bertsimas1996conservation,nino20072,whittle1988restless,weiss1988branching}. In a branching bandit, playing an arm may generate new arms, while an optimal Gittins priority rule is retained \citep{weiss1988branching}. More recently, \cite{choo2025adaptivefrontierexplorationgraphs} obtained an optimal index policy for adaptive frontier exploration on forests, with applications to network-based disease testing, while \cite{ling2026security} exploited a parallel-chain decomposition in security games. These settings allow tree-like branching or independent parallel chains but not the reconverging dependencies of general series--parallel graphs. To our knowledge, we give the first optimal Gittins-type index policy for a graph class supporting both branching and reconvergence.

\section{Security Games with SP Attack Graphs}
\label{sec:model}

We restrict our attention to attack graphs that are directed two-terminal series-parallel (SP) given by $\Network=(\mathcal{V},\mathcal{E})$. SP attack graphs have a unique source ($s$) and sink ($\varphi$); vertices represent intermediate attack states or milestones, while each edge represents a potential vulnerability that allows the attacker to progress from one state to another. Within each edge lies an atomic security measure (e.g. firewalls, authentication) that the defender may harden at some cost, hindering the attacker's ability to compromise the edge. An attack run is successful if the attacker can traverse a directed $s$--$\varphi$ path where every control in it has been compromised. Examples of SP attack graphs are shown in Figure~\ref{fig:example-edge}.

\begin{figure}[t]
\centering

\begin{subfigure}[t]{0.23\textwidth}
\centering
\begin{tikzpicture}[
  >=Stealth,
  thick,
  vtx/.style={
    circle,
    draw,
    fill=black,
    inner sep=0pt,
    minimum size=2.5pt
  },
  term/.style={
    inner sep=1pt,
    font=\small
  },
  lab/.style={
    font=\small,
    fill=white,
    inner sep=1pt
  }
]

\node[term] (s) at (0,0) {$s$};
\node[vtx] (u) at (1.15,0.9) {};
\node[vtx] (b) at (1.15,-0.9) {};
\node[vtx] (m) at (2.65,-0.9) {};
\node[term] (t) at (3.5,0) {$\varphi$};

\draw[->] (s) -- node[lab, above left] {$A$} (u);
\draw[->] (u) -- node[lab, above right] {$C$} (t);

\draw[->] (s) -- node[lab, below left] {$B$} (b);
\draw[->] (b) to[bend left=18]
  node[lab, above] {$D$} (m);
\draw[->] (b) to[bend right=18]
  node[lab, below] {$E$} (m);
\draw[->] (m) -- node[lab, below right] {$F$} (t);
\end{tikzpicture}
\caption{}
\label{fig:example-edge-compact}
\end{subfigure}
\begin{subfigure}[t]{0.23\textwidth}
\centering
\begin{tikzpicture}[
  >=Stealth,
  thick,
  vtx/.style={
    circle,
    draw,
    fill=black,
    inner sep=0pt,
    minimum size=2.4pt
  },
  lab/.style={
    font=\small,
    fill=white,
    inner sep=1pt
  }
]

\node[vtx, label=left:$s$] (s) at (0,0) {};
\node[vtx, label=right:$\varphi$] (t) at (2.3,0) {};

\node[vtx] (a1) at (1.15,0.95) {};

\node[vtx] (c1) at (1.15,-0.95) {};

\draw[->] (s)  -- node[pos=.5, above=2pt, lab] {$A_1$} (a1);
\draw[->] (a1) -- node[pos=.5, above=2pt, lab] {$A_2$} (t);

\draw[->] (s) -- node[pos=.5, above=2pt, lab] {$B$} (t);

\draw[->] (s)  -- node[pos=.5, below=2pt, lab] {$C_1$} (c1);
\draw[->] (c1) -- node[pos=.5, below=2pt, lab] {$C_2$} (t);

\end{tikzpicture}
\caption{}
\label{fig:three-parallel-chains}
\end{subfigure}

\caption{Examples of attack graphs that are two-terminal
series--parallel network with source \(s\) and sink \(\varphi\). The former is a running example used. The latter is an example of a parallel-chain graph used by \citet{ling2026security}.}
\label{fig:example-edge}
\end{figure}
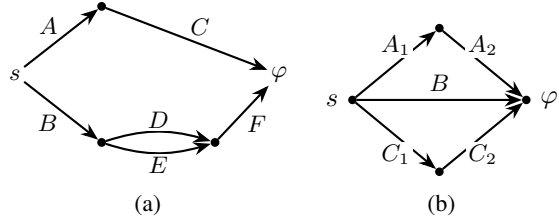

\subsection{Controls in SP Graphs}
The basic building block of a SP network is a control, which correspond to edges in the attack graph. Tagged to each control $e$ is its \textit{length} $\ell_e$, which captures the ``work factor'' or ``time'' that the attacker incurs every attempt it makes at compromising the control. \footnote{Aside from time, another equivalent interpretation is the ``rate'' that an attacker would be detected (thwarting the entire attack run) while attempting to compromise the control.} We denote by $p_e(k)$ the probability that an attempt to compromise $e$ is successful after $k$ prior failures. We make the following technical assumptions. 
\textbf{Assumption 1}: $p_e(k)$ are monotonically non-increasing in $k$ \citep{ling2026security}. \textbf{Assumption 2}: after $q_e > 0$ attempts, the probability of success drops to $0$, i.e., the attacker is \textit{locked out} of $e$ after $q_e$ failed attempts. 
These assumptions are based on existing work. 
\subsection{SP Attack Graphs and Parse Trees}
\label{sec:sp-topology}
The SP attack graph $\Network$ is written as a parse tree, 
where every sub-component is 
generated recursively by the grammar
\begin{align*}
\label{eq:n-nary-grammar}
    G \eqdef \Ctrl(e) \mid \Ser(G_1 ,\cdots, G_m) \mid \Par(G_1, \cdots, G_n).
\end{align*}

The leaves of the parse tree are $\Ctrl(e)$, where each $\Ctrl(e)$ can only appear \textit{once} in $G$. Each is visualized by a single directed edge $e \in \mathcal{E}$ in the attack graph as well as its two endpoints which serve trivially as source and sink.
These foundational units are recursively composed into larger sub-components by operations $\Ser$ and $\Par$.
A \textbf{series composition} $\Ser(G_1,\dots,G_m)$ 
merges each sink vertex of $G_i$ directly with the source vertex of $G_{i+1}$, forming a longer chain, with source equal to that of $G_1$ and sink that of $G_m$. Since the graphs are directed, $\Ser(G_1,\dots,G_m)$ is not commutative, i.e., $\Ser(G_1, G_2) \neq \Ser(G_2, G_1)$. 
A \textbf{parallel composition} $\Par(G_1,\dots,G_m)$ provides alternative routes under \textsc{or} semantics. The source and sinks of all constituents $G_i$ are merged to form a new common sink and source. 
Unlike $\Ser$, $\Par$ is topologically commutative. For ease of exposition, we will drop the $\Ctrl(\cdot)$ operator and directly use $A, B, C$ etc. to refer to a single edge/control.

\Cref{fig:example-edge-compact} shows a network we use as a running example with $\Network=\Par(\Ser(A,C),\allowbreak \Ser(B,\Ser(\Par(D,E),F)))$. 
\Cref{fig:three-parallel-chains} show an example of the special case of parallel-chains used by \cite{ling2026security}, here $\Network=\Par(\Ser(A_1, A_2), \allowbreak B,  \Ser(C_1,C_2))$, which more generally are of the form $\Par(\Ser(\dots), \Ser(\dots), \Ser(\dots))$.
A \textbf{sub-component} of the attack graph is one that is formed by a sub-tree of the \textit{parse tree} (\textit{not} to be confused a subgraph of the \textit{attack graph}). For example, in Figure~\ref{fig:example-edge-compact}, the graph induced by edges $(D,E)$ forms a subcomponent corresponding to $\Par(D, E)$. 
These compositions by our grammar maintains the invariant that any subcomponent $G$ always retains a single source and sink (and is acyclic). We denote the sink of a sub-component $G$ as $\snk(G)$.

\subsection{Temporally-Extended Attacks on SP Graphs}
In our paper, a single \textit{attack run} (henceforth \textit{run}) spans over a time period, and compromised controls expand the exposed attack surface, known as a frontier $\mathcal{F}$. Every failed attempt at compromising a control changes its underlying internal state (e.g., failure to enter a correct password for a long enough period may eventually lead to a lockout). We model this by ascribing to each control $e \in \mathcal{E}$ an \textit{internal state}.
The internal state of each control edge $e$ is its failure count, $k \in \{0, \dots, q_e\}$, bounded by the lockout limit $q_e$. 

\paragraph{Attack Frontier.}
At any point during an attack, the \emph{frontier}
$\mathcal F$ is the set of controls currently available for the
attacker to attempt. Initially, it consists of the controls leaving the
source $s$, and successful compromises expose new downstream controls.
We use \emph{eager pruning}: locked-out controls are removed from
$\mathcal F$. If an enclosing subcomponent becomes impossible to complete, this is propagated through the SP structure, and controls no longer able to contribute to reaching $\varphi$ are also pruned.
Similarly, completing any child of a parallel subcomponent completes
the entire subcomponent, after which all unfinished sibling
subcomponents are pruned. Thus, $\mathcal{F}$ contains
exactly those exposed controls whose outcomes can still affect whether
the target $\varphi$ is reached.

\paragraph{Anatomy of a run.} At each decision point, the attacker selects a control
$e\in\mathcal{F}$ to attempt. If $e$ has experienced $k$ prior
failures, the attempt succeeds with probability $p_e(k)$. \textbf{Upon failure},
$k$ is incremented. If $k$ reaches $q_e$, then $e$
is permanently locked out and removed from $\mathcal{F}$; otherwise, it
remains available for a subsequent attempt. \textbf{Upon success}, $e$ is
compromised and $\mathcal{F}$ is updated according to the series and
parallel semantics described above. If $e$ is adjacent to $\varphi$, the run is over. Otherwise, the attacker selects another control from the updated frontier and the process continues until $\varphi$ is exposed or the frontier becomes empty. Let $n_e$ denote the total number of attempts made on control $e$ during the run. The total elapsed time is $T=\sum_{e} n_e\ell_e$ and a successful run yields attacker utility $\exp(-\lambda T)$, where $\lambda>0$ is the discount rate. If the frontier becomes empty before $\varphi$ is reached, the run is unsuccessful and the attacker receives utility zero, or equivalently $T=\infty$.

\begin{figure}[!h]
\centering
\begin{subfigure}[t]{0.24\textwidth}
\centering
\begin{tikzpicture}[
  >=Stealth,
  thick,
  vtx/.style={
    circle,
    draw,
    fill=black,
    inner sep=0pt,
    minimum size=2.5pt
  },
  term/.style={
    inner sep=1pt,
    font=\small
  },
  lab/.style={
    font=\small,
    fill=white,
    inner sep=1pt
  }
]

\node[term] (s) at (0,0) {$s$};
\node[vtx] (u) at (1.15,0.9) {};
\node[vtx] (b) at (1.15,-0.9) {};
\node[vtx] (m) at (2.65,-0.9) {};
\node[term] (t) at (3.5,0) {$\varphi$};

\draw[->] (s) -- node[lab, above left] {$A:$ $\frac{2}{3}$} (u);
\draw[->] (u) -- node[lab, above right] {$C:$ $\frac{0}{5}$} (t);

\draw[->] (s) -- node[lab, below left] {$B:$ $\frac{0}{2}$} (b);
\draw[->] (b) to[bend left=18]
  node[lab, above] {$D:$ $\frac{0}{2}$} (m);
\draw[->] (b) to[bend right=18]
  node[lab, below] {$E:$ $\frac{0}{2}$} (m);
\draw[->] (m) -- node[lab, below right] {$F:$ $\frac{0}{3}$} (t);
\end{tikzpicture}
\caption{$\mathcal{F}=\{ A, B \}$}
\label{fig:anim-1}
\end{subfigure}
\begin{subfigure}[t]{0.24\textwidth}
\centering
\begin{tikzpicture}[
  >=Stealth,
  thick,
  vtx/.style={
    circle,
    draw,
    fill=black,
    inner sep=0pt,
    minimum size=2.5pt
  },
  term/.style={
    inner sep=1pt,
    font=\small
  },
  lab/.style={
    font=\small,
    fill=white,
    inner sep=1pt
  }
]

\node[term] (s) at (0,0) {$s$};
\node[vtx] (u) at (1.15,0.9) {};
\node[vtx] (b) at (1.15,-0.9) {};
\node[vtx] (m) at (2.65,-0.9) {};
\node[term] (t) at (3.5,0) {$\varphi$};

\draw[->] (s) -- node[lab, above left] {$A:$ $\frac{2}{3}$} (u);
\draw[->] (u) -- node[lab, above right] {$C:$ $\frac{0}{5}$} (t);

\draw[line width=2pt, ->] (s) -- node[lab, below left] {$B:$ $\frac{2}{2}$} (b);
\draw[->] (b) to[bend left=18]
  node[lab, above] {$D:$ $\frac{0}{2}$} (m);
\draw[->] (b) to[bend right=18]
  node[lab, below] {$E:$ $\frac{0}{2}$} (m);
\draw[->] (m) -- node[lab, below right] {$F:$ $\frac{0}{3}$} (t);
\end{tikzpicture}
\caption{$\mathcal{F}=\{ A, D, E \}$}
\label{fig:anim-2}
\end{subfigure}
\begin{subfigure}[t]{0.24\textwidth}
\centering
\begin{tikzpicture}[
  >=Stealth,
  thick,
  vtx/.style={
    circle,
    draw,
    fill=black,
    inner sep=0pt,
    minimum size=2.5pt
  },
  term/.style={
    inner sep=1pt,
    font=\small
  },
  lab/.style={
    font=\small,
    fill=white,
    inner sep=1pt
  }
]

\node[term] (s) at (0,0) {$s$};
\node[vtx] (u) at (1.15,0.9) {};
\node[vtx] (b) at (1.15,-0.9) {};
\node[vtx] (m) at (2.65,-0.9) {};
\node[term] (t) at (3.5,0) {$\varphi$};

\draw[->] (s) -- node[lab, above left] {$A:$ $\frac{2}{3}$} (u);
\draw[->] (u) -- node[lab, above right] {$C:$ $\frac{0}{5}$} (t);

\draw[line width=2pt, ->] (s) -- node[lab, below left] {$B:$ $\frac{2}{2}$} (b);
\draw[->] (b) to[bend left=18]
  node[lab, above] {$D:$ $\frac{1}{2}$} (m);
\draw[->] (b) to[bend right=18]
  node[lab, below] {$E:$ $\frac{1}{2}$} (m);
\draw[->] (m) -- node[lab, below right] {$F:$ $\frac{0}{3}$} (t);
\end{tikzpicture}
\caption{$\mathcal{F}=\{ A, D, E \}$}
\label{fig:anim-3}
\end{subfigure}
\begin{subfigure}[t]{0.24\textwidth}
\centering
\begin{tikzpicture}[
  >=Stealth,
  thick,
  vtx/.style={
    circle,
    draw,
    fill=black,
    inner sep=0pt,
    minimum size=2.5pt
  },
  term/.style={
    inner sep=1pt,
    font=\small
  },
  lab/.style={
    font=\small,
    fill=white,
    inner sep=1pt
  }
]

\node[term] (s) at (0,0) {$s$};
\node[vtx] (u) at (1.15,0.9) {};
\node[vtx] (b) at (1.15,-0.9) {};
\node[vtx] (m) at (2.65,-0.9) {};
\node[term] (t) at (3.5,0) {$\varphi$};

\draw[dotted, ->] (s) -- node[lab, above left] {$A:$ $\frac{3}{3}$} (u);
\draw[->] (u) -- node[lab, above right] {$C:$ $\frac{0}{5}$} (t);

\draw[line width=2pt, ->] (s) -- node[lab, below left] {$B:$ $\frac{2}{2}$} (b);
\draw[line width=2pt, ->] (b) to[bend left=18]
  node[lab, above] {$D:$ $\frac{2}{2}$} (m);
\draw[->] (b) to[bend right=18]
  node[lab, below] {$E:$ $\frac{1}{2}$} (m);
\draw[->] (m) -- node[lab, below right] {$F:$ $\frac{0}{3}$} (t);
\end{tikzpicture}
\caption{$\mathcal{F}=\{ F \}$}
\label{fig:anim-4}
\end{subfigure}
\caption{An illustration of a run unfolding over time. Bolded edges indicate compromised controls, dotted lines are controls that have been locked out. Each label shows the control and the number of attempts made on the control so far above its lockout limit $q_e$ (not to be confused with fractions) next to it. The current frontier is shown in the respective subcaptions.}
\label{fig:anim-full}
\end{figure}
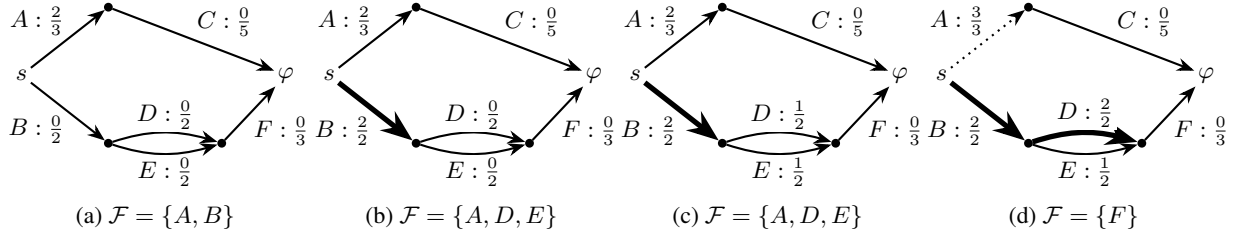

\begin{example} 
\label{ex:run-description}
Consider the run is based on the network topology $\mathcal{N}$ shown in Figure~\ref{fig:example-edge-compact} and is illustrated in Figure~\ref{fig:anim-full} based on $q_e$ illustrated there. (i) At the beginning, only controls $A$ and $B$ are in $\mathcal{F}$. The attacker attempts to compromise $A$. It tries twice and fails. This leads to Figure~\ref{fig:anim-1}. (ii) Now, the attacker decides that attempting $A$ yet again is not the best course of action, and pivots to $B$. It fails the first time but succeeds on the second attempt. Now, controls $D$ and $E$ are exposed (Figure~\ref{fig:anim-2}). (iii) At this point, the attacker can attempt either $A$, $D$, or $E$. It chooses $D$ and fails. Instead of continuing with $D$ again, it chooses to pivot and try $E$ and fails again (Figure~\ref{fig:anim-3}). (iv) Now, instead of continuing with either $D$ or $E$, the attacker pivots \textit{back} to $A$ but fails. This triggers a lockdown of control $A$ whose lockout limit $q_A=3$. Thus, only $D$ and $E$ are exposed. The attacker attempts $D$ once again and is successful. Since $E$ is parallel to $D$, it is removed from $\mathcal{F}$ (attacking it any more can never benefit the attacker). Also, control $F$ is now exposed (Figure~\ref{fig:anim-4}). (v) Now, only $F$ remains on the frontier. The attack is successful after two tries. This yield a successful run with a compromised $s-\varphi$ path of $B \rightarrow D \rightarrow F$. In total, controls $A,B,C,D,E,F$ were attempted $3$, $2$, $0$, $2$, $1$, $2$ times respectively, giving a total time of $T=3\ell_A + 2\ell_B + 2\ell_D + \ell_E + 2\ell_F$ and the attack receives a reward of $\exp(-\lambda T)$.
\end{example}

In Example~\ref{ex:run-description}, the attacker dynamically pivots between exposed controls, adapting to observed failures. This adaptive attacker strategy contrasts to a naive method of committing to a single path (e.g., $B \rightarrow D \rightarrow F$), but is significantly more complex. Note that runs differ substantially depending on which attacks succeed and how the attacker chooses to pivot (possibly even with the attacker failing), hence $T$ and $\exp(-\lambda T)$ are in general random variables.

An attack state is a tuple
$\sigma=(\mathcal F,(k_e)_{e\in\mathcal F})$, consisting of the current
frontier and the failure count of every control on it. A deterministic
attacker policy $\pi$ selects a control
$\pi(\sigma)\in\mathcal F$ at every nonterminal state $\sigma$.
Denote the set of such policies by $\Pi$.
For a fixed $\ell$ and attacker policy $\pi$, let $T$
denote the time when the attacker reaches $\varphi$, with
$T=\infty$ if the run is unsuccessful. This yields the value functions:
\[
    V_{\mathcal N}^{\pi}(\ell)
    :=\mathbb E\!\left[e^{-\lambda T}\mid \ell,\pi\right]
    \quad \text{and} \quad
    V_{\mathcal N}^{*}(\ell)
    :=\max_{\pi\in\Pi}V_{\mathcal N}^{\pi}(\ell).
\]

\subsection{Optimal Defense Allocation as a Game} \label{sec:defender}

Fix $\mathcal{N}$. The defender chooses $\ell_e$, but must balance security against usability --- extreme values of $\ell$ delay attackers but heavily penalise legitimate users. Like \cite{ling2026security}, this operational trade-off is captured by restricting the $\ell$ to the bounded polytope $\mathcal{L}$ 
the $|\mathcal{E}|$-simplex, $\mathcal{L} = \{ \ell \in \mathbb{R}^{|\mathcal{E}|}_+ \mid \mathbf{1}^T \ell = 1 \}$.
Given any $\ell$, the attacker best responds with $\pi^*_\ell = \argmax_{\pi \in \Pi} V_{\mathcal N}^{\pi}(\ell)$. Consequently, the defender seeks to minimise this worst-case scenario, formulating the interaction as the following min-max problem: 
\begin{equation} \min_{\ell \in \mathcal{L}} \max_{\pi\in\Pi} V^\pi_\mathcal{N}(\ell) = \min_{\ell \in \mathcal{L}} \max_{\pi\in\Pi} \mathbb E\!\left[e^{-\lambda T}\mid \ell,\pi\right]
\label{eq:minmax}
\end{equation}
Since $\mathcal{L}$ is compact and $V^*_\mathcal{N}$ continuous, the problem achieves a global minimum, with the maximum possible utility $V_\Network^{*}(\ell) \in [0,1]$. Finally, $V_\Network^{*}(\ell)$ is convex w.r.t. $\ell$ because it is the pointwise maximization over convex functions. 

\begin{remark}
    Since SP-networks are more general than parallel chains, the negative results of \cite{ling2026security} regarding heuristic-based solutions also apply. In particular, assuming heuristic attacker policies can be highly detrimental for either attacker and defender. 
    This shows the importance of respecting the attacker's ability to pivot. 
\end{remark}
\begin{remark}
While \eqref{eq:minmax} is identical to \cite{ling2026security} in terms of rewards and control structure, the move from parallel chains to general SP attack graphs necessitates a more complex frontier $\mathcal{F}$ and policy space $\Pi$. 
\end{remark}

\section{Solving for the Optimal Attacker Policies}
The min-max problem in \eqref{eq:minmax} can be solved by first-order methods (FOM) as long as one is able to readily able to obtain $\pi^*_\ell$ and subgradients $g(\ell) \in \partial V_\mathcal{N}^*(\ell)$ . 
We follow \cite{ling2026security} and utilize regret matching \citep{hart2013simple}, deferring details to the appendix. Thus, the rest of the section is devoted to our novel method of computing $\pi^*_\ell$ and $g(\ell)$ for a fixed $\ell$. Note about notation: in this paper, we write $g(\ell)$ to denote subgradients where $V^*_\mathcal{N}(\ell)$ is nonsmooth, typically it is the pointwise maximum of finitely many policies. When needed, we perform tiebreaks consistently (e.g., by lexicographic order). 

\subsection{Structure of Optimal Attacker Policies}

Let us focus on computing $\pi^*_\ell$. Clearly, the naive method of solving a huge MDP is inefficient; this was true even in the special case of parallel chains.  
Another naive method: expanding an SP network into independent parallel chains can also produce highly suboptimal attacker policies.

\begin{example}
\label{ex:counterexample-expand-sp}
Let $\mathcal N_n
    =
    \Par\!\left(
        A,\Ser\!\left(B,\Par(C_1,\ldots,C_n)\right)
    \right)$
with $\lambda=1$ (Figure~\ref{fig:decomposability-example}). Suppose $A$ and $B$ succeed deterministically ($p_A(0)=p_B(0)=1)$, with
$\ell_A=2.5$ and $\ell_B=1$, respectively. Each $C_i$ permits one
attempt ($q_{C_i}=1$), with $\ell_{C_i}=0.1$, $p_{C_i}(0)=0.1$. 
Thus, the attacker is choosing between the safe control $A$ and a
shared prerequisite $B$ that unlocks $n$ alternative downstream
controls. This models a situation where the attacker is choosing between a safe but slow attack in $A$, or, an initial attack on $B$ which opens up $n$ possible future paths. This decision depends on $n$: if $n=1$, then choosing $A$ is ideal; if $n$ is high, then the one-off attack on $B$ pays off.

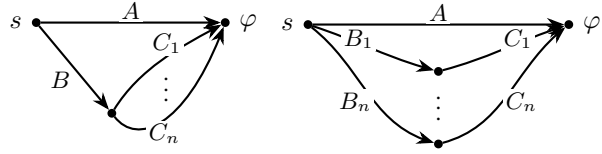
\begin{figure}[h!]
\centering
\begin{tikzpicture}[
  >=Stealth,
  thick,
  vtx/.style={
    circle,
    draw,
    fill=black,
    inner sep=0pt,
    minimum size=2.5pt
  },
  lab/.style={
    font=\small,
    fill=white,
    inner sep=1pt
  }
]

\node[vtx, label=left:$s$] (s) at (0,0) {};
\node[vtx] (m) at (1.0,-1.2) {};
\node[vtx, label=right:$\varphi$] (t) at (2.5,0) {};

\draw[->] (s) -- node[lab, above] {$A$} (t);

\draw[->] (s) -- node[lab, below left] {$B$} (m);

\draw[->]
  (m) .. controls (1.4,-0.55) and (2.05,-0.25) ..
  node[pos=.5, lab, above] {$C_1$}
  (t);

\node[font=\small] at (1.72,-0.78) {$\vdots$};

\draw[->]
  (m) .. controls (1.35,-1.65) and (2.0,-1.35) ..
  node[pos=.5, lab, below] {$C_n$}
  (t);
\end{tikzpicture}
\begin{tikzpicture}[
  >=Stealth,
  thick,
  vtx/.style={
    circle,
    draw,
    fill=black,
    inner sep=0pt,
    minimum size=2.5pt
  },
  lab/.style={
    font=\small,
    fill=white,
    inner sep=1pt
  }
]

\node[vtx, label=left:$s$] (s) at (0,0) {};
\node[vtx, label=right:$\varphi$] (t) at (3.45,0) {};

\node[vtx] (m1) at (1.725,-0.62) {};
\node[vtx] (mn) at (1.725,-1.58) {};

\draw[->]
  (s) -- node[lab, above] {$A$} (t);

\draw[->]
  (s) to[out=-18,in=165]
  node[pos=.42, lab, above, xshift=-2pt, yshift=-2.5pt] {$B_1$}
  (m1);

\draw[->]
  (m1) to[out=15,in=-162]
  node[pos=.58, lab, above, xshift=2pt, yshift=-2.5pt] {$C_1$}
  (t);

\node[font=\small] at (1.725,-1.00) {$\vdots$};

\draw[->]
  (s) to[out=-38,in=165]
  node[pos=.42, lab, below, xshift=-2pt, yshift=-2pt] {$B_n$}
  (mn);

\draw[->]
  (mn) to[out=15,in=-142]
  node[pos=.58, lab, below, xshift=2pt, yshift=-2pt] {$C_n$}
  (t);

\end{tikzpicture}
\caption{The graphs $\mathcal{N}_n$ (left) and $\widetilde{N}_n$ (right) in Example~\ref{ex:counterexample-expand-sp}.
}
\label{fig:decomposability-example}
\end{figure}

Now consider a naive \textbf{alternative} graph
$\widetilde{\mathcal N}_n
    =
    \Par\!\left(
        A,\Ser(B_1,C_1),\ldots,\Ser(B_n,C_n)
    \right),
$
where the $B_i$ are independent copies of $B$. Each $B_i$ unlocks only
one downstream control $C_i$, 
Now, the one-off nature of breaching $B$ is not accounted for, and $A$ remains the best option regardless of $n$. An optimal attacker in this alternative always selects $A$.
For $n=20$, deploying this policy $\widetilde{\pi}^*$ on the true network $\mathcal{N}_n$ obtains value $0.0821$, whereas the optimal SP-aware policy obtains value $0.1769$.
\end{example}

Example~\ref{ex:counterexample-expand-sp} shows that SP networks are not easily ``expanded'' into parallel-chains, and that tackling SP networks requires approaches beyond \cite{ling2026security}.
A natural approach is dynamic programming: focus on a single subcomponent $G$ (e.g., the parallel bundle in Example~\ref{ex:counterexample-expand-sp}), compute the optimal strategy $\pi^*$ for that subcomponent (treating source and sink as subcomponent terminals) and recombine them recursively bottom-up from the parse-tree. This intuition turns out to be correct, but requires more nuance. For suppose we efficiently found a optimal policy some sub-component $G$, say, using Gittins indices. Adopting this optimal policy for this sub-component results in a Markov Process, which may be treated as a single ``mega-arm'' which can be composed with its sibling components. 

However, this Markov Process written explicitly will have exponentially many states in $n$, and it is not obvious how to efficiently combine or compare these exponentially many states to other options when handling further recursion.

\subsection{Indexability and Optimal Attacker Policies}
The discussion above suggests that the exponentially large state of a
subcomponent might be summarized by a scalar priority assigned to each
available control state (and their internal failure count). Our main result shows that such priorities --- also known as \textit{indices} --- exist and are sufficient to characterize an optimal attacker policy.

\begin{theorem}[Optimal index policy]
\label{thm:index-policy}
Fix an SP network $\mathcal N$ satisfying Assumptions~1--2 and a
defender allocation $\ell$. There exists an index
$\alpha(e,k_e)\in\mathbb R$ for every available control $e$ at failure
count $k_e$, determined by its local state and the calibrated
continuation profile of the forward cone (the sub-component remaining to be breached upon breaching $e$) exposed. The index does not depend on the states of sibling subcomponents. Moreover, there exists an optimal adaptive attacker
policy $\pi_\ell^*$ satisfying
$
    \pi_\ell^*(\sigma)
    \in
    \arg\max_{e\in\mathcal F(\sigma)}
    \alpha(e,k_e)
$
at every nonterminal attack state $\sigma$.
\end{theorem}

We call $\alpha(e,k_e)$ the \emph{index} of control state
$(e,k_e)$. Theorem~\ref{thm:index-policy} reduces the attacker's
sequential decision problem to a greedy rule once these indices are
known. These indices are analogous to the classic Gittins index \citep{gittins1979bandit}, which was the cornerstone of \cite{ling2026security}.  We now define these indices and explain
the why this rule is optimal. Section~4.3 gives
an efficient algorithm for computing them.

To define these indices, we introduce an auxiliary
retirement problem in which the attacker may abandon the current
subcomponent in exchange for a buyout $\gamma$. The buyout represents
the value of the attacker's alternatives elsewhere on the frontier and
separates two decisions: whether to continue the attack and,
conditional on continuing, which available control to attempt. Suppose we have a sub-component $G$, state $\sigma$, exit reward $r>0$, and buyout $\gamma\in[0,r]$. We define its \textit{calibrated value profile} $\Ps_G$, where $\Ps_G(\gamma;\sgn\mid r)\eqdef\sup_{\pi}\E_\sgn\big[e^{-\lambda T} r\,\one\{T\le\tau_\pi\}+\gamma\,e^{-\lambda\tau_\pi}\one\{T>\tau_\pi\}\big]$, where $T$ is the time required to complete $G$ and $\tau_\pi$ is the retirement time. Retirement may be voluntary or forced when no route through $G$ remains. The original attacker problem corresponds to $\gamma=0$ and $r=1$ at the root.

For any fixed policy, the calibrated payoff is affine in $\gamma$ and $r$. Since the attack process is finite, $\Ps_G(\gamma; \sigma \mid r)$ is therefore the maximum of finitely many affine functions, and hence convex and piecewise linear in $\gamma$. Immediate retirement gives $\Ps_G(\gamma; \sigma \mid r) \geq \gamma$, with equality at $\gamma = r$. Thus, there \textit{exists} a threshold which we define $\alpha_G(\sigma \mid r) = \inf \{\gamma \in [0, r] : \Ps_G(\gamma; \sigma \mid r) = \gamma\}$, below which $G$ is worth pursuing and at or above which retirement is optimal. In this sense, $\alpha_G$ is exactly a priority: \textbf{how attractive can the outside alternative become before the attacker would rather abandon G?} Positive homogeneity of $\Ps_G$ lets us normalise to $r = 1$.

The subtlety is that $G$ is not an ordinary bandit arm. After deciding to continue with $G$, the attacker must still choose \textit{which} frontier control inside $G$ to attack, and in general that choice \textit{could} depend on $\gamma$. The SP structure \textit{prevents} this. At a $\Ser$ node, only the current child is active while its downstream children remains frozen. $\Ser$ composition therefore changes only the continuation obtained after success. At a $\Par$ node, the viable children are alternatives. The unchosen children remain frozen, and completing one child completes the parallel subcomponent and prunes the others. These two operations preserve the same maximising internal control throughout the continuation region.

Consequently, for every reachable state of every SP subcomponent $G$, we have that $\alpha_G(\sigma \mid 1) = \max_{e \in \frontier_G(\sigma)} \alpha(e, k_e)$ and a fixed tie-broken maximiser remains optimal for every buyout below this threshold. This is precisely the buyout-independent-control condition demonstrated. By the special Gittins theorem \citep{doi:https://doi.org/10.1002/9780470980033.ch4} which says that the Gittins policy (picking the largest index control to attack) is optimal. Structural induction over the SP parse tree therefore yields \Cref{thm:index-policy}.

This also explains \textit{locality}. Once $e$ is breached, only its downstream $\Ser$ continuation remain whereas in a parallel subcomponent, its siblings are pruned. Hence $\alpha(e, k_e)$ depends on $e$'s local state and its forward continuation, but not on sibling states. At the root, $\gamma = 0$ and $r = 1$, so repeatedly attacking a maximum-index frontier control is optimal. Full proofs are in the appendix.

\subsection{Compositional Index Computation}\label{sec:componsitional_index_comp}
\Cref{thm:index-policy} reduces the attacker problem to computing the control-state indices. The key observation is that a control does not need the state of the whole network but only the knowledge of what a successful compromise leads to. Let $\beta_e = \exp(-\lambda \cdot \ell_e)$. For $\beta_e < 1$, define $c(e, k) \eqdef \frac{\beta_e p_e(k)}{1 - \beta_e(1 - p_e(k))}$, with boundary values $c(e,k) = 1$ when $\beta_e = 1, p_e(k) > 0$ and $c(e, k) = 0$ when $\beta_e = 1, p_e(k) = 0$. This quantity summarises the local control. If successfully breaching $e$ immediately paid a fixed reward $r$, the index would simply be $c(e, k)r$. The full derivation can be found in the appendix.

In the network, however, success at $e$ exposes a continuation. Let $R_e(\gamma)$ be the calibrated value of everything that must still happen after breaching $e$, under outside option $\gamma$. The index is therefore the smallest solution of $\alpha(e, k) = c(e, k)R_e(\alpha(e, k))$. This equation has a useful interpretation. The left side is the outside option at which the attacker is indifferent, while the right side is the value of attacking $e$ after $k$ prior failed attempts (on $e$) and then receiving its continuation. Since $R_e$ is piecewise linear, the solution is cheap to find; on a piece $R_e(\gamma) = A +B\gamma$, we have that $\alpha(e, k) = \frac{c(e,k)A}{1 - c(e, k)B}$. Moreover, because $p_e(k)$ decreases with $k$, all indices of $e$ can be obtained in one monotone sweep. 

So the remaining task is to compute every continuation profile $R_e$. We do this in two passes over the same SP parse tree. For a \textit{fresh} subcomponent $G$ (denoted by state $\sigma_G^0$), let $\Psh_G(\gamma) \eqdef \Ps_G(\gamma; \sigma_G^0\mid 1)$. Think of $\Psh_G$ \textbf{as a polyline summarising the entire subtree $G$}. Its breakpoints are precisely the values of $\gamma$ at which the optimal behaviour inside $G$ changes. We describe the algorithm below:

\textbf{Pass 1: Summarising each subcomponent.}  Starting from the controls, we replace progressively larger subcomponents by their profiles. For $G = \Par(G_1, \cdots, G_m)$, all children face the same outside option. We therefore sweep their breakpoints together. Between two consecutive breakpoints every child lies on a fixed affine segment, and the slope of the parent segment is the product of the corresponding child slopes. For $G = \Ser(G_1, \cdots, G_m)$, we instead work backward from $G_m$: the value of completing $G_j$ is whatever can subsequently be obtained from $G_{j+1}, \cdots, G_m$. Folding these continuations from downstream to upstream produces $\Psh_G$, while storing the suffix after each child for Pass 2. These operations preserve piecewise linearity, and the root node gives $V^*_\Network (\ell) = \Psh_\Network(0)$. As profiles are piecewise linear in $\gamma$, computation is easy as it is a matter of recording breakpoints and gradient slopes. 

\textbf{Pass 2: Supply each control's context.} Start above the root $R \equiv 1$. A $\Par$ node passes the same surrounding continuation profile to every child, because its children are alternatives to one another. A $\Ser$ node passes each child its stored downstream suffix profile followed by the continuation outside the $\Ser$ node, using the homogeneity $\Ps_G(\gamma; \sigma \mid r) = r \Psh_G(\gamma / r)$. When this pass reaches a control (which is a leaf on the parse tree) $e$, its continuation profile is exactly $R_e$, and the fixed-point equation above yields all of its indices.

\begin{proposition}[Index Computation Complexity]
\label{prop:complexity}
Let $d$ be the height of the SP parse tree of $\Network$. The value $\Ps_\Network$ and all $Q$ indices are computed in $\Oh{Qd}$ time and $\Oh{Qd}$ space.
\end{proposition}

Each sub-component $G$ has at most $Q_G+1$ profile pieces, where $Q_G=\sum_{e\in G}q_e$. Consequently, evaluating any $n$-ary $\Ser$ or $\Par$ node requires only a linear-time monotone merge of its children, at a cost of $\Oh{\sum_i Q_{G_i}}$. Aggregating over the tree, each control is processed at most once per ancestor, yielding a total time complexity of $\Oh{\sum_e q_e d_e} \le \Oh{Qd}$. This scales from $\Oh{Q\log n}$ on balanced networks to $\Oh{mQ}$ for $k$ parallel chains of length $m$. For space complexity, the two-pass algorithm only needs to store the piecewise-linear value functions. Memoising them during the bottom-up fold (Pass 1) for use in the top-down index pass retains each control's breakpoints at every ancestor, storing $\sum_G Q_G = \sum_e q_e d_e \le Qd$ breakpoints in aggregate. As the output table stores exactly $Q$ indices, the working set is bounded by $\Oh{Qd}$.

\subsection{Exact Subgradients by Reversing the Fold}\label{sec:defender}
\Cref{sec:componsitional_index_comp} gives a compact computation of the attacker value: from $\ell$ to leaf profiles to subcomponent profiles to $V^*_\Network(\ell)$. The defender asks the reverse question: \textbf{how much does the root value depend on each $\ell_e$?} Rather than enumerate attack trajectories, we differentiate this same computation backward.

We introduce the $\mathbf{n}$\textbf{-ary Reverse-Tape} ($n$-RT) algorithm records the arithmetic used by the bottom-up profile fold, including the construction of affine pieces and their breakpoint locations. After evaluating $V^*_\Network(\ell)$, the the tape is traversed in reverse and the chain rule propagates sensitivity from the root back to every control. Breakpoints must be recorded as computed quantities because their locations also vary with $\ell$.

The only additional issue is high arity. At a $\Par$ node with $m$ children, the active parent slope is a product of child slopes, so a balanced product tree processes each child breakpoint in $\Oh{\log m}$ time. At a $\Ser$ node, define the map $\tau_G(\gamma) = \frac{\gamma}{\Psh_G(\gamma)}$. On each affine piece this map is fractional-linear and these maps are closed under Series composition, so a balanced tree likewise maintains the current composition in $\Oh{\log m}$ time per breakpoint. Every recorded forward operation is then reversed once. A more detailed derivation and breakdown of the algorithm is presented in the appendix due to space constraints and the technical complexity of the algorithm.

At a leaf (a control), the reverse pass yields $\partial V^*_\Network / \partial \beta_e$. Since $\beta_e = \exp(-\lambda \cdot \ell_e)$, the subgradient is then $g_e(\ell) = -\lambda \cdot \beta_e \cdot \partial V^*_\Network / \partial \beta_e$. Thus a single reverse pass obtains all defender subgradients simultaneously. When the active affine pieces and breakpoint ordering are locally fixed, the recorded computation is differentiable and $n$-RT returns the exact subgradient. At a non-smooth point, simultaneous events are processed consistently and the fixed tie rule selects a limiting gradient. Convexity of $V^*_\Network(\ell)$ makes this a valid subgradient. The formal tie argument is deferred to the appendix.

\begin{proposition}[Subgradient complexity]
\label{prop:subgradient-complexity}
For an internal node $G$, let $m_G$ denote its arity, and define \(\delta \eqdef \max_{e\in E}\sum_{G\succeq e}\log m_G,\) where the sum ranges over the ancestors of $e$ (denoted by $\succeq e$) in the parse tree. At a regular parameter vector, $n$-RT computes $V_{\mathcal N}^*(\ell)$ and its exact subgradient in \(\mathcal O\big(\sum_G Q_G\log m_G\bigr) =\Oh{Q\delta}\) time and space. Under the fixed tie-breaking rule, the same procedure returns a valid subgradient at nonsmooth points.
\end{proposition}

A node $G$ processes at most $\Oh{Q_G}$ breakpoint events, each requiring
$\Oh{\log m_G}$ segment-tree operations. Hence its cost is
$\Oh{Q_G\log m_G}$. Summing over the parse tree and charging each control
state to its ancestors gives
\(\sum_G Q_G\log m_G=\sum_e q_e\sum_{G\succeq e}\log m_G \le Q\delta.\)
In particular, $\delta\le d\log M$ when every node has arity at most
$M$; bounded arity therefore gives $\Oh{Qd}$ complexity, while a single
flat $m$-ary node costs $\Oh{Q\log m}$. The same count bounds the space: each event leaves constant-size tape records, so the tape holds $\Oh{Q_G\log m_G}$ entries per node and $\Oh{Q\delta}$ in total. Unlike the forward fold, the tape cannot be released as it is built, since reverse mode replays it in full; the retained profiles, totalling $\sum_G Q_G\le Qd\le Q\delta$, are subsumed.

\subsection{Intractability in the Parallel-AND setting}
Recall that in our model $\Par(G_1, G_2, \dots G_n)$ means that as long as \textit{any} of $G_1$ is compromised, then the entire sub-component is as well. Another natural interpretation is that \textit{all} of $G_1,\dots,G_n$ need to be compromised. This would be similar to our $\Ser$ formulation, but does \textit{not} impose sequential constraints (and is hence commutative). Let us distinguish these two notions of ``parallel'' by $\ParOR$ and $\ParAND$ respectively. We have the following negative result, whose proof is left to the Appendix.
\begin{theorem}[Informal] 
Extending the grammar of our SP-networks to allow both $\ParOR$ and $\ParAND$ makes finding the optimal solution $\pi^*$ for arbitrary $\ell$ weakly NP-hard. 
\end{theorem}

\section{Experimental Evaluation}\label{sec:experiments}

\paragraph{Experimental Setup and Validation} All experiments were done on a dedicated 64-bit workstation running Ubuntu 24.04.3 LTS and equipped with 125 GiB of RAM. Before the experiments, we performed some basic verifications. For example, to verify that the Gittins policy is optimal, we validate both the computed values and the index policy against an independent brute-force MDP oracle that solves the game exactly by exhaustive backward induction over the joint state space. As this oracle is exponential in the number of controls, we apply it only to small networks. For subgradient computation, we use the centered finite-differences method to verify the correctness of the exact subgradients from our algorithms.

\paragraph{Synthetic and Real-World SP Networks used} We use a total of 5 classes of SP network topologies, 3 synthetically generated and 2 from real-world examples obtained from the FFORT website \citep{ffort_website}. The 3 synthetic classes are \textbf{SH} (series-heavy), \textbf{PH} (parallel-heavy) and \textbf{MN} (mixed-nested), whereas the 2 real-world classes are attack fault trees which model an attack on a password-protected file \citep{phillips1998graph}(denoted by \textbf{PWD}) and an attack to obtain administrator privileges of a system \citep{jurgenson2008computing}(denoted by \textbf{ADM}). In our experiments, $n$ is the number of controls and $Q$ is the sum of allowed failures across all controls $Q \eqdef \sum_e q_e$.

\subsection{Computational Efficiency of the Gittins Index} 
To evaluate scalability, we benchmarked our recursive method against a standard MDP solver (capped at a 15-minute timeout). Unsurprisingly, the standard MDP solver timed out on all of the instances. In stark contrast, the folding algorithm computed the Gittins indices for every instance in under 15 seconds.

\subsection{Evaluation of Exact Subgradient Algorithms} 
Our main results are reported in \Cref{tab:exp2_results}.
All methods use regret matching \citep{hart2013simple} and differ only in the gradient oracle: $n$-RT and its 8-core variant $n$-RT-8 are exact, while SG and SG-8 use stochastic estimates with batch sizes 1 and 8 (the latter across 8 cores). We adopt a time-to-target protocol: run $n$-RT until its objective nearly converges, denote that value $\widehat V$, and record for each method the wall-clock time---and, for SG and SG-8, the iterations---needed to first reach $\widehat V$ from a common initialisation. We observe that $n$-RT is generally faster than SG, and on 8 cores significantly outperforms both stochastic variants. In particular, our methods are more than $10$ times faster than stochastic methods for ADM. The only exception is \textbf{PH}, whose highly parallel structure yields many attack paths and hence near-certain success, giving a low-variance estimator and fast convergence. This is corroborated by the observation that SG-8 is slightly slower than SG, the overheads from parallelization exceed it's benefits even a batch size of 1. Otherwise SG is competitive only at small $Q$, with $n$-RT scaling better as $Q$ grows. 

\begin{table*}[h!]
  \centering
  \begin{tabular}{c c c | c c c c | c c c}
    \multicolumn{3}{c|}{Instances} & \multicolumn{4}{c|}{Runtime (min)} & \multicolumn{3}{c}{Iterations} \\
    \hline
    Class & $n$ & $Q$ & SG & SG-8 & $n$-RT  & $n$-RT-8 & Exact & SG & SG-8 \\
    \hline
    \textbf{SH}  & 450 & 20K & $305 \pm 82.0$  & $104 \pm 12.0$ & $240 \pm 9.40$  & $\mathbf{71.4} \pm 3.60$  & 5K  & 22.9K  & 6.75K\\
    \textbf{MN}  & 600 & 20K & $652 \pm 314$ & $356 \pm 58.8$ & $130 \pm 7.60$ & $\mathbf{53.1} \pm 7.10$  & 5K  & 187K & 48.1K\\
    \textbf{PH}  & 450 & 20K & $\mathbf{131} \pm 52.0$  & $140 \pm 47.6$ & $354 \pm 14.0$ & $161 \pm 11.5$ & 25K & 41.4K  & 27.7K\\
    \hline
    \textbf{PWD} & 11  & 10K & $81.7 \pm 16.6$   &$76.6 \pm 20.1$ & $31.8 \pm 3.60$ & $\mathbf{15.4} \pm 0.40$  & 20K & 35.1K  & 22.3K\\
    \textbf{ADM} & 12  & 10K & $270 \pm 187$ & $232 \pm 170$ & $20.6 \pm 2.00$ & $\mathbf{11.1} \pm 0.30$  & 15K & 252K & 112K\\
  \end{tabular}
  \caption{Wall-clock runtime and iteration counts required to reach the reference objective across network classes. Runtime values are reported in minutes (mean $\pm$ standard deviation) over 5 runs. Boldface indicates the lowest runtime for each class.}
  \label{tab:exp2_results}
\end{table*}

In addtion, we compared runtimes using the dynamic programming (DP) approach (designed specifically for parallel chains) by \cite{ling2026security} for subgradient computation. Runtime averages for 5000 iterations are reported in \Cref{tab:m-ary_vs_lck}. We find that our algorithm indeed outperforms the DP approach significantly as $n, Q$ is large, despite both being exact methods.

\begin{table}[h!]
    \centering
    \begin{tabular}{c|cc|cc}
        \# Chains & $n$ & $Q$ & $n$-RT (min) & DP (min) \\\hline
        4 & 20 & 1K & $0.79 \pm 0.03$ & $1.23 \pm 0.04$ \\
        5 & 50 & 2K & $4.13 \pm 1.21$ & $6.24 \pm 0.94$ \\
        6 & 75 & 4K & $10.2 \pm 0.55$ & $17.5 \pm 0.74$ \\
        8 & 90 & 5K & $8.02 \pm 1.52$ & $21.2 \pm 1.39$ \\
        6--9 & 400 & 10K & $63.5 \pm 9.35$ & $299\pm 65.0$ \\
    \end{tabular}
    \caption{Comparison of runtimes in minutes for parallel chains between
    $n$-RT and DP from \cite{ling2026security}. Values are reported as
    mean $\pm$ standard deviation.}
    \label{tab:m-ary_vs_lck}
\end{table}

\section{Conclusion}

In this paper, we developed a framework for efficiently computing the Gittins indices for an adaptive attacker targeting an SP network of controls and further gave efficient algorithms to compute the optimal allocations of resources for the defender against an optimal attacker. We also show NP-hardness in the Parallel-AND setting. Our experiments do show strong performance of our algorithms for both computing the Gittins indices and subgradients which outperforms stochastic gradients and existing methods for parallel chains.

\section*{Acknowledgements}
This research/project is supported by the National Research Foundation, Singapore under its AI Singapore Programme (AISG Award No: AISG3-AMP/2025-08-005), and is supported by the National University of Singapore, under the Startup-Grant Scheme.

\bibliographystyle{biblio}
\bibliography{references}

@inproceedings{ling2026security,
 author = {Chun Kai Ling and Jakub Cerny and Chin Hui Han and Garud Iyengar and Christian Kroer},
 booktitle = {AAAI (Oral)},
 date = {2026-02},
 title = {Security Games with Layered Defenses: Adaptive Adversaries and Gittins Indices},
 url = {},
 year = {2026}
}

@misc{choo2025adaptivefrontierexplorationgraphs,
      title={Adaptive Frontier Exploration on Graphs with Applications to Network-Based Disease Testing}, 
      author={Davin Choo and Yuqi Pan and Tonghan Wang and Milind Tambe and Alastair van Heerden and Cheryl Johnson},
      year={2025},
      eprint={2505.21671},
      archivePrefix={arXiv},
      primaryClass={cs.AI},
      url={https://arxiv.org/abs/2505.21671}, 
}

@inbook{doi:https://doi.org/10.1002/9780470980033.ch4,

publisher = {John Wiley \& Sons, Ltd},
author = {John Gittins and Kevin Glazebrook and Richard Weber},
isbn = {9780470980033},
title = {Superprocesses, Precedence Constraints and Arrivals},
booktitle = {Multi‐Armed Bandit Allocation Indices},
chapter = {4},
pages = {79-114},
doi = {https://doi.org/10.1002/9780470980033.ch4},
url = {https://onlinelibrary.wiley.com/doi/abs/10.1002/9780470980033.ch4},
eprint = {https://onlinelibrary.wiley.com/doi/pdf/10.1002/9780470980033.ch4},
year = {2011}
}

@article{knie2026graph,
  title={A Graph-Based Infrastructure for Characterizing Structural Risk and Lateral Movement Patterns in APT Campaigns},
  author={Knie, Trever and Mink, Dustin and Bagui, Sikha and Bagui, Subhash},
  journal={IEEE access},
  year={2026},
  publisher={IEEE}
}

@inproceedings{kordy2011foundations,
  author    = {Kordy, Barbara and Mauw, Sjouke and Radomirovi{\'c}, Sa{\v{s}}a and Schweitzer, Patrick},
  title     = {Foundations of Attack--Defense Trees},
  booktitle = {Formal Aspects of Security and Trust},
  series    = {Lecture Notes in Computer Science},
  volume    = {6561},
  pages     = {80--95},
  publisher = {Springer},
  year      = {2011},
  doi       = {10.1007/978-3-642-19751-2_6}
}

@inproceedings{jha2002attackgraphs,
  author    = {Jha, Somesh and Sheyner, Oleg and Wing, Jeannette M.},
  title     = {Two Formal Analyses of Attack Graphs},
  booktitle = {Proceedings of the 15th IEEE Computer Security Foundations Workshop},
  pages     = {49--63},
  year      = {2002},
  doi       = {10.1109/CSFW.2002.1021806}
}

@inproceedings{ou2005mulval,
  author    = {Ou, Xinming and Govindavajhala, Sudhakar and Appel, Andrew W.},
  title     = {{MulVAL}: A Logic-Based Network Security Analyzer},
  booktitle = {Proceedings of the 14th USENIX Security Symposium},
  year      = {2005},
  url       = {https://www.usenix.org/conference/14th-usenix-security-symposium/mulval-logic-based-network-security-analyzer}
}

@inproceedings{durkota2015attackgraphgames,
  author    = {Durkota, Karel and Lis{\'y}, Viliam and Kiekintveld, Christopher and Bo{\v{s}}ansk{\'y}, Branislav},
  title     = {Game-Theoretic Algorithms for Optimal Network Security Hardening Using Attack Graphs},
  booktitle = {Proceedings of the 2015 International Conference on Autonomous Agents and Multiagent Systems},
  pages     = {1773--1774},
  year      = {2015},
  url       = {https://www.ifaamas.org/Proceedings/aamas2015/aamas/p1773.pdf}
}

@inproceedings{ruijters2019ffort,
  author    = {Ruijters, Enno and Budde, Carlos E. and Nakhaee, Mohammad Amin and Stoelinga, Mari{\"e}lle and Bucur, Doina and Hiemstra, Djoerd},
  title     = {{FFORT}: A Benchmark Suite for Fault Tree Analysis},
  booktitle = {Proceedings of the 29th European Safety and Reliability Conference (ESREL 2019)},
  year      = {2019},
  publisher = {Research Publishing},
  doi       = {10.3850/981-973-0000-00-0_main},
  url       = {https://djoerdhiemstra.com/wp-content/uploads/esrel2019.pdf}
}

@misc{ffort_website,
  author       = {{FMT Group, University of Twente}},
  title        = {{FFORT}: The Extended Fault Tree Forest},
  year         = {2026},
  url          = {https://dftbenchmarks.utwente.nl/},
  note         = {Accessed 2026-07-10}
}

@techreport{nist_sp80082r3_2023,
  author      = {Stouffer, Keith and Pillitteri, Victoria and Lightman, Suzanne and Abrams, Marshall and Hahn, Adam},
  title       = {Guide to Operational Technology ({OT}) Security},
  institution = {{NIST}},
  type        = {Special Publication},
  number      = {800-82 Revision 3},
  year        = {2023},
  doi         = {10.6028/NIST.SP.800-82r3},
  url         = {https://csrc.nist.gov/pubs/sp/800/82/r3/final},
  note        = {National Institute of Standards and Technology}
}

@book{hart2013simple,
  title={Simple adaptive strategies: from regret-matching to uncoupled dynamics},
  author={Hart, Sergiu and Mas-Colell, Andreu},
  volume={4},
  year={2013},
  publisher={World Scientific}
}

@inproceedings{pietre2010beyond,
  title={Beyond attack trees: dynamic security modeling with Boolean logic Driven Markov Processes (BDMP)},
  author={Pi{\`e}tre-Cambac{\'e}d{\`e}s, Ludovic and Bouissou, Marc},
  booktitle={2010 European Dependable Computing Conference},
  pages={199--208},
  year={2010},
  organization={IEEE}
}

@inproceedings{jurgenson2008computing,
  title={Computing exact outcomes of multi-parameter attack trees},
  author={J{\"u}rgenson, Aivo and Willemson, Jan},
  booktitle={OTM Confederated International Conferences" On the Move to Meaningful Internet Systems"},
  pages={1036--1051},
  year={2008},
  organization={Springer}
}

@techreport{jointtaskforce2020,
  author      = {{Joint Task Force}},
  title       = {Security and Privacy Controls for Information Systems and Organizations},
  institution = {National Institute of Standards and Technology},
  number      = {NIST Special Publication 800-53, Revision 5},
  year        = {2020},
  doi         = {10.6028/NIST.SP.800-53r5}
}

@article{kordy2014attack,
  title={Attack--defense trees},
  author={Kordy, Barbara and Mauw, Sjouke and Radomirovi{\'c}, Sa{\v{s}}a and Schweitzer, Patrick},
  journal={Journal of logic and computation},
  volume={24},
  number={1},
  pages={55--87},
  year={2014},
  publisher={Oxford University Press}
}

@inproceedings{jha2002two,
  title={Two formal analyses of attack graphs},
  author={Jha, Somesh and Sheyner, Oleg and Wing, Jeannette},
  booktitle={Proceedings 15th IEEE Computer Security Foundations Workshop. CSFW-15},
  pages={49--63},
  year={2002},
  organization={IEEE}
}

@article{durkota2015game,
  title={Game-theoretic algorithms for optimal network security hardening using attack graphs},
  author={Durkota, Karel and Lis{\`y}, Viliam and Kiekintveld, Christopher and Bo{\v{s}}ansk{\`y}, Branislav},
  journal={Database},
  volume={20},
  pages={4xPC},
  year={2015}
}

@misc{formalmethods2026ffort,
  author       = {{Formal Methods and Tools Group, University of Twente}},
  title        = {{FFORT}: The Extended Fault Tree Forest},
  year         = {2026},
  howpublished = {\url{https://dftbenchmarks.utwente.nl/}},
  note         = {Accessed 2026-07-24}
}

@inproceedings{pita2008armor,
  title={ARMOR Security for Los Angeles International Airport.},
  author={Pita, James and Jain, Manish and Ord{\'o}nez, Fernando and Portway, Christopher and Tambe, Milind and Western, Craig and Paruchuri, Praveen and Kraus, Sarit},
  booktitle={AAAI},
  pages={1884--1885},
  year={2008}
}

@inproceedings{basilico2009leader,
  title={Leader-follower strategies for robotic patrolling in environments with arbitrary topologies},
  author={Basilico, Nicola and Gatti, Nicola and Amigoni, Francesco and others},
  booktitle={Proceedings of the International Joint Conference on Autonomous Agents and Multi Agent Systems (AAMAS)},
  pages={57--64},
  year={2009}
}

@inproceedings{lisy2016counterfactual,
  title={Counterfactual regret minimization in sequential security games},
  author={Lisy, Viliam and Davis, Trevor and Bowling, Michael},
  booktitle={Proceedings of the AAAI conference on artificial intelligence},
  volume={30},
  number={1},
  year={2016}
}

@inproceedings{sinha2018stackelberg,
  title={Stackelberg security games: Looking beyond a decade of success},
  author={Sinha, Arunesh and Fang, Fei and An, Bo and Kiekintveld, Christopher and Tambe, Milind},
  year={2018},
  organization={IJCAI}
}

@inproceedings{phillips1998graph,
  title={A graph-based system for network-vulnerability analysis},
  author={Phillips, Cynthia and Swiler, Laura Painton},
  booktitle={Proceedings of the 1998 workshop on New security paradigms},
  pages={71--79},
  year={1998}
}

@inproceedings{sheyner2002automated,
  title={Automated generation and analysis of attack graphs},
  author={Sheyner, Oleg and Haines, Joshua and Jha, Somesh and Lippmann, Richard and Wing, Jeannette M},
  booktitle={Proceedings 2002 IEEE Symposium on Security and Privacy},
  pages={273--284},
  year={2002},
  organization={IEEE}
}

@inproceedings{ou2006scalable,
  title={A scalable approach to attack graph generation},
  author={Ou, Xinming and Boyer, Wayne F and McQueen, Miles A},
  booktitle={Proceedings of the 13th ACM conference on Computer and communications security},
  pages={336--345},
  year={2006}
}

@article{lallie2020review,
  title={A review of attack graph and attack tree visual syntax in cyber security},
  author={Lallie, Harjinder Singh and Debattista, Kurt and Bal, Jay},
  journal={Computer Science Review},
  volume={35},
  pages={100219},
  year={2020},
  publisher={Elsevier}
}

@inproceedings{milani2020harnessing,
  title={Harnessing the power of deception in attack graph-based security games},
  author={Milani, Stephanie and Shen, Weiran and Chan, Kevin S and Venkatesan, Sridhar and Leslie, Nandi O and Kamhoua, Charles and Fang, Fei},
  booktitle={International conference on decision and game theory for security},
  pages={147--167},
  year={2020},
  organization={Springer}
}

@inproceedings{guo2022practical,
  title={Practical fixed-parameter algorithms for defending active directory style attack graphs},
  author={Guo, Mingyu and Li, Jialiang and Neumann, Aneta and Neumann, Frank and Nguyen, Hung},
  booktitle={Proceedings of the AAAI Conference on Artificial Intelligence},
  volume={36},
  number={9},
  pages={9360--9367},
  year={2022}
}

@article{gittins1979bandit,
  author  = {Gittins, John C.},
  title   = {Bandit Processes and Dynamic Allocation Indices},
  journal = {Journal of the Royal Statistical Society:
             Series B (Methodological)},
  volume  = {41},
  number  = {2},
  pages   = {148--164},
  year    = {1979},
  doi     = {10.1111/j.2517-6161.1979.tb01068.x}
}

@book{gittins2011multiarmed,
  author    = {Gittins, John C. and Glazebrook, Kevin D.
               and Weber, Richard R.},
  title     = {Multi-Armed Bandit Allocation Indices},
  edition   = {2},
  publisher = {John Wiley \& Sons},
  address   = {Chichester, UK},
  year      = {2011},
  isbn      = {9780470670026},
  doi       = {10.1002/9780470980033}
}

@article{bertsimas1996conservation,
  title={Conservation laws, extended polymatroids and multiarmed bandit problems; a polyhedral approach to indexable systems},
  author={Bertsimas, Dimitris and Nino-Mora, Jos{\'e}},
  journal={Mathematics of Operations Research},
  volume={21},
  number={2},
  pages={257--306},
  year={1996},
  publisher={INFORMS}
}

@article{nino20072,
  title={A (2/3) n 3 fast-pivoting algorithm for the Gittins index and optimal stopping of a Markov chain},
  author={Ni{\~n}o-Mora, Jos{\'e}},
  journal={INFORMS Journal on Computing},
  volume={19},
  number={4},
  pages={596--606},
  year={2007},
  publisher={INFORMS}
}

@article{whittle1988restless,
  title={Restless bandits: Activity allocation in a changing world},
  author={Whittle, Peter},
  journal={Journal of applied probability},
  volume={25},
  number={A},
  pages={287--298},
  year={1988},
  publisher={Cambridge University Press}
}

@article{weiss1988branching,
  title={Branching bandit processes},
  author={Weiss, Gideon},
  journal={Probability in the Engineering and Informational Sciences},
  volume={2},
  number={3},
  pages={269--278},
  year={1988},
  publisher={Cambridge University Press}
}

@article{nemirovskij1983problem,
  title={Problem complexity and method efficiency in optimization},
  author={Nemirovskij, Arkadij Semenovi{\v{c}} and Yudin, David Borisovich},
  year={1983},
  publisher={Wiley-Interscience}
}

\onecolumn
\appendix
\begingroup
\onecolumn 


\fontsize{11pt}{13.6pt}\selectfont

\let\normalsize\undefined
\newcommand{\normalsize}{\fontsize{11pt}{13.6pt}\selectfont}

\setcounter{secnumdepth}{2} 

\usetikzlibrary{backgrounds,shapes.geometric,positioning,fit, arrows.meta}




\theoremstyle{plain}
\newtheorem{property}{Property}
\newtheorem{appendixlemma}{Lemma}[subsection]
\newtheorem{appendixcor}{Corollary}[subsection]
\newtheorem{appendixprop}{Proposition}[subsection]
\newtheorem{appendixproperty}{Property}[subsection]
\newtheorem{appendixtheorem}{Theorem}[subsection]
\newtheorem{appendixdef}{Definition}[subsection]

\newcommand{\restatedlemmaname}{}
\newtheorem*{restatedlemma}{\restatedlemmaname}

\newcommand{\restatedtheoremname}{}
\newtheorem*{restatedtheorem}{\restatedtheoremname}

\newcommand{\restatedpropositionname}{}
\newtheorem*{restatedproposition}{\restatedpropositionname}

\newenvironment{restatelemma}[1]
  {\renewcommand{\restatedlemmaname}{Lemma~#1}%
   \begin{restatedlemma}}
  {\end{restatedlemma}}

\newenvironment{restatetheorem}[1]
  {\renewcommand{\restatedtheoremname}{Theorem~#1}%
   \begin{restatedtheorem}}
  {\end{restatedtheorem}}

\newenvironment{restateproposition}[1]
  {\renewcommand{\restatedpropositionname}{Proposition~#1}%
   \begin{restatedproposition}}
  {\end{restatedproposition}}

\tikzset{
  vtx/.style   = {circle, draw, line width=0.4pt, fill=white,
                  inner sep=0pt, minimum size=4.5pt},
  edg/.style   = {-{Stealth[length=4pt,width=3pt]}, line width=0.5pt,
                  shorten >=1pt, shorten <=1pt},
  elab/.style  = {font=\scriptsize, inner sep=1.5pt},
  vlab/.style  = {font=\scriptsize, inner sep=2pt},
  tnode/.style = {draw, rounded corners=2pt, line width=0.4pt, inner sep=2pt,
                  minimum height=4.2mm, minimum width=8.5mm, font=\scriptsize},
  lnode/.style = {tnode, minimum width=5.5mm},
  tedg/.style  = {line width=0.4pt},
  steplab/.style = {font=\scriptsize, anchor=east},
  bbox/.style  = {draw, dashed, line width=0.4pt, rounded corners=1.5pt},
  blk/.style   = {draw, rounded corners=1.5pt, line width=0.4pt, fill=white,
                  inner sep=1pt, minimum width=1cm, minimum height=4.4mm,
                  font=\scriptsize},
}

\title{Appendix}

\setcounter{assumption}{0}
\setcounter{theorem}{0}
\setcounter{proposition}{0}
\setcounter{corollary}{0}
\setcounter{remark}{0}

\centerline{\textbf{Appendix}}

\paragraph{Numbering convention.}
Theorems, lemmas, propositions, and corollaries introduced only in the appendix are numbered with an appendix-letter prefix, such as Theorem~A.1. Results restated from the main paper retain their original unprefixed numbering. Thus, a reference to Theorem~1 in the appendix refers to the same Theorem~1 stated in the main paper.

\paragraph{Notations} Note that we use $\Ctrl(e)$ and control $e$ interchangeably. In the main paper, we used $\gamma$ as the input parameter of $\Psh_G$ but for the appendix, we use $s \in[0, 1]$ in place of $\gamma$ for ease of reading.

\section{Regret Matching as a First Order Method to solve for the optimal $\ell$}
One of the technical contribution of this paper relates to efficiently computing a subgradient $g(\ell) \in \partial V_\mathcal{N}^*(\ell)$. This subgradient is then used to solve \eqref{eq:minmax} using first order methods (FOM), as part of the ``inner loop''. Such FOM include projected subgradient descent and mirror descent~\citep{nemirovskij1983problem}, as well as various online learning methods. We do not assume smoothness of $V_\mathcal{N}^*(\ell)$; indeed, this is not true in general. 

In the following, we recall the regret matching algorithm used \citep{hart2013simple}. We have included this for completeness. It is essentially the same algorithm used as \citet{ling2026security} and is fairly standard, so we are \textit{not} claiming any novelty here. 

\begin{algorithm}
\caption{Regret Matching for Optimizing $\ell$}
\label{alg:rm}
\begin{algorithmic}[1]
    \State $\ell^{(1)} \gets \mathbf{1}/|\mathcal{E}|$
        \Comment{Initialize lengths uniformly}
    \State $y^{(1)} \gets \mathbf{0}$
        \Comment{Initialize regrets}

    \For{$T = 1,\dots,T_{\max}$}
        \State $\pi^* \gets
            \Call{AttackerOptimalStrategy}{\ell^{(T)}}$
        \State $g^{(T)} \gets
            \Call{Subgradient}{\ell^{(T)},\pi^*}$
        \State $y^{(T+1)} \gets y^{(T)}
            + \mathbf{1}\langle g^{(T)},\ell^{(T)}\rangle
            - g^{(T)}$
            \Comment{Update regrets}

        \If{$y^{(T+1)}$ is not all $\leq \mathbf{0}$}
            \State $\displaystyle
                \ell^{(T+1)}
                \gets
                \frac{\max\{y^{(T+1)},\mathbf{0}\}}
                {\sum_{e\in\mathcal{E}}
                \max\{y^{(T+1)}(e),0\}}$
                \Comment{Normalize positive regrets}
        \Else
            \State $\ell^{(T+1)} \gets \mathbf{1}/|\mathcal{E}|$
                \Comment{Use the uniform allocation}
        \EndIf
    \EndFor

    \State \Return $\displaystyle
        \frac{1}{T_{\max}}
        \sum_{T=1}^{T_{\max}}\ell^{(T)}$
        \Comment{Return the average strategy}
\end{algorithmic}
\end{algorithm}

RM (Algorithm~\ref{alg:rm}) is guaranteed to converge to the optimal $\ell$ at a rate of $\mathcal{O}(1/\sqrt{T})$. We use it mainly because of its simplicity and attractive quality of being parameter free (i.e., no learning rate/schedule to be tuned), even though it may not have the optimal dependence on the number of controls. Since our main contribution is the finding of $\pi^*$ and $g^{(T)}$, we have simply used vanilla RM. There are many variants of RM (e.g., RM+, predictive RM, linear and discounted methods) that could be employed as well. Finally, note that RM continues to work with stochastic estimates for gradients, hence it is an appropriate method for the baseline of stochastic gradients used in our experiments.

\section{Supplementary Material on the SP Parse Tree $\mathcal{T}$}

\subsection{Converting an SP network to its parse tree}
An SP network is by construction the result of repeatedly composing single controls with $\Ser$ and $\Par$, so its parse tree is nothing but the record of that construction. Recovering the tree from a drawing, therefore, amounts to running the construction backwards: repeatedly \textit{reduce} a two-edge pattern to a single edge, and emit the corresponding tree node. \Cref{fig:rules} gives the two patterns.
 
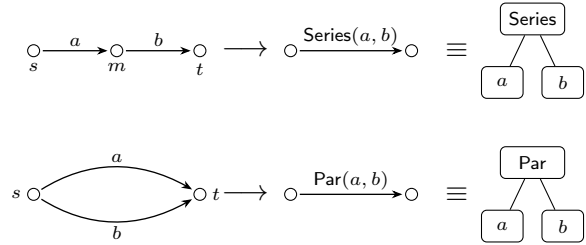
\begin{figure}[h!]
\centering
\begin{tikzpicture}
\node[vtx,label={[vlab]below:$s$}] (s1) at (0,0) {};
\node[vtx,label={[vlab]below:$m$}] (m1) at (1.1,0) {};
\node[vtx,label={[vlab]below:$t$}] (t1) at (2.2,0) {};
\draw[edg] (s1) -- node[elab,above] {$a$} (m1);
\draw[edg] (m1) -- node[elab,above] {$b$} (t1);
\node at (2.8,0) {$\longrightarrow$};
\node[vtx] (s1r) at (3.4,0) {};
\node[vtx] (t1r) at (5.0,0) {};
\draw[edg] (s1r) -- node[elab,above] {$\Ser(a,b)$} (t1r);
\node at (5.6,0) {$\equiv$};
\node[tnode] (r1)  at (6.6, 0.42) {$\Ser$};
\node[lnode] (r1a) at (6.2,-0.42) {$a$};
\node[lnode] (r1b) at (7.0,-0.42) {$b$};
\draw[tedg] (r1) -- (r1a);
\draw[tedg] (r1) -- (r1b);
\node[vtx,label={[vlab]left:$s$}]  (s2) at (0,-1.9) {};
\node[vtx,label={[vlab]right:$t$}] (t2) at (2.2,-1.9) {};
\draw[edg] (s2) to[bend left=32]  node[elab,above] {$a$} (t2);
\draw[edg] (s2) to[bend right=32] node[elab,below] {$b$} (t2);
\node at (2.8,-1.9) {$\longrightarrow$};
\node[vtx] (s2r) at (3.4,-1.9) {};
\node[vtx] (t2r) at (5.0,-1.9) {};
\draw[edg] (s2r) -- node[elab,above] {$\Par(a,b)$} (t2r);
\node at (5.6,-1.9) {$\equiv$};
\node[tnode] (r2)  at (6.6,-1.48) {$\Par$};
\node[lnode] (r2a) at (6.2,-2.32) {$a$};
\node[lnode] (r2b) at (7.0,-2.32) {$b$};
\draw[tedg] (r2) -- (r2a);
\draw[tedg] (r2) -- (r2b);
\end{tikzpicture}
\caption{The two reduction rules, each shown as network pattern, reduced edge, and emitted parse-tree node. The $\Ser$ rule requires the shared vertex $m$ to have exactly one incoming and one outgoing edge and to be neither terminal of the whole network; the $\Par$ rule requires the two edges to agree on both endpoints.}
\label{fig:rules}
\end{figure}
 
Either rule may be applied wherever it fits, in any order, and the network is SP exactly when repeated application leaves a single edge. \Cref{fig:reduce} carries out the reduction on a five-control example. Note that $c$ cannot be reduced at Step~1 as it shares its endpoints with the \textit{path} $a,b$ rather than with a single edge, so the $\Ser$ reduction on $a,b$ must come first.
 
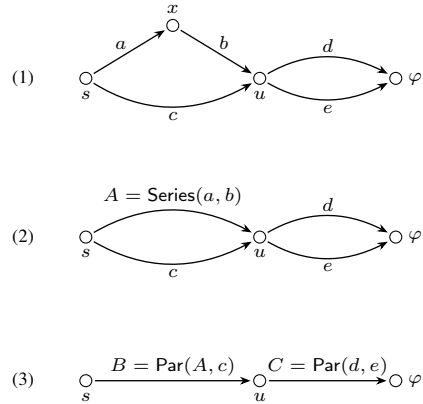
\begin{figure}[h!]
\centering
\begin{tikzpicture}
\node[steplab] at (-0.55,0) {(1)};
\node[vtx,label={[vlab]below:$s$}]  (s)  at (0,0)      {};
\node[vtx,label={[vlab]above:$x$}]  (x)  at (1.15,0.7) {};
\node[vtx,label={[vlab]below:$u$}]  (u)  at (2.3,0)    {};
\node[vtx,label={[vlab]right:$\varphi$}] (p) at (4.1,0) {};
\draw[edg] (s) -- node[elab,above left=-1pt]  {$a$} (x);
\draw[edg] (x) -- node[elab,above right=-1pt] {$b$} (u);
\draw[edg] (s) to[bend right=32] node[elab,below] {$c$} (u);
\draw[edg] (u) to[bend left=30]  node[elab,above] {$d$} (p);
\draw[edg] (u) to[bend right=30] node[elab,below] {$e$} (p);
\node[steplab] at (-0.55,-2.1) {(2)};
\node[vtx,label={[vlab]below:$s$}]  (s2) at (0,-2.1)   {};
\node[vtx,label={[vlab]below:$u$}]  (u2) at (2.3,-2.1) {};
\node[vtx,label={[vlab]right:$\varphi$}] (p2) at (4.1,-2.1) {};
\draw[edg] (s2) to[bend left=30]  node[elab,above] {$A=\Ser(a,b)$} (u2);
\draw[edg] (s2) to[bend right=30] node[elab,below] {$c$} (u2);
\draw[edg] (u2) to[bend left=30]  node[elab,above] {$d$} (p2);
\draw[edg] (u2) to[bend right=30] node[elab,below] {$e$} (p2);
\node[steplab] at (-0.55,-4.0) {(3)};
\node[vtx,label={[vlab]below:$s$}]  (s3) at (0,-4.0)   {};
\node[vtx,label={[vlab]below:$u$}]  (u3) at (2.3,-4.0) {};
\node[vtx,label={[vlab]right:$\varphi$}] (p3) at (4.1,-4.0) {};
\draw[edg] (s3) -- node[elab,above] {$B=\Par(A,c)$} (u3);
\draw[edg] (u3) -- node[elab,above] {$C=\Par(d,e)$} (p3);
\end{tikzpicture}
\caption{Reducing a five-control network. One further $\Ser$ reduction collapses
step~(3) to a single edge $\Ser(B,C)$, which is the root of the parse tree in
\Cref{fig:tree}.}
\label{fig:reduce}
\end{figure}
 
Reading the reductions from the inside out gives \Cref{fig:tree}: the first
reductions sit near the leaves and the last is the root.
 
\begin{figure}[h!]
\centering
\begin{tikzpicture}
\node[tnode] (root) at (2.325, 0.0) {$\Ser$};
\node[tnode] (pl)   at (1.100,-0.9) {$\Par$};
\node[tnode] (pr)   at (3.550,-0.9) {$\Par$};
\node[tnode] (sl)   at (0.450,-1.8) {$\Ser$};
\node[lnode] (c)    at (1.750,-1.8) {$c$};
\node[lnode] (d)    at (3.100,-1.8) {$d$};
\node[lnode] (e)    at (4.000,-1.8) {$e$};
\node[lnode] (a)    at (0.000,-2.7) {$a$};
\node[lnode] (b)    at (0.900,-2.7) {$b$};
\draw[tedg] (root) -- (pl);
\draw[tedg] (root) -- (pr);
\draw[tedg] (pl)   -- (sl);
\draw[tedg] (pl)   -- (c);
\draw[tedg] (pr)   -- (d);
\draw[tedg] (pr)   -- (e);
\draw[tedg] (sl)   -- (a);
\draw[tedg] (sl)   -- (b);
\end{tikzpicture}
\caption{Parse tree of the network of \Cref{fig:reduce}, namely
$\Ser\bigl(\Par(\Ser(a,b),c),\,\Par(d,e)\bigr)$. Left-to-right order is meaningful
under $\Ser$ and immaterial under $\Par$.}
\label{fig:tree}
\end{figure}
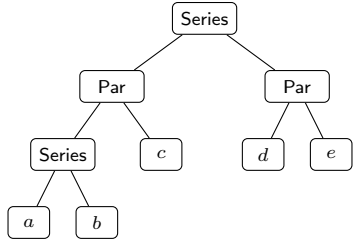
 
The reduction order is immaterial. Different orders differ only by reassociating chains of the same type of nodes and permuting the children of a $\Par$ node, which is exactly the freedom the semantics permits (since $\Par$ is commutative). For a $\Ser$ node, it is different, and the ordering (left to right) matters because $\Ser$ is associative but not commutative. Flattening every maximal same-type chain into one $m$-ary node yields the canonical decomposition tree, computable in linear time. Reading it off is easy. Descending from the root, take \textit{only the first uncleared child} at a $\Ser$ node and \textit{all live children} at a $\Par$ node; the leaves reached are exactly the controls on the current frontier (denoted as $\mathcal{F}$ in the main paper). This is why $\Ser$ must retain its child order while $\Par$ need not.
 
Finally, \Cref{fig:bridge} shows an example of a directed acyclic graph with single source and sink that is \textit{not} a SP network.
 
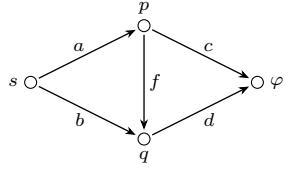
\begin{figure}[h!]
\centering
\begin{tikzpicture}
\node[vtx,label={[vlab]left:$s$}]        (s) at (0,0)     {};
\node[vtx,label={[vlab]above:$p$}]       (p) at (1.5,0.75){};
\node[vtx,label={[vlab]below:$q$}]       (q) at (1.5,-0.75){};
\node[vtx,label={[vlab]right:$\varphi$}] (t) at (3.0,0)   {};
\draw[edg] (s) -- node[elab,above left=-1pt]  {$a$} (p);
\draw[edg] (s) -- node[elab,below left=-1pt]  {$b$} (q);
\draw[edg] (p) -- node[elab,above right=-1pt] {$c$} (t);
\draw[edg] (q) -- node[elab,below right=-1pt] {$d$} (t);
\draw[edg] (p) -- node[elab,right]            {$f$} (q);
\end{tikzpicture}
\caption{A DAG that cannot be written as a SP-network.}
\label{fig:bridge}
\end{figure}

\section{Omitted Proofs for Calibrated Value Profile}
\label{app:calibrated_value_profile}
\subsection{Derivation of the Calibrated Value Profile}
Recall the calibrated value profile for the subcomponent \(G\), denoted as \(\Ps_G\)
\begin{appendixdef}[Calibrated value profile]\label{a-def:psi}
For a subcomponent $G$, state $\sgn$, exit reward $r$ and buyout $\gamma\in[0,r]$,
\begin{equation}\label{a-eq:psi}
\Psi_G(\gamma;\sgn\mid r)\ \eqdef \sup_\pi\ \E_\sgn\!\Big[e^{-\lambda T}r\one\{T\le \tau_\pi\}\;+\;\gamma e^{-\lambda \tau_\pi}\one\{T>\tau_\pi\}\Big],
\end{equation}
the supremum over all adaptive policies.
\end{appendixdef}
\paragraph{Derivation} Augment the game with a standing buyout offer worth $\gamma$ (collectible at any time). A play ends at time $\min(T,\tau_\pi)$ with two mutually exclusive outcomes: (i) \textbf{Win}. That is, $T\le\tau_\pi$ and collect the discounted reward $e^{-\lambda T}r$ or (ii) \textbf{Quit}. That is where $T>\tau_\pi$ and collect $e^{-\lambda\tau_\pi}\gamma$ at $\tau_\pi$. Summing and optimising gives \Cref{a-eq:psi}. At $\gamma=0$, which is our true game situation, quitting is worthless and $\Ps_G(0;\sgn\mid r)$ is the true game value.

\subsection{Proof of Structural Properties of the Calibrated Value Profiles \(\Ps_G\)}
\begin{appendixproperty}[Affine representation]\label{a-lem:affine}
For every subcomponent $G$ and state $\sgn$, there is a family of pairs $\{(a_\pi,b_\pi)\}_\pi$, with $a_\pi\ge0$, $b_\pi\in[0,1]$ and $a_\pi+b_\pi\le1$, depending only on $\pi$ such that
\begin{equation}
    \Ps_G(\gamma;\sgn\mid r)=\sup_\pi\bigl(a_\pi r+b_\pi\gamma\bigr)
\end{equation} where $a_\pi=\E_\sgn[e^{-\lambda T}\one\{T\le\tau_\pi\}]$ and $b_\pi=\E_\sgn[e^{-\lambda\tau_\pi}\one\{T>\tau_\pi\}]$.
\end{appendixproperty}
\begin{proof}
$a_\pi r+b_\pi\gamma$ is precisely the expression in the bracket of \Cref{a-eq:psi}, with the stated $a_\pi,b_\pi$. Note that $a_\pi$ and $b_\pi$ are independent of both $r$ and $\gamma$. $a_\pi$ and $b_\pi$ are clearly non-negative. Note that the events $\{T\le\tau_\pi\},\{T>\tau_\pi\}$ partition the sample space. Furthermore, we have that $a_\pi+b_\pi=\E_\sgn[e^{-\lambda\min(T,\tau_\pi)}]\le1$ since $\min(T,\tau_\pi)\ge0$. In particular $b_\pi\le1$.
\end{proof}

\begin{appendixproperty}[Shape]\label{a-lem:shape}
Fix $G,\sgn$. For $\gamma \in [0, r]$, the function $\Ps_G$ is convex and non-decreasing in $\gamma$, $1$--Lipschitz in $\gamma$, satisfies $\gamma\le\Ps_G(\gamma;\sgn\mid r)\le r$, with the boundary values $\Ps_G(\gamma;\sgn\mid\gamma)=\gamma$ and $\Ps_G(r;\sgn\mid r)=r$. Its right-derivative $g_G\eqdef\partial^+_\gamma\Ps_G$ exists everywhere, is non-decreasing and right-continuous, and lies in $[0,1]$.
\end{appendixproperty}
\begin{proof}
We know from Property~\ref{a-lem:affine} that the $\Ps_G$ is a pointwise supremum of functions affine in $\gamma$, hence convex. The slope of each said function is $b_\pi\in[0,1]$, and hence is non-decreasing and $1$-Lipschitz in $\gamma$. A finite convex function has right-derivatives everywhere, is non-decreasing and right-continuous, and for our case, valued in $[0,1]$. Suppose we have a policy $\pi\equiv$``retire immediately (don't start at all)''. This will result in $a_\pi = 0, b_\pi = 1$, resulting in $\Ps_G\ge\gamma$. Suppose that $r \geq \gamma$. Then $\Ps_G \leq \sup_\pi (a_\pi + b_\pi)\max(r, \gamma) = \sup_\pi(a_\pi + b_\pi)r \leq r$. For the case where \(r = \gamma\), then $\Ps_G(\gamma; \sigma \mid \gamma) = \sup_\pi (\gamma e^{-\lambda\min(T, \tau_\pi)})$. Since \(\tau_\pi = 0\), then \(\Ps_G(\gamma; \sigma \mid \gamma) = \gamma\). This means that \(\Ps_G\) is an identity at \(\gamma = r\).
\end{proof}

\begin{appendixproperty}[Piecewise linearity]
\label{a-lem:pwl}
Fix a subcomponent $G$, state $\sgn$, and exit reward $r>0$. The map $\gamma\mapsto\Ps_G(\gamma;\sgn\mid r)$ is piecewise linear in $\gamma$ on $[0,r]$.
\end{appendixproperty}

\begin{proof}
The attacker faces a finite, acyclic decision process. Each attempt either breaches a control or increases its failure count, and every failure count is bounded by its lockout limit. Hence, there are finitely many reachable states and finitely many deterministic Markov policies. By backward induction over the acyclic state-transition graph, an optimal policy may be chosen to be deterministic and Markov. For every such policy $\pi$, Property~\ref{a-lem:affine} gives
\begin{equation}
    \Ps_G^\pi(\gamma;\sgn\mid r)=a_\pi r+b_\pi\gamma.
\end{equation}
Therefore,
\begin{equation}
    \Ps_G(\gamma;\sgn\mid r)=\max_{\pi\in\Pi_G}\{a_\pi r+b_\pi\gamma\},
\end{equation}
where $\Pi_G$ is a finite set. The maximum of finitely many affine functions is piecewise linear.
\end{proof}

\begin{appendixproperty}[Homogeneity]\label{a-lem:homog}
$\Ps_G(\lambda\gamma;\sgn\mid\lambda r)=\lambda\Ps_G(\gamma;\sgn\mid r)$ for all $\lambda>0$. Hence, writing $\Psh_G(\gamma)\eqdef\Ps_G(\gamma;\sgn\mid1)$, one has $\Ps_G(\gamma;\sgn\mid r)=r\Psh_G(\gamma/r)$ for $r>0$.
\end{appendixproperty}
\begin{proof}
By Property \ref{a-lem:affine}, scaling $(\gamma,r)\mapsto(\lambda\gamma,\lambda r)$ scales each affine term $a_\pi r+b_\pi\gamma$ by a positive $\lambda$, hence the statement holds. The second statement is achieved by $\lambda=1/r$. We can safely assume $r > 0$ as for $r \leq 0$, the attacker can gain no extra reward. Equivalently, if the sub-component offers a downstream reward of $r \leq 0$, the sub-component can be removed from the SP network.
\end{proof}

\begin{appendixcor}[Retirement threshold]
\label{a-cor:retirement-threshold}
For every subcomponent $G$, reachable state $\sgn$, and exit reward $r>0$, define
\begin{equation}
\alpha_G(\sgn\mid r)\eqdef\inf\{\gamma\in[0,r]:\Ps_G(\gamma;\sgn\mid r)=\gamma\}.
\end{equation}
Then continuing is strictly preferable when $\gamma<\alpha_G(\sgn\mid r)$, whereas retirement is optimal when $\gamma\ge\alpha_G(\sgn\mid r)$. Moreover,
\begin{equation}
\alpha_G(\sgn\mid r)=r\alpha_G(\sgn\mid1).
\end{equation}
\end{appendixcor}

\begin{proof}
Define $H(\gamma)\eqdef\Ps_G(\gamma;\sgn\mid r)-\gamma$. By Property~\ref{a-lem:shape}, $H(\gamma)\ge0$ and $H(r)=0$. Moreover, for $0\le\gamma_1<\gamma_2\le r$, the $1$-Lipschitz property gives
\begin{equation}
H(\gamma_2)-H(\gamma_1)=\Ps_G(\gamma_2;\sgn\mid r)-\Ps_G(\gamma_1;\sgn\mid r)-(\gamma_2-\gamma_1)\le0.
\end{equation}
Thus, $H$ is non-increasing, and its zero set is a nonempty interval of the form $[\alpha_G(\sgn\mid r),r]$. Below this interval, $H(\gamma)>0$, so continuing is strictly preferable. On the interval, $\Ps_G(\gamma;\sgn\mid r)=\gamma$, so immediate retirement is optimal.

Finally, Property~\ref{a-lem:homog} gives
\begin{equation}
\Ps_G(\gamma;\sgn\mid r)=\gamma\quad\Longleftrightarrow\quad\Psh_G(\gamma/r)=\gamma/r.
\end{equation}
Taking the smallest fixed point on both sides yields $\alpha_G(\sgn\mid r)=r\alpha_G(\sgn\mid1)$.
\end{proof}

\section{Structural Results for SP Subcomponents}
\subsection{Indices of Individual Controls}
\label{app:individual-control-indices}

We first consider an isolated control whose successful compromise yields a fixed exit reward. We then replace this scalar reward by the calibrated continuation profile of a downstream subcomponent.

\paragraph{Zero-duration controls.}
The defender's feasible set permits $\ell_e=0$, in which case $\beta_e=e^{-\lambda\ell_e}=1$. The attacker problem remains well-defined because each control can be attempted at most $q_e$ times, so even a sequence of zero-duration attempts is finite. We define $c(e, k)$ as the following:
\begin{equation}
c(e,k)\eqdef
\begin{cases}
\dfrac{\beta_ep_e(k)}{1-\beta_e(1-p_e(k))}, & \beta_e<1,\\
1, & \beta_e=1\text{ and }p_e(k)>0,\\
0, & \beta_e=1\text{ and }p_e(k)=0.
\end{cases}
\end{equation}
The final two cases are the continuous extensions of the first expression as $\beta_e\to1$. When $c(e,k)<1$, the equation $\gamma=c(e,k)R_e(\gamma)$ has a unique solution. When $c(e,k)=1$, the fixed-point set may contain an interval; throughout, $\alpha(e,k)$ denotes its smallest fixed point, which is the retirement threshold. Thus, all index formulas and policy results include controls with $\ell_e=0$.

\begin{appendixlemma}[Isolated-control index]
\label{lem:isolated-control-index}
Fix a control $e$, failure count $k<q_e$, and exit reward $r\in[0,1]$. Its calibrated value satisfies
\begin{equation}
\Psi_e(\gamma;k\mid r)=\max\left\{\gamma,\beta_e\left[p_e(k)r+\bigl(1-p_e(k)\bigr)\Psi_e(\gamma;k+1\mid r)\right]\right\},
\end{equation}
with terminal condition $\Psi_e(\gamma;q_e\mid r)=\gamma$. Under the non-increasing probabilities assumption, its retirement threshold is
\begin{equation}
\alpha(e,k\mid r)=c(e,k)r
\end{equation}
Continuing is strictly optimal when $\gamma<\alpha(e,k\mid r)$, while retirement is optimal when $\gamma\ge\alpha(e,k\mid r)$. Moreover, $\alpha(e,k\mid r)$ is non-increasing in $k$.
\end{appendixlemma}

\begin{proof}
Write $p_k\eqdef p_e(k)$, $\beta\eqdef \beta_e$, $c_k\eqdef c(e,k)$, and $\alpha_k\eqdef c_kr$. The map $p\mapsto\beta p/[1-\beta(1-p)]$ is increasing on $[0,1]$. Since $p_k$ is non-increasing in $k$, both $c_k$ and $\alpha_k$ are non-increasing in $k$.

We prove the threshold statement by backward induction. At $k=q_e$, no further attempt is available, so $\Psi_e(\gamma;q_e\mid r)=\gamma$.
Suppose the statement holds at $k+1$. If $\gamma\ge\alpha_k$, then $\gamma\ge\alpha_{k+1}$, and therefore $\Psi_e(\gamma;k+1\mid r)=\gamma$. The value of attempting $e$ is \(\beta\left[p_kr+(1-p_k)\gamma\right]\). Rearranging gives:
\begin{equation}
\beta\left[p_kr+(1-p_k)\gamma\right]\le\gamma\quad\Longleftrightarrow\quad \gamma\ge c_kr=\alpha_k.
\end{equation}
Thus, retirement is optimal whenever $\gamma\ge\alpha_k$.

Now suppose $\gamma<\alpha_k$. If $\gamma\ge\alpha_{k+1}$, then the continuation value after a failure is $\gamma$, and the preceding equivalence shows that attempting $e$ has value strictly greater than $\gamma$. If $\gamma<\alpha_{k+1}$, then the induction hypothesis gives $\Psi_e(\gamma;k+1\mid r)>\gamma$, and hence:
\begin{equation}
\beta\left[p_kr+(1-p_k)\Psi_e(\gamma;k+1\mid r)\right]\ge\beta\left[p_kr+(1-p_k)\gamma\right]>\gamma.
\end{equation}
Thus, continuing is strictly optimal whenever $\gamma<\alpha_k$.
\end{proof}

Having resolved the isolated control case, let us now proceed to the case of an intermediate control. Note that an intermediate control does not yield a fixed exit reward. Instead, its successful compromise \textbf{exposes a downstream continuation} whose value \textit{depends on the buyout}.

\begin{appendixlemma}[Intermediate-control index]
\label{lem:intermediate-control-index}
Fix a control $e$ at failure count $k<q_e$, and let $R_e:[0,1]\to[0,1]$ denote the calibrated continuation profile exposed upon successfully compromising $e$. Suppose that $R_e$ is continuous, non-decreasing, and $1$-Lipschitz. The calibrated value of $e$ satisfies
\begin{equation}
\Psi_e(\gamma;k \mid R_e)=\max\left\{\gamma,\beta_e\left[p_e(k)R_e(\gamma)+\bigl(1-p_e(k)\bigr)\Psi_e(\gamma;k+1 \mid R_e)\right]\right\},
\end{equation}
with $\Psi_e(\gamma;q_e \mid R_e)=\gamma$. Its index $\alpha(e,k)$ is the smallest solution in $[0,1]$ of
\begin{equation}
\alpha(e,k)=c(e,k)R_e\bigl(\alpha(e,k)\bigr).
\end{equation}
Continuing is strictly optimal when $\gamma<\alpha(e,k)$, while retirement is optimal when $\gamma\ge\alpha(e,k)$. Moreover, $\alpha(e,k)$ is non-increasing in $k$.
\end{appendixlemma}

\begin{proof}
Write $c_k\eqdef c(e,k)$ and define $h_k(\gamma)\eqdef c_kR_e(\gamma)-\gamma$. We first consider $c_k<1$. For $0\le x<y\le1$, the $1$-Lipschitz property of $R_e$ gives
\begin{equation}
h_k(y)-h_k(x)=c_k\bigl(R_e(y)-R_e(x)\bigr)-(y-x)\le(c_k-1)(y-x)<0.
\end{equation}
Thus, $h_k$ is strictly decreasing and has a unique zero.

Now suppose $c_k=1$, which occurs when $\ell_e=0$ and $p_e(k)>0$. Then $h_k(\gamma)=R_e(\gamma)-\gamma$ is non-increasing, and its zero set is the retirement region of the continuation profile. We define $\alpha_k$ to be the smallest zero of $h_k$. For $\gamma<\alpha_k$, one has $R_e(\gamma)>\gamma$, and an immediate attempt is strictly preferable because $p_e(k)>0$. For $\gamma\ge\alpha_k$, both successful completion and retirement have value $\gamma$, and the backward-induction argument gives retirement as an optimal action. Finally, if $c_k=0$, then $p_e(k)=0$ and, by the non-increasing-probabilities assumption, all subsequent success probabilities are also zero and hence $\alpha_k=0$.

Because $p_e(k)$ is non-increasing, $c_{k+1}\le c_k$. Hence $h_{k+1}(\gamma)\le h_k(\gamma)$ for every $\gamma$, which implies $\alpha_{k+1}\le\alpha_k$.

We prove the threshold statement by backward induction. At $k=q_e$, no attempt remains and $\Psi_e(\gamma;q_e \mid R_e)=\gamma$. Suppose the statement holds at $k+1$. If $\gamma\ge\alpha_k$, then $\gamma\ge\alpha_{k+1}$, so $\Psi_e(\gamma;k+1 \mid R_e)=\gamma$. The value of attempting $e$ is
\begin{equation}
\beta_e\left[p_e(k)R_e(\gamma)+\bigl(1-p_e(k)\bigr)\gamma\right].
\end{equation}
Rearranging gives
\begin{equation}
\beta_e\left[p_e(k)R_e(\gamma)+\bigl(1-p_e(k)\bigr)\gamma\right]\le\gamma\quad\Longleftrightarrow\quad c_kR_e(\gamma)\le\gamma.
\end{equation}
Since $h_k$ is strictly decreasing and vanishes at $\alpha_k$, this inequality holds exactly when $\gamma\ge\alpha_k$.

If $\gamma<\alpha_k$ and $\gamma\ge\alpha_{k+1}$, then the same equivalence shows that attempting $e$ has value strictly greater than $\gamma$. If $\gamma<\alpha_{k+1}$, then $\Psi_e(\gamma;k+1 \mid R_e)>\gamma$, and therefore
\begin{equation}
\beta_e\left[p_e(k)R_e(\gamma)+\bigl(1-p_e(k)\bigr)\Psi_e(\gamma;k+1\mid R_e)\right]\ge\beta_e\left[p_e(k)R_e(\gamma)+\bigl(1-p_e(k)\bigr)\gamma\right]>\gamma.
\end{equation}
This completes the induction.
\end{proof}

\begin{appendixcor}[Index on an affine continuation piece]
\label{cor:affine-continuation-index}
Suppose that $\alpha(e,k)$ lies on a piece of the continuation profile on which $R_e(\gamma)=A+B\gamma$. Then
\begin{equation}
\alpha(e,k)=\frac{c(e,k)A}{1-c(e,k)B}.
\end{equation}
\end{appendixcor}

\begin{proof}
Substituting $R_e(\alpha)=A+B\alpha$ into the fixed-point equation gives $\alpha=c(e,k)(A+B\alpha)$. Rearranging yields the result. The denominator is positive because $0\le B\le1$ and $c(e,k)<1$.
\end{proof}

\subsection{Compositional Rules}
\label{app:compositional-rules}

We next prove the rules used to combine intrinsic calibrated value profiles at $\Par$ and $\Ser$ nodes.

\begin{appendixlemma}[Threshold representation]
\label{lem:threshold-representation}
Let $f:[0,1]\to[0,1]$ be convex, non-decreasing, and $1$-Lipschitz, with $f(1)=1$. Then there exists an auxiliary random variable $X_f\in[0,1]$ such that
\begin{equation}
f(s)=\mathbb E[\max\{s,X_f\}]
\end{equation}
for every $s\in[0,1]$. Moreover, $X_f$ may be chosen so that $\Pr(X_f\le s)=\partial_s^+f(s)$ for $s\in[0,1)$.
\end{appendixlemma}

\begin{proof}
Because $f$ is convex and $1$-Lipschitz, its right derivative $\partial_s^+f$ is non-decreasing and takes values in $[0,1]$. Define a distribution function $F_f$ on $[0,1]$ by $F_f(s)\eqdef \partial_s^+f(s)$ for $s<1$ and $F_f(1)\eqdef 1$, and let $X_f$ have distribution function $F_f$.

For any $s\in[0,1]$, we have $(X_f-s)_+=\int_s^1\mathbf 1\{X_f>z\}\,dz$. Taking expectations and interchanging expectation and integration gives
\begin{equation}
\mathbb E[\max\{s,X_f\}]=s+\int_s^1\Pr(X_f>z)\,dz=1-\int_s^1F_f(z)\,dz.
\end{equation}
Since a convex function on a compact interval is absolutely continuous,
\begin{equation}
f(1)-f(s)=\int_s^1 \partial_z^+f(z)\,dz.
\end{equation}
Using $f(1)=1$ and $F_f=\partial_z^+f$ gives
\begin{equation}
f(s)=1-\int_s^1F_f(z)\,dz=\mathbb E[\max\{s,X_f\}].
\end{equation}
\end{proof}

\begin{appendixlemma}[Parallel composition]
\label{lem:parallel-composition-proof}
Let $G=\Par(G_1,\cdots,G_m)$. Then
\begin{equation}
\widehat\Psi_G(\gamma)=1-\int_\gamma^1\prod_{i=1}^m\widehat g_{G_i}(z)\,dz.
\end{equation}
Equivalently,
\begin{equation}
\widehat g_G(\gamma)=\prod_{i=1}^m\widehat g_{G_i}(\gamma).
\end{equation}
\end{appendixlemma}

\begin{proof}
Apply Lemma~\ref{lem:threshold-representation} to each child profile. For every $i$, choose an auxiliary threshold variable $X_i$ satisfying
\begin{equation}
\widehat\Psi_{G_i}(\gamma)=\mathbb E[\max\{\gamma,X_i\}],\qquad \Pr(X_i\le z)=\widehat g_{G_i}(z).
\end{equation}
The $X_i$ may be taken to be independent because the child subcomponents have independent local dynamics. Under the superprocess representation used in the proof of Theorem~\ref{thm:index-policy}, a parallel subcomponent is worth continuing whenever at least one child is worth continuing. Its effective threshold is therefore $X_G\eqdef \max_iX_i$.

By independence,
\begin{equation}
\Pr(X_G\le z)=\Pr(X_1\le z,\cdots,X_m\le z)=\prod_{i=1}^m\widehat g_{G_i}(z).
\end{equation}
The threshold representation (a.k.a Lemma~\ref{lem:threshold-representation} now applied to $\Par$), then gives
\begin{equation}
\widehat\Psi_G(\gamma)=\gamma+\int_\gamma^1\Pr(X_G>z)\,dz=1-\int_\gamma^1\prod_{i=1}^m\widehat g_{G_i}(z)\,dz.
\end{equation}
Taking the right derivative with respect to $\gamma$ proves the slope identity.
\end{proof}

\begin{appendixlemma}[Series composition]
\label{lem:series-composition-proof}
Let $G=\Ser(G_1,\cdots,G_m)$. Define $\rho_{m+1}(\gamma)\eqdef 1$ and, recursively for $j=m,m-1,\cdots,1$,
\begin{equation}
\rho_j(\gamma)\eqdef \Psi_{G_j}\bigl(\gamma;\sigma_{G_j}^0\mid\rho_{j+1}(\gamma)\bigr).
\end{equation}
Then
\begin{equation}
\widehat\Psi_G(\gamma)=\rho_1(\gamma).
\end{equation}
More generally, $\rho_j(\gamma)$ is the calibrated continuation value upon entering $G_j$.
\end{appendixlemma}

\begin{proof}
For $j\in\{1,\cdots,m\}$, let $H_j\eqdef \Ser(G_j,\cdots,G_m)$ denote the suffix beginning at $G_j$, and let $H_{m+1}$ denote the breached terminal subcomponent. We prove by backward induction that
\begin{equation}
\rho_j(\gamma)=\Psi_{H_j}\bigl(\gamma;\sigma_{H_j}^0\mid1\bigr).
\end{equation}

For $j=m+1$, the suffix subcomponent is already breached, and hence
\begin{equation}
\Psi_{H_{m+1}}\bigl(\gamma;\sigma_{H_{m+1}}^0\mid1\bigr)=1=\rho_{m+1}(\gamma).
\end{equation}

Suppose the claim holds for $j+1$. Upon entering $H_j$, the attacker must first breach $G_j$. If the attacker retires before breaching $G_j$, it receives the buyout $\gamma$. If it breaches $G_j$, the remaining suffix is $H_{j+1}$ in its fresh initial state. By the induction hypothesis, the value of this continuation, measured from the time at which $G_j$ is breached, is $\rho_{j+1}(\gamma)$.

The controls in the suffix remain frozen while $G_j$ is attempted, and exponential discounting is multiplicative across the two stages. Therefore, from the perspective of $G_j$, successful breach yields the exit reward $\rho_{j+1}(\gamma)$. Hence
\begin{equation}
\Psi_{H_j}\bigl(\gamma;\sigma_{H_j}^0\mid1\bigr)=\Psi_{G_j}\bigl(\gamma;\sigma_{G_j}^0\mid\rho_{j+1}(\gamma)\bigr)=\rho_j(\gamma).
\end{equation}
This completes the induction. Taking $j=1$ and using $H_1=G$ gives $\widehat\Psi_G(\gamma)=\rho_1(\gamma)$.
\end{proof}

\begin{appendixcor}[Binary series rule]
\label{cor:binary-series-rule}
For $G=\Ser(G_1,G_2)$,
\begin{equation}
\widehat\Psi_G(\gamma)=\widehat\Psi_{G_2}(\gamma)\widehat\Psi_{G_1}\left(\frac{\gamma}{\widehat\Psi_{G_2}(\gamma)}\right).
\end{equation}
\end{appendixcor}

\begin{proof}
Lemma~\ref{lem:series-composition-proof} gives
\begin{equation}
\widehat\Psi_G(\gamma)=\Psi_{G_1}\bigl(\gamma;\sigma_{G_1}^0\mid\widehat\Psi_{G_2}(\gamma)\bigr).
\end{equation}
Applying Property~\ref{a-lem:homog} with $r=\widehat\Psi_{G_2}(\gamma)$ yields the result. If $\gamma=\widehat\Psi_{G_2}(\gamma)=0$, the expression is interpreted by its continuous extension.
\end{proof}

\subsection{Statewise Maximum-Index Property}
\label{app:statewise-index-property}

\begin{appendixprop}[Statewise maximum-index property]
\label{prop:statewise-max-index}
For every SP subcomponent $G$, reachable non-terminal state $\sigma$,
and exit reward $r>0$, define
\begin{equation}
\overline\alpha_G(\sigma)
\eqdef
\max_{e\in\mathcal F_G(\sigma)}\alpha(e,k_e).
\end{equation}
Then
\begin{equation}
\alpha_G(\sigma\mid r)=\overline\alpha_G(\sigma).
\end{equation}
Moreover, after fixing a deterministic tie-breaking rule, there exists
a control
\begin{equation}
e_G^*(\sigma)\in
\arg\max_{e\in\mathcal F_G(\sigma)}\alpha(e,k_e)
\end{equation}
that is optimal for every
$\gamma<\alpha_G(\sigma\mid r)$. Retirement is optimal for every
$\gamma\ge\alpha_G(\sigma\mid r)$.
\end{appendixprop}

The proof is organised into 3 cases: the $\Ctrl (e)$ case, the $\Ser$ and $\Par$ case and our method of proof is structural induction on the SP parse tree $\mathcal T$.

\begin{appendixlemma}[$\Ctrl$ (Leaf) base case]
\label{lem:statewise-leaf}
Proposition~\ref{prop:statewise-max-index} holds when
$G=\Ctrl(e)$.
\end{appendixlemma}

\begin{proof}
At a leaf, the frontier contains only $e$. By Lemma~\ref{lem:intermediate-control-index}, attempting $e$ is strictly optimal for $\gamma<\alpha(e,k_e)$, while retirement is optimal for $\gamma\ge\alpha(e,k_e)$. Hence $\alpha_G(\sigma\mid r)=\alpha(e,k_e)$. This is trivial.
\end{proof}

The induction hypothesis is understood in the following slightly stronger form. The exit reward may be replaced by any admissible continuation profile $R_e:[0,1]\to[0,1]$ arising from a downstream SP subcomponent. The index of a frontier control is then computed against the continuation profile obtained from its forward cone and $R_e$. Let
$R_e:[0,1]\to[0,1]$ be an admissible continuation profile. Define
\begin{equation}
\Psi_G(\gamma;\sigma\mid R_e) \eqdef  \sup_\pi \mathbb E_\sigma^\pi\!\left[ e^{-\lambda T}R_e(\gamma)\mathbf 1\{T\le\tau_\pi\} + e^{-\lambda\tau_\pi}\gamma\mathbf 1\{\tau_\pi<T\} \right].
\end{equation}
Thus, if $G$ is breached, the attacker receives the continuation value $R_e(\gamma)$, measured from the completion time of $G$. A scalar exit reward $r$ is identified with the constant profile $R_e(\gamma)\equiv r$, so this definition contains $\Psi_G(\gamma;\sigma\mid r)$ as a special case. For a non-terminal state $\sigma$ and a frontier control $e\in\mathcal F_G(\sigma)$, let $\Psi_{G\mid e}(\gamma;\sigma\mid R)$ denote the optimal calibrated value subject to attempting $e$ first and behaving optimally thereafter. Hence
\begin{equation}
    \Psi_G(\gamma;\sigma\mid R_e) = \max\left\{ \gamma, \max_{e\in\mathcal F_G(\sigma)} \Psi_{G\mid e}(\gamma;\sigma\mid R_e) \right\}.
\end{equation}

\begin{appendixlemma}[Series closure]
\label{lem:statewise-series-closure}
Let $G=\Ser(G_1,\cdots,G_m)$. Suppose that the statewise maximum-index property holds for each child $G_i$ under every admissible continuation profile. Then the statewise maximum-index property holds for $G$.
\end{appendixlemma}

\begin{proof}
By positive homogeneity, it suffices to consider the normalised exit reward $r=1$. Fix a reachable non-terminal state $\sigma$ of $G$. By the semantics of $\Ser$ composition and eager pruning, there is a unique active child $G_j$: the children $G_1,\cdots,G_{j-1}$ have already been breached, $G_j$ is still available, and the later children $G_{j+1},\cdots,G_m$ remain frozen in their fresh states. In particular,
\begin{equation}
\mathcal F_G(\sigma) = \mathcal F_{G_j}(\sigma_j),
\end{equation}
where $\sigma_j$ is the state induced on $G_j$.

Let
\begin{equation}
H_{j+1}\eqdef\Ser(G_{j+1},\cdots,G_m)
\end{equation}
denote the suffix following $G_j$, with $H_{m+1}$ interpreted as the breached empty suffix. Define the continuation profile
\begin{equation}
R_j(\gamma)\eqdef\Psi_{H_{j+1}}(\gamma;\sigma_{H_{j+1}}^0\mid1),
\end{equation}
where $R_m(\gamma)=1$ when $j=m$. Thus, if the attacker breaches $G_j$, it enters the suffix in its fresh state and receives continuation value $R_j(\gamma)$, measured from the completion time of $G_j$.

While $G_j$ is being attacked, every control in $H_{j+1}$ remains frozen. Moreover, exponential discounting is multiplicative across the active child and the suffix. Consequently, every policy for $G$ from state $\sigma$ corresponds to a policy for $G_j$ from state $\sigma_j$ with completion continuation $R_j$, and the two policies have identical pathwise discounted payoffs. Therefore,
\begin{equation}
\Psi_G(\gamma;\sigma\mid1)=\Psi_{G_j}(\gamma;\sigma_j\mid R_j).
\end{equation}
The same correspondence holds after fixing the first attempted control. For every $e\in\mathcal F_G(\sigma)$,
\begin{equation}
\Psi_{G\mid e}(\gamma;\sigma\mid1)=\Psi_{G_j\mid e}(\gamma;\sigma_j\mid R_j).
\end{equation}

The continuation profile seen after compromising any frontier control $e\in\mathcal F_{G_j}(\sigma_j)$ consists of the remainder of $G_j$, followed by $H_{j+1}$. Hence the embedded index of $e$ in the parent $\Ser$ subcomponent is exactly the index assigned to $e$ in $G_j$ under continuation profile $R_j$.

Applying the induction hypothesis to $G_j$ with continuation profile $R_j$ gives
\begin{equation}
\alpha_G(\sigma\mid1)=\max_{e\in\mathcal F_G(\sigma)}\alpha(e,k_e).
\end{equation}
It also gives, after fixing the deterministic tie-breaking rule, a control
\begin{equation}
e_G^*(\sigma)\in\arg\max_{e\in\mathcal F_G(\sigma)}\alpha(e,k_e)
\end{equation}
such that
\begin{equation}
\Psi_{G_j\mid e_G^*(\sigma)}(\gamma;\sigma_j\mid R_j)=\Psi_{G_j}(\gamma;\sigma_j\mid R_j)
\end{equation}
for every $\gamma<\alpha_G(\sigma\mid1)$. The two value correspondences above therefore imply
\begin{equation}
\Psi_{G\mid e_G^*(\sigma)}(\gamma;\sigma\mid1)=\Psi_G(\gamma;\sigma\mid1)
\end{equation}
throughout the same continuation region.
Finally, the induction hypothesis states that retirement is optimal for the active child exactly when
\begin{equation}
\gamma\ge\max_{e\in\mathcal F_G(\sigma)}\alpha(e,k_e).
\end{equation}
Because the $\Ser$ parent and active child have the same value and the same available actions, the same statement holds for $G$. Positive homogeneity extends the conclusion from the normalised case to every exit reward $r>0$.
\end{proof}

The key point above is that a reachable series states contains only one active child. $\Ser$ composition therefore introduces no competition between different children but it only changes the continuation profile against which the controls of the active child are evaluated. Before proving the $\Par$ case, we need an important lemma.

\begin{appendixlemma}[Parallel action-gap identity]
\label{lem:parallel-action-gap}
Let $G=\Par(G_1,G_2)$ and fix a reachable state $\sigma=(\sigma_1,\sigma_2)$. Define
\begin{equation}
\Phi_i(\gamma)\eqdef\Psi_{G_i}(\gamma;\sigma_i\mid1),\qquad g_i(\gamma)\eqdef\partial_\gamma^+\Phi_i(\gamma),
\end{equation}
for $i\in\{1,2\}$. For a control $e\in\mathcal F_{G_1}(\sigma_1)$, define
\begin{equation}
h_e(\gamma)\eqdef\Phi_1(\gamma)-\Psi_{G_1\mid e}(\gamma;\sigma_1\mid1).
\end{equation}
Then
\begin{equation}
\Psi_G(\gamma;\sigma\mid1)-\Psi_{G\mid e}(\gamma;\sigma\mid1)=h_e(\gamma)g_2(\gamma)+\int_{(\gamma,1]}h_e(z)\,dg_2(z).
\end{equation}
The analogous identity holds for a control in $G_2$ after exchanging the two children.
\end{appendixlemma}

\begin{proof}
By the statewise version of the parallel composition rule,
\begin{equation}
\Psi_G(\gamma;\sigma\mid1)=1-\int_\gamma^1g_1(z)g_2(z)\,dz.
\end{equation}
Because $\Phi_1$ is absolutely continuous with derivative $g_1$ almost everywhere, integration by parts for the Lebesgue--Stieltjes integral gives
\begin{equation}
\Psi_G(\gamma;\sigma\mid1)=\Phi_1(\gamma)g_2(\gamma)+\int_{(\gamma,1]}\Phi_1(z)\,dg_2(z).
\end{equation}
Indeed, $\Phi_1(1)=1$ and $g_2(1)=1$, so the boundary term at $1$ is equal to one.

The same transformation applies when the first action in $G_1$ is fixed to be $e$. To see this directly, let $X_2$ be the auxiliary threshold variable associated with $\Phi_2$, so that
\begin{equation}
\Pr(X_2\le z)=g_2(z).
\end{equation}
Conditional on $X_2=z$, the second branch acts as an outside option of value $\max\{\gamma,z\}$ for the first branch. Since attempting $e$ changes only the state of $G_1$, while $G_2$ remains frozen, the value obtained by fixing $e$ as the first action is
\begin{equation}
\Psi_{G\mid e}(\gamma;\sigma\mid1)=\mathbb E\!\left[\Psi_{G_1\mid e}\bigl(\max\{\gamma,X_2\};\sigma_1\mid1\bigr)\right].
\end{equation}
Writing this expectation using the distribution function $g_2$ yields
\begin{equation}
\Psi_{G\mid e}(\gamma;\sigma\mid1)=\Psi_{G_1\mid e}(\gamma;\sigma_1\mid1)g_2(\gamma)+\int_{(\gamma,1]}\Psi_{G_1\mid e}(z;\sigma_1\mid1)\,dg_2(z).
\end{equation}
Subtracting this identity from the corresponding identity for $\Psi_G$ gives
\begin{equation}
\Psi_G(\gamma;\sigma\mid1)-\Psi_{G\mid e}(\gamma;\sigma\mid1)=h_e(\gamma)g_2(\gamma)+\int_{(\gamma,1]}h_e(z)\,dg_2(z).
\end{equation}
Since the profiles are piecewise linear, the Stieltjes integral is simply a finite sum over the jumps of $g_2$.
\end{proof}

\begin{appendixlemma}[Parallel closure]
\label{lem:statewise-parallel-closure}
Let $G=\Par(G_1,\cdots,G_m)$. Suppose that the statewise maximum-index property holds for each child $G_i$. Then the statewise maximum-index property holds for $G$.
\end{appendixlemma}

\begin{proof}
By positive homogeneity, it suffices to prove the result for normalised exit reward $r=1$. We first consider the binary case
\begin{equation}
G=\Par(G_1,G_2).
\end{equation}
Fix a reachable non-terminal state $\sigma=(\sigma_1,\sigma_2)$. If one child has already become impossible and has therefore been pruned, the parent reduces to the remaining child and the result follows immediately from the induction hypothesis. We may therefore assume that both children remain viable.

Define
\begin{equation}
\Phi_i(\gamma)\eqdef\Psi_{G_i}(\gamma;\sigma_i\mid1),\qquad g_i(\gamma)\eqdef\partial_\gamma^+\Phi_i(\gamma),
\end{equation}
and let
\begin{equation}
\theta_i\eqdef\max_{e\in\mathcal F_{G_i}(\sigma_i)}\alpha(e,k_e).
\end{equation}
By the induction hypothesis, $\theta_i$ is the retirement threshold of child $G_i$. Hence
\begin{equation}
\Phi_i(\gamma)=\gamma\quad\text{for every }\gamma\ge\theta_i,
\end{equation}
and therefore
\begin{equation}
g_i(\gamma)=1\quad\text{for every }\gamma>\theta_i.
\end{equation}

Let
\begin{equation}
\theta\eqdef\max\{\theta_1,\theta_2\}.
\end{equation}
The statewise parallel-composition formula gives
\begin{equation}
\Psi_G(\gamma;\sigma\mid1)=1-\int_\gamma^1g_1(z)g_2(z)\,dz.
\end{equation}
If $\gamma\ge\theta$, then $g_1(z)=g_2(z)=1$ throughout $(\gamma,1)$, and consequently
\begin{equation}
\Psi_G(\gamma;\sigma\mid1)=1-\int_\gamma^1 1\,dz=\gamma.
\end{equation}
Thus, retirement is optimal for every $\gamma\ge\theta$.

Now suppose that $\gamma<\theta$. Without loss of generality, assume that $\theta_1=\theta$. Since $\gamma<\theta_1$, the induction hypothesis gives $\Phi_1(\gamma)>\gamma$. Using
\begin{equation}
\Phi_1(\gamma)-\gamma=\int_\gamma^1\bigl(1-g_1(z)\bigr)\,dz
\end{equation}
and $0\le g_2\le1$, we obtain
\begin{align*}
\Psi_G(\gamma;\sigma\mid1)-\gamma
&=\int_\gamma^1\bigl(1-g_1(z)g_2(z)\bigr)\,dz\\
&\ge\int_\gamma^1\bigl(1-g_1(z)\bigr)\,dz\\
&=\Phi_1(\gamma)-\gamma\\
&>0.
\end{align*}
Therefore, continuing is strictly preferable whenever $\gamma<\theta$. It follows that
\begin{equation}
\alpha_G(\sigma\mid1)=\theta=\max_{e\in\mathcal F_G(\sigma)}\alpha(e,k_e).
\end{equation}

It remains to prove that one fixed maximum-index control is optimal throughout this continuation region. By the induction hypothesis for $G_1$, after applying the deterministic tie-breaking rule there exists
\begin{equation}
e^*(\sigma)\in\arg\max_{e\in\mathcal F_{G_1}(\sigma_1)}\alpha(e,k_e)
\end{equation}
such that
\begin{equation}
\Psi_{G_1\mid e^*(\sigma)}(z;\sigma_1\mid1)=\Phi_1(z)
\end{equation}
for every $z<\theta_1$. Both sides are continuous, so equality also holds at $z=\theta_1$. Consequently, the standalone action gap
\begin{equation}
h_{e^*(\sigma)}(z)\eqdef\Phi_1(z)-\Psi_{G_1\mid e^*(\sigma)}(z;\sigma_1\mid1)
\end{equation}
vanishes throughout $[0,\theta_1]$.
Since $g_2(z)=1$ for every $z>\theta_2$, the Stieltjes measure $dg_2$ is supported on $[0,\theta_2]$. Because \(\theta_2\le\theta_1\), this support is contained in the interval on which $h_{e^*(\sigma)}$ vanishes. For every $\gamma<\theta_1$, Lemma~\ref{lem:parallel-action-gap} therefore gives
\begin{align}
\Psi_G(\gamma;\sigma\mid1)-\Psi_{G\mid e^*(\sigma)}(\gamma;\sigma\mid1)
&=h_{e^*(\sigma)}(\gamma)g_2(\gamma)
+\int_{(\gamma,1]}h_{e^*(\sigma)}(z)\,dg_2(z)\\
&=0. \nonumber
\end{align}
Hence
\begin{equation}
\Psi_{G\mid e^*(\sigma)}(\gamma;\sigma\mid1)=\Psi_G(\gamma;\sigma\mid1)
\end{equation}
for every $\gamma<\alpha_G(\sigma\mid1)$. The same maximum-index control is therefore optimal independently of the buyout throughout the continuation region.

The result for an $m$-ary parallel node follows by repeated binary composition. Equivalently, choose a child $G_j$ whose maximum frontier index is largest, combine the remaining children into
\begin{equation}
H\eqdef\Par(G_1,\cdots,G_{j-1},G_{j+1},\cdots,G_m),
\end{equation}
and apply the binary argument to $\Par(G_j,H)$. Associativity of parallel composition gives
\begin{equation}
\alpha_G(\sigma\mid1)=\max_{1\le i\le m}\max_{e\in\mathcal F_{G_i}(\sigma_i)}\alpha(e,k_e),
\end{equation}
and the tie-broken maximum-index control in a maximising child remains optimal for every buyout below this threshold. Positive homogeneity extends the result to every exit reward $r>0$.
\end{proof}

\begin{proof}[Proof of Proposition~\ref{prop:statewise-max-index}]
The result follows by structural induction over the SP parse tree. The leaf case is Lemma~\ref{lem:statewise-leaf}; series and parallel composition preserve the property by Lemmas~\ref{lem:statewise-series-closure} and \ref{lem:statewise-parallel-closure}.
\end{proof}

\section{On Indexability of the Attacker's Problem and the Gittins Index Policy}
\subsection{Background on the Gittins Index}
\label{app:gittins-background}

The Gittins index is a scalar priority assigned to the current state of a discounted stochastic process. In the classical multi-armed bandit setting, several independent processes, or \textit{arms}, are available. At each decision epoch, exactly one arm is selected; the selected arm generates reward and changes state, while every unselected arm remains frozen. Future rewards are discounted. The Gittins-index theorem states that there exists a state-dependent scalar index for each arm such that an optimal policy always selects an arm with the largest current index \citep{gittins1979bandit}. Thus, a dynamic optimisation problem over the joint state of all arms is reduced to comparing one scalar per arm. Below are several equivalent interpretations of the index which are common.

\begin{enumerate}
    \item \textbf{Priority Interpretation:} The index measures the attractiveness of operating an arm in its current state, and only the ordering of the indices is needed to determine the next action. Different normalisations may assign different numerical values to an index while preserving the same ordering.
    \item \textbf{Outside/Retirement Option:} Suppose that the decision-maker may either continue operating the arm or retire and receive a payoff $\gamma$. For small values of $\gamma$, continuing is preferable, whereas for sufficiently large values retirement is optimal. Under the usual discounted-bandit assumptions, the Gittins index is the critical value of $\gamma$ at which retirement first becomes optimal. This is often called the \textit{retirement}, \textit{calibration}, or \textit{break-even} interpretation of the index.
    \item \textbf{Reward-to-Work Ratio} In a reward-rate formulation, let $R_\tau(x)$ denote the expected discounted reward obtained by operating an arm from state $x$ until a stopping time $\tau$, and let $W_\tau(x)$ denote the corresponding expected discounted amount of active time. One common normalisation of the Gittins index is
        \begin{equation}
            \nu(x):=\sup_{\tau>0}\frac{R_\tau(x)}{W_\tau(x)}.
        \end{equation} Equivalently, $\nu(x)$ is the largest charge per unit of discounted work for which it remains worthwhile to operate the arm. The ratio, retirement, and fair-charge interpretations describe the same break-even comparison under different normalisations.
\end{enumerate}  

The retirement interpretation is the \textit{most natural} for the attacker problem studied here. The outside option represents the value of abandoning the current attack subcomponent and pursuing alternatives elsewhere on the frontier. For a subcomponent $G$ in state $\sigma$, we therefore introduce a buyout $\gamma$ and define its index as the \textit{smallest buyout} at which retirement becomes optimal. This calibrated threshold is the quantity denoted by $\alpha_G(\sigma\mid r)$ below.

There is, however, one distinction from the classical bandit model. A standard arm has a single action whenever it is selected, whereas continuing with an SP subcomponent requires the attacker to choose among several frontier controls. Such an object is called a \textit{superprocess}. To obtain a valid index for a superprocess, it is not enough that a retirement threshold exists. An extra condition is that the internal control used when continuing \textit{must also be selectable independently of the buyout}. The remainder of this section proves that every SP subcomponent has this buyout-independent continuation property. Consequently, each subcomponent admits a well-defined Gittins-type index, and greedily selecting a maximum-index frontier control is optimal.

\subsection{Indexability of Controlled Superprocesses}
\label{app:superprocess-indexability}

\paragraph{Controlled superprocesses.} We first state the controlled-process extension of the classical Gittins-index framework used in our proof.
A standard bandit arm has a state but no internal control choice, that is whenever the arm is selected, its reward and state transition are determined by its current state. A \textit{superprocess} generalises an arm by allowing an additional internal decision. Formally, a superprocess has a state $x$, an admissible internal-control set $\mathcal U(x)$, and, for each $u\in\mathcal U(x)$, a reward and state-transition rule. Selecting the superprocess therefore requires both selecting the process itself and choosing one of its internal controls.

\paragraph{Alternative superprocesses.} A \textit{family of alternative superprocesses} consists of superprocesses indexed by $i\in\{1,\cdots,n\}$, with joint state $(x_1,\cdots,x_n)$. At each decision epoch, the decision-maker selects exactly one superprocess $i$ and then chooses an internal control $u\in\mathcal U_i(x_i)$. The selected superprocess generates the current reward and changes state.

\paragraph{Simple families.} The family is called \textit{simple} when every unselected superprocess remains frozen. More precisely, suppose that the joint state before a decision is $(x_1,\cdots,x_n)$ and that superprocess $i$ is selected. The selected state $x_i$ may change according to the transition rule associated with the chosen internal control, whereas the state of every unselected superprocess remains equal to its pre-decision value. Thus, writing $(x_1',\cdots,x_n')$ for the joint state after the transition,
\begin{equation}
x_j'=x_j\qquad\text{for every }j\ne i.
\end{equation}
The superprocesses are therefore coupled only through the restriction that exactly one of them may be selected at each decision epoch. This is the controlled analogue of the frozen-arm assumption in the classical multi-armed bandit problem.

\paragraph{Retirement thresholds.} To determine whether a scalar priority can be assigned to the state of a superprocess, consider its auxiliary retirement problem. Fix a superprocess in state $x$ and offer the decision-maker a retirement payoff $\gamma$. Let $V(\gamma;x)$ denote the optimal value when the decision-maker may either retire for $\gamma$ or continue operating the superprocess. Suppose that the retirement set has threshold form (i.e there exists a scalar $\alpha(x)$ such that continuing is strictly preferable for $\gamma<\alpha(x)$, while retirement is optimal for $\gamma\ge\alpha(x)$).

\paragraph{Buyout-independent internal control.} For an ordinary bandit arm, this threshold is sufficient to define an index because the arm has only one continuation action. For a superprocess, however, the internal control chosen while continuing could depend on the competing retirement payoff. The following additional property rules out such dependence.

\begin{appendixdef}[Buyout-independent control condition]
\label{def:buyout-independent-control}
A superprocess satisfies the \textit{buyout-independent control
condition} if, for every state $x$, there exists an internal control
$g(x)\in\mathcal U(x)$ such that $g(x)$ is optimal for every buyout
$\gamma<\alpha(x)$.
\end{appendixdef}

Thus, the decision of whether to continue may depend on the buyout, but, conditional on continuing, the same state-dependent internal control may be used throughout the continuation region. The buyout-independent control condition is referred to as ``\textit{Condition D}" by \citet[Chapter~4]{doi:https://doi.org/10.1002/9780470980033.ch4}.

\paragraph{The Superprocess Index Theorem.} We use the following form of the superprocess index theorem.

\begin{appendixtheorem}[Gittins Index Theorem for Simple Alternative Superprocesses]
\label{thm:superprocess-index-theorem}
Consider a simple family of alternative discounted superprocesses. Suppose that, for every superprocess and every reachable state, the retirement set has threshold form and the buyout-independent control condition holds. Then each superprocess state admits a scalar Gittins index. Under the retirement normalisation, this index is its critical retirement value $\alpha(x)$.

Moreover, there exists an optimal policy that, at every decision epoch, selects a superprocess satisfying
\begin{equation}
i^*\in\arg\max_{1\le i\le n}\alpha_i(x_i)
\end{equation}
and, within the selected superprocess, applies a
buyout-independent maximizing control $g_i(x_i)$.
\end{appendixtheorem}

The theorem separates the proof of index optimality into two tasks. First, one must show that each process has a \textit{scalar retirement threshold}. Second, because the process has internal controls, one must show that a \textit{single internal control remains optimal throughout the continuation region}. The preceding statewise maximum-index property establishes both requirements for SP attack subcomponents.

\subsection{Indexability of SP Attack Subcomponents}
\label{app:sp-indexability}

\paragraph{Superprocess representation.}
Fix an SP subcomponent $G$ in a reachable state $\sgn$, and write $\mathcal F_G(\sgn)$ for the controls of $G$ currently on the attack frontier. We regard $G$ as a controlled superprocess whose internal controls are the elements of $\mathcal F_G(\sgn)$. Selecting an internal control $e\in\mathcal F_G(\sgn)$ means attempting $e$ once. The resulting success or failure updates the local state of $e$ and may trigger a breach, lockout, impossibility propagation, or pruning according to the SP semantics.

At a $\Par$ node, the viable child subcomponents are alternatives to each other. That means that the attacker operates one child at a time, while every unselected child remains frozen. If the selected child compromises the parallel parent, the remaining children are pruned, and their subsequent states are irrelevant. At a $\Ser$ node, only the current available child is active, while every later child remains frozen in its fresh state until the preceding child is breached. Hence, before termination or pruning, an unselected subcomponent does not evolve. The attacker problem therefore satisfies the frozen-process requirement of a simple family of alternative superprocesses.

\paragraph{Verification of the indexability conditions.}
The preceding statewise maximum-index property establishes both conditions required by Theorem~\ref{thm:superprocess-index-theorem}.

\begin{appendixprop}[Indexability of SP subcomponents]
\label{prop:sp-subcomponent-indexability}
Fix an SP subcomponent $G$ and a reachable non-terminal state $\sgn$. For each $e\in\mathcal F_G(\sgn)$, let $k_e$ denote the failure count of $e$ in state $\sgn$. Under the normalised exit reward $r=1$, the auxiliary retirement problem for $G$ has threshold
\begin{equation}
\alpha_G(\sgn\mid1)=\max_{e\in\mathcal F_G(\sgn)}\alpha(e,k_e).
\end{equation}
Moreover, after fixing a deterministic tie-breaking rule, there exists a control
\begin{equation}
e_G^*(\sgn)\in\arg\max_{e\in\mathcal F_G(\sgn)}\alpha(e,k_e)
\end{equation}
that is optimal for every buyout $\gamma<\alpha_G(\sgn\mid1)$. Consequently, every SP subcomponent satisfies the buyout-independent control condition.
\end{appendixprop}

\begin{proof}
By Proposition~\ref{prop:statewise-max-index}, continuing with $G$ is strictly preferable when
\begin{equation}
\gamma<\max_{e\in\mathcal F_G(\sgn)}\alpha(e,k_e),
\end{equation}
whereas retirement is optimal at and above this value. Hence the retirement set has threshold form, with
\begin{equation}
\alpha_G(\sgn\mid1)=\max_{e\in\mathcal F_G(\sgn)}\alpha(e,k_e).
\end{equation}
The same proposition provides a tie-broken maximum-index control $e_G^*(\sgn)$ that is optimal for every buyout below this threshold. The optimal internal control therefore depends on the subcomponent state $\sgn$, but not on the buyout within the continuation region. This is exactly the buyout-independent control condition of Definition~\ref{def:buyout-independent-control}.
\end{proof}

\paragraph{Identification of the subcomponent index.}
Proposition~\ref{prop:sp-subcomponent-indexability} verifies the hypotheses of Theorem~\ref{thm:superprocess-index-theorem}. Therefore, every reachable SP subcomponent state admits a well-defined scalar Gittins index. Under the retirement normalisation, this index is its critical retirement value,
\begin{equation}
\alpha_G(\sgn\mid1)=\max_{e\in\mathcal F_G(\sgn)}\alpha(e,k_e).
\end{equation}
Thus, $\alpha(e,k_e)$ is the scalar priority of the frontier control state $(e,k_e)$, while their maximum is the Gittins index of the active SP subcomponent.

\subsection{Dependence on the Forward Cone}
\label{app:index-locality}

\paragraph{Forward-cone continuation.}
For a frontier control $e$, let $R_e:[0,1]\to[0,1]$ denote the calibrated continuation profile obtained after successfully compromising $e$. This profile includes every subcomponent that must subsequently be breached before the terminal reward is obtained. We call this downstream structure the \textit{forward cone} of $e$.

\begin{appendixprop}[Locality of control indices]
\label{prop:index-locality}
Fix a frontier control $e$ with failure count $k_e$. Its index $\alpha(e,k_e)$ depends only on the local parameters and failure count of $e$ and on the calibrated continuation profile $R_e$ of its forward cone. In particular, it is independent of the current states of controls in parallel sibling subcomponents.
\end{appendixprop}

\begin{proof}
By Lemma~\ref{lem:intermediate-control-index}, the index of $e$ is characterised by
\begin{equation}
\alpha(e,k_e)=c(e,k_e)R_e\bigl(\alpha(e,k_e)\bigr),
\end{equation}
where $c(e, k_e)$ was defined earlier. The factor $c(e,k_e)$ depends only on the duration, success probability, and failure count of $e$. It therefore remains to identify the information contained in $R_e$.

After a successful attempt on $e$, the attacker continues only through subcomponents that lie downstream of $e$. At a $\Ser$ ancestor, the continuation appends the remaining series suffix, which belongs to the forward cone of $e$. At a $\Par$ ancestor, breaching the selected child breaches the entire parallel subcomponent and prunes its remaining children. Those parallel siblings are therefore not part of the continuation obtained after breaching the selected child, and their current states do not enter $R_e$.

Applying this argument successively along the path from $\Ctrl(e)$ to the root shows that $R_e$ is determined entirely by the downstream series continuations and the continuation supplied above the forward cone. No current state of a parallel sibling enters its construction. Since $\alpha(e,k_e)$ is determined by $c(e,k_e)$ and $R_e$, the claimed locality follows.
\end{proof}

\paragraph{Consequence for index evaluation.}
The proposition separates the state information needed to evaluate an index from the full attacker state. The failure count $k_e$ supplies the local control state, while $R_e$ summarises the entire relevant downstream continuation. States belonging to alternative parallel branches need not be included. This is the locality property asserted in Theorem~\ref{thm:index-policy}.

\subsection{Optimality of the Gittins Index Policy and Proof of Theorem 1}
\label{app:gittins-policy-optimality}

\paragraph{Greedy index policy.}
We now apply the controlled-superprocess theorem to the full attacker problem.

\begin{proof}[Proof of Theorem~\ref{thm:index-policy}]
Fix a defender allocation $\ell$ and a reachable non-terminal attacker state $\sgn$. At every $\Par$ node, the viable child subcomponents form a simple family of alternative superprocesses: the attacker operates one child, while every unselected child remains frozen until it is selected or pruned. At every $\Ser$ node, only the current available child is active. Proposition~\ref{prop:sp-subcomponent-indexability} verifies that every active SP subcomponent has a scalar retirement threshold and satisfies the buyout-independent control condition. Theorem~\ref{thm:superprocess-index-theorem} therefore, implies that an optimal action selects an active subcomponent with largest Gittins index and applies its buyout-independent maximising control.

By Proposition~\ref{prop:sp-subcomponent-indexability}, the Gittins index of a subcomponent $G$ in state $\sgn$ is
\begin{equation}
\alpha_G(\sgn\mid1)=\max_{e\in\mathcal F_G(\sgn)}\alpha(e,k_e).
\end{equation}
Consequently, maximising over the active subcomponents is equivalent to maximising over the current attack frontier. After fixing a deterministic tie-breaking rule, the selected internal control may therefore be chosen as
\begin{equation}
\pi_\ell^*(\sgn)\in\arg\max_{e\in\mathcal F(\sgn)}\alpha(e,k_e).
\end{equation}
If the maximum index is positive, the statewise maximum-index property shows that this control is optimal at the root buyout $\gamma=0$. If the maximum index is zero, the calibrated value is also zero, and a maximum-index control remains optimal.

Applying the same rule after every state transition defines a deterministic Markov policy. Since the attacker state process is finite and acyclic, and the selected action is Bellman-optimal at every reachable state, this policy is optimal. Finally, setting the root buyout to $\gamma=0$ and the terminal exit reward to $r=1$ recovers the original attacker problem. Hence, there exists an optimal attacker policy satisfying
\begin{equation}
\pi_\ell^*(\sgn)\in\arg\max_{e\in\mathcal F(\sgn)}\alpha(e,k_e)
\end{equation}
at every reachable non-terminal state.
\end{proof}

\section{Supplementary Material for Gittins-Index Computation by the 2-Phase Fold Algorithm}
\label{app:index-computation}
Having first established basic properties of calibrated value profiles in \Cref{app:calibrated_value_profile}, for this section, we give a technical review of the complete folding algorithm to compute the Gittins indices of the game. Throughout this section, fix a defender allocation $\ell$ and define $\beta_e\eqdef \exp{(-\lambda\ell_e)}$. For a subcomponent $G$, let $\sigma_G^0$ denote its fresh initial state.

\subsection{2-Phase Folding Algorithm to compute the Gittins Index}
\label{app:index-algorithm}

We now combine the preceding results into a two-pass algorithm. The first pass computes the intrinsic profile of every subcomponent. The second pass propagates the continuation supplied by the surrounding parse-tree context and recovers the continuation profile $R_e$ of every control. For a profile $f$ and a continuation profile $r$, define their perspective composition pointwise, whenever $r(\gamma)>0$, by
\begin{equation}
\mathcal P_f[r](\gamma)\eqdef r(\gamma)f\left(\frac{\gamma}{r(\gamma)}\right).
\end{equation}
At a point where $\gamma=r(\gamma)=0$, the expression is interpreted by continuity. By homogeneity, $\mathcal P_f[r]$ is the value of a subcomponent with intrinsic profile $f$ when successful breach yields the continuation reward $r(\gamma)$.

\paragraph{Pass 1: Bottom-up profile construction. (\Cref{alg:profile})}
Traverse the SP parse tree $\mathcal T$ from the leaves to the root. At a leaf $\Ctrl(e)$, construct the intrinsic profile $\widehat\Psi_e$ from the recursion in Lemma~\ref{lem:isolated-control-index} with $r=1$ and initial state $k=0$.

At a $\Par$ node $G=\Par(G_1,\cdots,G_m)$, merge the sorted breakpoints of the child profiles. On each interval between consecutive breakpoints, multiply the active child slopes and integrate backward from $\widehat\Psi_G(1)=1$, as prescribed by Lemma~\ref{lem:parallel-composition-proof}.

At a $\Ser$ node $G=\Ser(G_1,\cdots,G_m)$, construct the suffix profiles backward. Set $\rho_{m+1}(\gamma)\eqdef 1$ and compute
\begin{equation}
\rho_j(\gamma)=\mathcal P_{\widehat\Psi_{G_j}}[\rho_{j+1}](\gamma),\qquad j=m,m-1,\cdots,1.
\end{equation}
Store the suffix profiles $\rho_j$ and set $\widehat\Psi_G\eqdef \rho_1$. The profile at the root satisfies
\begin{equation}
V_{\mathcal N}^*(\ell)=\widehat\Psi_{\mathcal N}(0).
\end{equation}

Refer to \Cref{alg:profile} for pseudocode.

\begin{algorithm}[h!]
\caption{\textsc{Profile}$(G)$: constructs and stores the intrinsic profile $\Psh_G$ (Pass 1).}
\label{alg:profile}
\begin{algorithmic}[1]
\Function{Profile}{$G$}
  \If{$G=\Ctrl(e)$}
     \State $\Psh\gets(s\mapsto s)$ \Comment{state $k=q_e$ is locked out}
     \For{$k=q_e-1,\dots,0$}
        \State $\Psh\gets\bigl(s\mapsto\max\{s,\beta_e[p_e(k)+(1-p_e(k))\Psh(s)]\}\bigr)$
     \EndFor
     \State Store $\Psh_G\gets\Psh$
  \ElsIf{$G=\Par(G_1,\dots,G_n)$}
     \For{$i=1,\dots,n$}
        \State $\Psh_i\gets\Call{Profile}{G_i}$
     \EndFor
     \State $\Psh_G\gets\bigl(s\mapsto1-\int_s^1\prod_{i=1}^n\partial_z^+\Psh_i(z)\,dz\bigr)$
  \Else\Comment{$G=\Ser(G_1,\dots,G_n)$}
     \For{$i=1,\dots,n$}
        \State $\Psh_i\gets\Call{Profile}{G_i}$
     \EndFor
     \State $\Psh_{\mathrm{acc}}\gets(s\mapsto1)$
     \For{$i=n,\dots,1$}
        \State $\Psh_{\mathrm{acc}}\gets\bigl(s\mapsto\Psh_{\mathrm{acc}}(s)\Psh_i(s/\Psh_{\mathrm{acc}}(s))\bigr)$
     \EndFor
     \State Store $\Psh_G\gets\Psh_{\mathrm{acc}}$
  \EndIf
  \State \Return $\Psh_G$
\EndFunction
\State $V_{\Network}^*(\ell)\gets\Call{Profile}{\Network}(0)$
\end{algorithmic}
\end{algorithm}

\paragraph{Pass 2: Top-down continuation propagation. (\Cref{alg:indices})}
For every node $G$, let $R_G:[0,1]\to[0,1]$ denote the calibrated continuation profile obtained after breaching $G$ in its surrounding context. Initialize the root continuation by $R_{\mathcal N}(\gamma)\eqdef 1$.

If $G=\Par(G_1,\cdots,G_m)$, breaching any child compromises $G$, so every child receives the same continuation:
\begin{equation}
R_{G_i}\eqdef R_G,\qquad i=1,\cdots,m.
\end{equation}

If $G=\Ser(G_1,\cdots,G_m)$, let $S_{j+1}\eqdef \Ser(G_{j+1},\cdots,G_m)$ denote the suffix following $G_j$, with $S_{m+1}$ the empty breached suffix. Breaching $G_j$ requires breaching $S_{j+1}$ before receiving $R_G$. Hence
\begin{equation}
R_{G_j}(\gamma)=\mathcal P_{\widehat\Psi_{S_{j+1}}}[R_G](\gamma),\qquad j=1,\cdots,m,
\end{equation}
where $\widehat\Psi_{S_{m+1}}\equiv1$. The suffix profiles needed in this expression were constructed during the bottom-up pass (Pass 1).

At a leaf $\Ctrl(e)$, set $R_e\eqdef R_{\Ctrl(e)}$. For every failure count $k=0,\cdots,q_e-1$, compute $\alpha(e,k)$ as the smallest solution of
\begin{equation}
\alpha(e,k)=c(e,k)R_e\bigl(\alpha(e,k)\bigr).
\end{equation}
Because $R_e$ is piecewise linear, the root can be found by scanning its pieces. If $R_e(\gamma)=A+B\gamma$ on the piece containing the root, Corollary~\ref{cor:affine-continuation-index} gives
\begin{equation}
\alpha(e,k)=\frac{c(e,k)A}{1-c(e,k)B}.
\end{equation}
Since $c(e,k)$ and $\alpha(e,k)$ are non-increasing in $k$, all indices of a control can be obtained in one monotone sweep through the pieces of $R_e$. Refer to \Cref{alg:indices} for pseudocode.

\begin{algorithm}[h!]
\caption{\textsc{Indices}$(G,R)$: stores the indices in $G$, where $R$ is the continuation profile after breaching $G$ (Pass 2).}
\label{alg:indices}
\begin{algorithmic}[1]
\Function{Indices}{$G,R$}
  \If{$G=\Ctrl(e)$}
     \For{$k=q_e-1,\dots,0$}
        \State $\al(e,k)\gets\inf\{\gamma\in[0,1]:\gamma=c(e,k)R(\gamma)\}$
     \EndFor
  \ElsIf{$G=\Ser(G_1,\dots,G_n)$}
     \State $R_{\mathrm{acc}}\gets R$
     \For{$i=n,\dots,1$}
        \State $\Call{Indices}{G_i,R_{\mathrm{acc}}}$
        \State $R_{\mathrm{acc}}\gets\bigl(\gamma\mapsto R_{\mathrm{acc}}(\gamma)\Psh_{G_i}(\gamma/R_{\mathrm{acc}}(\gamma))\bigr)$
     \EndFor
  \Else\Comment{$G=\Par(G_1,\dots,G_n)$}
     \For{$i=1,\dots,n$}
        \State $\Call{Indices}{G_i,R}$
     \EndFor
  \EndIf
\EndFunction
\State $\Call{Indices}{\Network,\gamma\mapsto1}$
\end{algorithmic}
\end{algorithm}

\begin{proposition}[Complexity]
\label{prop:index-algorithm-complexity}
Let $Q\eqdef \sum_{e\in E}q_e$ and let $d$ be the depth of the SP parse tree. The two-pass procedure computes the attacker value and all control indices in $\Oh{Qd}$ arithmetic operations and uses $\Oh{Qd}$ space when all profiles required by the top-down pass are memoised.
\end{proposition}

\begin{proof}
For a subcomponent $G$, let $Q_G\eqdef\sum_{e\in G}q_e$. The intrinsic profile of a leaf $\Ctrl(e)$ has at most $q_e+1$ affine pieces. At a $\Par$ node, all slope changes occur at child breakpoints, so the number of pieces is at most the sum of the numbers of child pieces. At a $\Ser$ node, the recursive series compositions are monotone, so each breakpoint of a child or suffix profile is encountered at most once. Consequently, $\widehat\Psi_G$ has $\Oh{Q_G}$ pieces and can be constructed in $\Oh{Q_G}$ arithmetic operations from sorted child profiles.

Summing over all parse-tree nodes gives
\begin{equation}
\sum_{G\in\mathcal T}Q_G=\sum_{e\in E}q_ed_e\le Qd,
\end{equation}
where $d_e$ is the number of parse-tree nodes containing $\Ctrl(e)$. Hence the bottom-up pass takes $\Oh{Qd}$ time.

The same accounting applies to the top-down continuation profiles: each control-state breakpoint is propagated through at most $d$ parse-tree levels. At a leaf, monotonicity of $\alpha(e,k)$ in $k$ allows all $q_e$ fixed points to be found by one sweep through the continuation profile. The top-down pass therefore also takes $\Oh{Qd}$ time. Memoising the intrinsic, suffix, and continuation profiles stores $\Oh{Q_G}$ pieces at each relevant parse-tree node, giving $\Oh{Qd}$ total space.
\end{proof}

\section{Supplementary Material for $n$-RT Reverse Tape Algorithm}

\subsection{Introduction}
As a recap, the defender must compute subgradients $g(\ell) \in \partial V^\star_\Network (\ell)$ and execute regret matching to obtain optimal defensive allocation $\ell$ on every control against an adaptive attacker. This appendix presents an exact value-and-gradient algorithm ($n$-ary Reverse Tape, or $n$-RT for short) for series-parallel (SP) security networks, denoted by $\Network$, whose parse trees may contain internal nodes of arbitrary arity. 

\subsubsection{Overview of the Algorithm}
Recall that each subcomponent is represented by a piecewise-linear calibrated profile constructed from its children. The algorithm processes child breakpoints in order, using a balanced tree to maintain the product of active slopes at $\Par$ nodes, and the composed series update at $\Ser$ nodes. During this construction, it \textit{records all arithmetic operations on a reverse-mode tape}. After evaluating the root profile, the tape is \textit{traversed backwards}, and the chain rule yields derivatives with respect to all defender delays \textit{in a single reverse pass}.

\subsection{Problem setup and profile representation}
\label{app:mary-setup}

\paragraph{Basic Notation} For each leaf control \(e\in \mathcal E\), let \(q_e\) denote the maximum number of permitted attempts at that control. Let \(\ell_e\) denote the defender delay incurred by one attempt at control \(e\). The delay is converted into a multiplicative discount factor $\beta_e=\exp(-\lambda \ell_e),$ where \(\lambda>0\) is a fixed discount-rate parameter. For any node \(G\), let \(E(G)\subseteq E\) denote the set of leaf controls in the subtree rooted at \(G\). The total attempt budget in that subtree is defined by $Q_G=\sum_{e\in \mathcal  E(G)}q_e.$ If \(\Network\) denotes the root of $\mathcal{T}$, then the total attempt budget of the complete network is $Q=Q_\Network=\sum_{e\in \mathcal{E}}q_e$. For an internal node $G$, let $m_G$ be its number of children.

\paragraph{Intrinsic Calibrated Profile Representation for $G$ in parse tree} Each node \(G\) can be represented by an intrinsic calibrated profile, $\widehat{\Psi}_G(s)$.  The variable \(s\) is the attacker's outside option and it is the payoff that the attacker may accept by stopping instead of continuing to attack the subcomponent \(G\). The value \(\widehat{\Psi}_G(s)\) is the attacker's optimal expected discounted payoff when this outside option is available. Let \(\Network\) denote the root node. The attacker value of the complete network is $V^\star_\Network(\ell)=\widehat{\Psi}_\Network(0),$ where $\ell=(\ell_e)_{e\in \mathcal E}$ is the vector of defender delays.

The calibrated profile of every node is continuous and piecewise linear.
For node \(G\), let
\begin{equation}
     0=x_{G,0}<x_{G,1}<\cdots<x_{G,M_G}=1
\end{equation} be its ordered breakpoints, where \(M_G\) is the number of affine pieces.
On the \(j\)-th interval,
\begin{equation}
    s\in[x_{G,j-1},x_{G,j}],
\end{equation}
the profile is represented as
\[\widehat{\Psi}_G(s) = \alpha_{G,j}s+\beta_{G,j}, \qquad j=1,\cdots,M_G,\] 
where \(\alpha_{G,j}\) is the slope and \(\beta_{G,j}\) is the intercept of the \(j\)-th affine piece. The \textit{right derivative} of the profile is denoted by \(g_G(s)=\widehat{\Psi}_G'(s)\) It is a step function whose value on the \(j\)-th interval is \(\alpha_{G,j}\).

The number of pieces satisfies $M_G\leq Q_G+1$. Consequently, the complete profile of a node can be stored using a number of breakpoints and affine coefficients that is linear in its subtree attempt budget. The breakpoint locations are part of the profile construction and depend on the model parameters. They must therefore be differentiated together with the slopes and intercepts. In particular, treating the breakpoints as fixed would omit the effect of changes in the widths of the affine intervals.

\paragraph{Outline of subsequent sections} In the following sections, we will talk about individual parts of the algorithms that deal with $\Ser$ and $\Par$ nodes, $\Ctrl$ leaf profile construction, followed by the complete overview on how each parts of the algorithm fit into the entire picture.

\subsection{Reverse-Tape Framework}
\label{app:reverse-tape-framework}

\paragraph{Forward and reverse passes.}
\(n\)-RT first constructs calibrated profiles from the leaves to the root. Every scalar operation used to compute affine coefficients, breakpoint locations, and endpoint values is recorded on a reverse-mode tape. After evaluating
\begin{equation}
V_{\Network}^*(\ell)=\widehat\Psi_{\Network}(0),
\end{equation}
the tape is traversed backwards, applying the chain rule to obtain derivatives with respect to all defender delays in one reverse pass. Breakpoints are recorded as computed quantities, rather than fixed constants, because changes in \(\ell\) may move both child and parent breakpoints.

At regular points, where the active affine pieces and breakpoint ordering are locally unique, the reverse pass returns the exact gradient. At ties, coincident changes are processed together and a fixed deterministic rule selects an active computation branch, yielding a valid subgradient. The same framework applies to both \(\Par\) and \(\Ser\) nodes, although the quantities maintained by the segment tree differ. We begin with parallel composition, where all child profiles are evaluated at the same outside option and the parent slope is the product of the active child slopes. This provides the simpler setting in which to introduce the breakpoint sweep and its reverse differentiation.

\subsection{Direct \(n\)-ary parallel composition}
\label{app:nary-parallel}

Let
\begin{equation}
G=\Par(G_1,\ldots,G_n).
\end{equation}
For each child $G_i$, write $\widehat{\Psi}_{G_i}$ for its calibrated profile and
\begin{equation}
g_i(s)=\partial_s^+\widehat{\Psi}_{G_i}(s)
\end{equation}
for its right derivative. Each $g_i$ is a step function because $\widehat{\Psi}_{G_i}$ is piecewise linear.

\begin{appendixprop}[Parallel slope identity]
\label{prop:nary-parallel-slope}
The right derivative of the parent profile is
\begin{equation}
g_G(s)=\prod_{i=1}^{n}g_i(s).
\label{eq:nary-parallel-slope}
\end{equation}
Consequently, the parent slope is constant while every child remains on its current affine piece, and it changes only when at least one child reaches a breakpoint.
\end{appendixprop}

\begin{proof}
The binary parallel rule gives $g_{\Par(H_1,H_2)}=g_{H_1}g_{H_2}$. Repeated application of this identity yields Equation~\ref{eq:nary-parallel-slope}. Associativity of multiplication shows that the result is independent of the binary bracketing.
\end{proof}

\paragraph{Segment-tree invariant.}
Sweep the common input $s$ from $0$ to $1$. At the current sweep position, child $G_i$ is on one affine piece. Let $c_i$ be its slope and let $t_i$ be the right endpoint of that piece. For the final piece, set $t_i=1$.

Store the active child data in a balanced segment tree with one leaf for each child. Leaf $i$ stores $(c_i,t_i,i)$. Every internal node stores:

\begin{enumerate}
\item the product of the active slopes in its subtree; and
\item the earliest active breakpoint in its subtree, together with every child attaining that breakpoint.
\end{enumerate}

If the left and right subtrees store products $p_L,p_R$ and earliest breakpoints $t_L,t_R$, the parent stores
\begin{equation}
p=p_Lp_R,\qquad t=\min\{t_L,t_R\}.
\label{eq:parallel-tree-combine}
\end{equation}
At the root,
\begin{equation}
p_{\mathrm{root}}=\prod_{i=1}^{n}c_i,
\qquad
t_{\mathrm{root}}=\min_{1\leq i\leq n}t_i.
\label{eq:parallel-root-summary}
\end{equation}

\paragraph{Forward sweep.}
Let $s_{\mathrm{cur}}$ be the current parent input and set
\begin{equation}
s_{\mathrm{next}}=\min\{t_{\mathrm{root}},1\}.
\end{equation}
By Proposition~\ref{prop:nary-parallel-slope},
\begin{equation}
g_G(s)=p_{\mathrm{root}},
\qquad
s\in[s_{\mathrm{cur}},s_{\mathrm{next}}].
\label{eq:parallel-current-interval}
\end{equation}
Record this interval and slope on the reverse tape. At $s_{\mathrm{next}}$, advance every child whose current piece ends there, update the affected leaf-to-root paths, and continue. Coincident child breakpoints are processed together.

Suppose the sweep produces
\begin{equation}
0=x_0<x_1<\cdots<x_r=1,
\end{equation}
and let $d_j$ be the parent slope on $[x_{j-1},x_j]$. The profile values are recovered backwards from $\widehat{\Psi}_G(1)=1$. Writing $y_j=\widehat{\Psi}_G(x_j)$,
\begin{equation}
y_r=1,\qquad
y_{j-1}=y_j-d_j(x_j-x_{j-1}),
\qquad
j=r,\ldots,1.
\label{eq:parallel-backward-integration}
\end{equation}
Hence, on the $j$th interval,
\begin{equation}
\widehat{\Psi}_G(s)=y_{j-1}+d_j(s-x_{j-1}).
\label{eq:parallel-parent-piece}
\end{equation}

\begin{algorithm}[h!]
\caption{\(n\)-ary parallel profile construction}
\label{alg:nary-parallel}
\begin{algorithmic}[1]
\Require Child profiles of $G=\Par(G_1,\ldots,G_n)$
\Ensure Parent profile and reverse tape
\State For each child $G_i$, activate its first slope $c_i$ and next breakpoint $t_i$
\State Build a balanced segment tree whose leaves store $(c_i,t_i,i)$
\State Store the product in Equation~\ref{eq:parallel-tree-combine} at each internal node
\State Store the earliest breakpoint and all children attaining it
\State $s\gets0$
\While{$s<1$}
    \State Read $(p_{\mathrm{root}},t_{\mathrm{root}})$ from the root
    \State $s'\gets\min\{t_{\mathrm{root}},1\}$
    \State Record $g_G(z)=p_{\mathrm{root}}$ for $z\in[s,s']$
    \State Advance every child satisfying $t_i=s'$
    \State Update the affected leaf-to-root paths and record the operations
    \State $s\gets s'$
\EndWhile
\State Set $\widehat{\Psi}_G(1)\gets1$
\For{the recorded intervals $[u,v]$ in reverse order}
    \State $\widehat{\Psi}_G(u)\gets\widehat{\Psi}_G(v)-g_G(z)(v-u)$
    \State Record this integration operation and its endpoints
\EndFor
\State \Return $\widehat{\Psi}_G$ and the reverse tape
\end{algorithmic}
\end{algorithm}

\paragraph{Reverse differentiation.}
The reverse pass stores, for every recorded scalar $z$, the full derivative
\begin{equation}
\frac{\partial V_{\Network}^*}{\partial z}.
\label{eq:full-derivative-notation}
\end{equation}
A variable may be used by several later operations, so every contribution is added to the derivative already stored for that variable.

For the product operation
\begin{equation}
p=p_Lp_R,
\end{equation}
the reverse updates are
\begin{equation}
\frac{\partial V_{\Network}^*}{\partial p_L}
\mathrel{+}=
p_R
\frac{\partial V_{\Network}^*}{\partial p},
\label{eq:parallel-product-derivative-left}
\end{equation}
\begin{equation}
\frac{\partial V_{\Network}^*}{\partial p_R}
\mathrel{+}=
p_L
\frac{\partial V_{\Network}^*}{\partial p}.
\label{eq:parallel-product-derivative-right}
\end{equation}

For the backward-integration operation
\begin{equation}
y_{j-1}=y_j-d_j(x_j-x_{j-1}),
\end{equation}
the derivative with respect to the right endpoint value is updated by
\begin{equation}
\frac{\partial V_{\Network}^*}{\partial y_j}
\mathrel{+}=
\frac{\partial V_{\Network}^*}{\partial y_{j-1}}.
\label{eq:parallel-integration-derivative-y}
\end{equation}
The derivative with respect to the interval slope is updated by
\begin{equation}
\frac{\partial V_{\Network}^*}{\partial d_j}
\mathrel{+}=
-(x_j-x_{j-1})
\frac{\partial V_{\Network}^*}{\partial y_{j-1}}.
\label{eq:parallel-integration-derivative-d}
\end{equation}
The two endpoint derivatives are
\begin{equation}
\frac{\partial V_{\Network}^*}{\partial x_j}
\mathrel{+}=
-d_j
\frac{\partial V_{\Network}^*}{\partial y_{j-1}},
\label{eq:parallel-integration-derivative-right-endpoint}
\end{equation}
\begin{equation}
\frac{\partial V_{\Network}^*}{\partial x_{j-1}}
\mathrel{+}=
d_j
\frac{\partial V_{\Network}^*}{\partial y_{j-1}}.
\label{eq:parallel-integration-derivative-left-endpoint}
\end{equation}
These endpoint updates account for movements of the parent breakpoints. If a minimum breakpoint is unique, its derivative contribution is passed to the child that supplied it. Coincident breakpoints are handled by the fixed tie rule used for subgradient selection.

\begin{appendixprop}[Parallel-node complexity]
\label{prop:nary-parallel-complexity}
For the complexity analysis, retain the notation $m_G$ for the arity of $G$ and define
\begin{equation}
K_G=\sum_{i=1}^{m_G}M(\widehat{\Psi}_{G_i}),
\label{eq:parallel-KG}
\end{equation}
where $M(\widehat{\Psi}_{G_i})$ is the number of affine pieces in child $G_i$. The forward construction and its reverse pass require
\begin{equation}
\Oh{K_G\log m_G}
\end{equation}
work and tape storage. Since $K_G=\Oh{Q_G}$, this is
\begin{equation}
\Oh{Q_G\log m_G}.
\end{equation}
\end{appendixprop}

\begin{proof}
The tree is built in $\Oh{m_G}$ work. Each child breakpoint is processed once, and advancing one child updates $\Oh{\log m_G}$ tree nodes. Thus the sweep uses $\Oh{K_G\log m_G}$ work. Every update and integration operation creates a constant-sized tape record, and every record is processed once in reverse. Finally, $M(\widehat{\Psi}_{G_i})=\Oh{Q_{G_i}}$ and the child subcomponents are disjoint, so $K_G=\Oh{Q_G}$.
\end{proof}

\subsection{Direct \(n\)-ary series composition}
\label{app:nary-series}

Let
\begin{equation}
G=\Ser(G_1,\ldots,G_n),
\end{equation}
where $G_1$ is the most upstream child and $G_n$ is the most downstream child. Unlike parallel composition, the children are not evaluated at a common input. The parent input is transformed successively by the downstream children before it reaches an upstream child.

\paragraph{Transfer maps.}
For any subcomponent $H$, define
\begin{equation}
\tau_H(s)=\frac{s}{\widehat{\Psi}_H(s)}.
\label{eq:transfer-map}
\end{equation}
The binary series identity is
\begin{equation}
\widehat{\Psi}_{\Ser(G_L,G_R)}(s)
=
\widehat{\Psi}_{G_R}(s)
\widehat{\Psi}_{G_L}\!\left(\tau_{G_R}(s)\right),
\label{eq:binary-series-profile}
\end{equation}
where $G_L$ is upstream and $G_R$ is downstream.

\begin{appendixlemma}[Composition of transfer maps]
\label{lem:transfer-map-composition}
If $G_C=\Ser(G_L,G_R)$, then
\begin{equation}
\tau_{G_C}=\tau_{G_L}\circ\tau_{G_R}.
\label{eq:binary-transfer-composition}
\end{equation}
Consequently,
\begin{equation}
\tau_G=\tau_{G_1}\circ\tau_{G_2}\circ\cdots\circ\tau_{G_n},
\label{eq:nary-transfer-composition}
\end{equation}
with the rightmost map applied first.
\end{appendixlemma}

\begin{proof}
Using Equation~\ref{eq:binary-series-profile},
\begin{equation}
\tau_{G_C}(s)
=
\frac{s}{
\widehat{\Psi}_{G_R}(s)
\widehat{\Psi}_{G_L}(\tau_{G_R}(s))
}
=
\frac{\tau_{G_R}(s)}{
\widehat{\Psi}_{G_L}(\tau_{G_R}(s))
}
=
\tau_{G_L}(\tau_{G_R}(s)).
\end{equation}
Repeated application gives Equation~\ref{eq:nary-transfer-composition}.
\end{proof}

\paragraph{Active-piece summaries.}
Suppose child $G_i$ is currently on the affine piece
\begin{equation}
\widehat{\Psi}_{G_i}(s)=a_is+b_i.
\end{equation}
Its transfer map is
\begin{equation}
\tau_{G_i}(s)=\frac{s}{a_is+b_i}.
\label{eq:leaf-transfer-map}
\end{equation}
This fractional-linear form is closed under composition. Hence a consecutive series block $H$ can be summarised by coefficients $(\alpha_H,\beta_H)$ such that
\begin{equation}
\tau_H(s)=\frac{s}{\alpha_Hs+\beta_H}.
\label{eq:series-summary-form}
\end{equation}

A balanced segment tree stores the ordered children. Every node represents a consecutive series subcomponent and stores:

\begin{enumerate}
\item its transfer coefficients $(\alpha_H,\beta_H)$;
\item its next breakpoint $\eta_H$, measured in the input coordinate of that subcomponent; and
\item every child responsible for that breakpoint.
\end{enumerate}

At a leaf corresponding to $G_i$, store
\begin{equation}
\alpha_{G_i}=a_i,\qquad \beta_{G_i}=b_i,
\end{equation}
and let $\eta_{G_i}$ be the right endpoint of the active child piece.

\begin{appendixlemma}[Combination of series summaries]
\label{lem:series-summary-combine}
Let an internal node combine an upstream block $G_L$ and the immediately following downstream block $G_R$, and let $G_C=\Ser(G_L,G_R)$. Then
\begin{equation}
\alpha_C=\alpha_L+\beta_L\alpha_R,
\qquad
\beta_C=\beta_L\beta_R.
\label{eq:series-coefficient-combine}
\end{equation}
The upstream candidate breakpoint, expressed in the input coordinate of $G_C$, is
\begin{equation}
t_L=\frac{\eta_L\beta_R}{1-\eta_L\alpha_R},
\label{eq:mapped-upstream-breakpoint}
\end{equation}
whereas the downstream candidate is
\begin{equation}
t_R=\eta_R.
\end{equation}
Thus the next breakpoint stored at the combined node is
\begin{equation}
\eta_C=\min\{t_L,t_R\}.
\label{eq:series-next-breakpoint}
\end{equation}
\end{appendixlemma}

\begin{proof}
Substituting the summaries into Equation~\ref{eq:binary-transfer-composition} gives
\begin{equation}
\tau_{G_C}(s)
=
\frac{s}{(\alpha_L+\beta_L\alpha_R)s+\beta_L\beta_R},
\end{equation}
which proves Equation~\ref{eq:series-coefficient-combine}.

A breakpoint in $G_R$ is already expressed in the input coordinate of $G_C$, so its candidate is $\eta_R$. A breakpoint in $G_L$ occurs when the local input to $G_L$ reaches $\eta_L$, that is,
\begin{equation}
\tau_{G_R}(s)=\frac{s}{\alpha_Rs+\beta_R}=\eta_L.
\end{equation}
Solving for $s$ gives Equation~\ref{eq:mapped-upstream-breakpoint}. The denominator is positive on the active domain because the inverse image lies before the end of the current downstream piece. Taking the earlier candidate gives Equation~\ref{eq:series-next-breakpoint}.
\end{proof}

At the root, Equation~\ref{eq:series-summary-form} and $\tau_G(s)=s/\widehat{\Psi}_G(s)$ imply that the active parent piece is
\begin{equation}
\widehat{\Psi}_G(s)=\alpha_{\mathrm{root}}s+\beta_{\mathrm{root}}.
\label{eq:series-parent-piece}
\end{equation}

\paragraph{Forward sweep.}
Begin at parent input $s=0$ with every child on its first affine piece. At each step, the root supplies the coefficients of the active parent piece and the next breakpoint. If the current position is $s_{\mathrm{cur}}$, set
\begin{equation}
s_{\mathrm{next}}=\min\{\eta_{\mathrm{root}},1\}.
\end{equation}
Record the affine piece Equation~\ref{eq:series-parent-piece} on $[s_{\mathrm{cur}},s_{\mathrm{next}}]$. Then advance every child responsible for the event at $s_{\mathrm{next}}$, update the affected leaf-to-root paths, and continue. Coefficient combinations, mapped breakpoints, minimum selections, and emitted profile pieces are all recorded on the reverse tape.

\begin{algorithm}[h!]
\caption{\(n\)-ary series profile construction}
\label{alg:nary-series}
\begin{algorithmic}[1]
\Require Child profiles of $G=\Ser(G_1,\ldots,G_n)$
\Ensure Parent profile and reverse tape
\State Activate the first affine piece $a_is+b_i$ of every child $G_i$
\State Build a balanced segment tree with leaf summaries $(a_i,b_i,\eta_i,i)$
\State Combine coefficients using Equation~\ref{eq:series-coefficient-combine}
\State Map the upstream breakpoint using Equation~\ref{eq:mapped-upstream-breakpoint}
\State Store the minimum in Equation~\ref{eq:series-next-breakpoint}
\State Record every child responsible for the minimum
\State $s\gets0$
\While{$s<1$}
    \State Read $(\alpha_{\mathrm{root}},\beta_{\mathrm{root}},\eta_{\mathrm{root}})$ from the root
    \State $s'\gets\min\{\eta_{\mathrm{root}},1\}$
    \State Record $\widehat{\Psi}_G(z)=\alpha_{\mathrm{root}}z+\beta_{\mathrm{root}}$ for $z\in[s,s']$
    \State Advance every child responsible for the breakpoint at $s'$
    \State Update the affected leaf-to-root paths and record the operations
    \State $s\gets s'$
\EndWhile
\State \Return $\widehat{\Psi}_G$ and the reverse tape
\end{algorithmic}
\end{algorithm}

\paragraph{Reverse differentiation.}
Consider one internal segment-tree node with forward coefficient updates
\begin{equation}
\alpha_C=\alpha_L+\beta_L\alpha_R,
\qquad
\beta_C=\beta_L\beta_R.
\end{equation}
The derivative contribution to $\alpha_L$ is
\begin{equation}
\frac{\partial V_{\Network}^*}{\partial\alpha_L}
\mathrel{+}=
\frac{\partial V_{\Network}^*}{\partial\alpha_C}.
\label{eq:series-derivative-alpha-L}
\end{equation}
The derivative contribution to $\beta_L$ is
\begin{equation}
\frac{\partial V_{\Network}^*}{\partial\beta_L}
\mathrel{+}=
\alpha_R
\frac{\partial V_{\Network}^*}{\partial\alpha_C}
+
\beta_R
\frac{\partial V_{\Network}^*}{\partial\beta_C}.
\label{eq:series-derivative-beta-L}
\end{equation}
The derivative contribution to $\alpha_R$ is
\begin{equation}
\frac{\partial V_{\Network}^*}{\partial\alpha_R}
\mathrel{+}=
\beta_L
\frac{\partial V_{\Network}^*}{\partial\alpha_C}.
\label{eq:series-derivative-alpha-R-coeff}
\end{equation}
The derivative contribution to $\beta_R$ is
\begin{equation}
\frac{\partial V_{\Network}^*}{\partial\beta_R}
\mathrel{+}=
\beta_L
\frac{\partial V_{\Network}^*}{\partial\beta_C}.
\label{eq:series-derivative-beta-R-coeff}
\end{equation}

The mapped upstream breakpoint is
\begin{equation}
t_L=\frac{\eta_L\beta_R}{1-\eta_L\alpha_R}.
\end{equation}
For readability, set
\begin{equation}
d=1-\eta_L\alpha_R.
\end{equation}
The derivative contribution to the upstream breakpoint is
\begin{equation}
\frac{\partial V_{\Network}^*}{\partial\eta_L}
\mathrel{+}=
\frac{\beta_R}{d^2}
\frac{\partial V_{\Network}^*}{\partial t_L}.
\label{eq:series-derivative-eta-L}
\end{equation}
The mapped-breakpoint contribution to $\beta_R$ is
\begin{equation}
\frac{\partial V_{\Network}^*}{\partial\beta_R}
\mathrel{+}=
\frac{\eta_L}{d}
\frac{\partial V_{\Network}^*}{\partial t_L}.
\label{eq:series-derivative-beta-R-breakpoint}
\end{equation}
The mapped-breakpoint contribution to $\alpha_R$ is
\begin{equation}
\frac{\partial V_{\Network}^*}{\partial\alpha_R}
\mathrel{+}=
\frac{\eta_L^2\beta_R}{d^2}
\frac{\partial V_{\Network}^*}{\partial t_L}.
\label{eq:series-derivative-alpha-R-breakpoint}
\end{equation}
The contributions in Equation~\ref{eq:series-derivative-beta-R-coeff} and Equation~\ref{eq:series-derivative-beta-R-breakpoint} are accumulated, as are the contributions in Equation~\ref{eq:series-derivative-alpha-R-coeff} and Equation~\ref{eq:series-derivative-alpha-R-breakpoint}. If the minimum breakpoint is unique, its derivative contribution is passed only to the selected candidate. Coincident candidates are processed according to the fixed tie rule.

\begin{appendixprop}[Series-node complexity]
\label{prop:nary-series-complexity}
Let
\begin{equation}
K_G=\sum_{i=1}^{m_G}M(\widehat{\Psi}_{G_i}).
\label{eq:series-KG}
\end{equation}
The forward construction and reverse differentiation at an \(m_G\)-ary series node require
\begin{equation}
\Oh{K_G\log m_G}=\Oh{Q_G\log m_G}
\end{equation}
work and tape storage.
\end{appendixprop}

\begin{proof}
The local input received by child $G_i$ is
\begin{equation}
x_i(s)=
\left(
\tau_{G_{i+1}}\circ\cdots\circ\tau_{G_{m_G}}
\right)(s),
\qquad
x_{m_G}(s)=s.
\label{eq:series-local-input}
\end{equation}
Every transfer map is continuous and nondecreasing, so $x_i$ is nondecreasing. Hence each local child breakpoint is crossed at most once during the parent sweep. The tree is built in $\Oh{m_G}$ work, and each of the $K_G$ breakpoint events updates $\Oh{\log m_G}$ nodes. Each forward operation creates a constant-sized tape record, and each record is reversed once. Since $K_G=\Oh{Q_G}$, the claimed bounds follow.
\end{proof}

\subsection{Propagation to leaf parameters}
\label{app:local-to-leaf-derivatives}

The preceding rules propagate derivatives from the output profile of an \(n\)-ary node to the slopes, intercepts, and breakpoints of its child profiles. Reversing the child tapes recursively continues this propagation through the parse tree until it reaches the primitive controls.

\paragraph{Reverse accumulation.}
If a recorded variable $z$ is used to compute later variables $w_1,\ldots,w_q$, then the chain rule gives
\begin{equation}
\frac{\partial V_{\Network}^*}{\partial z}
=
\sum_{r=1}^{q}
\frac{\partial V_{\Network}^*}{\partial w_r}
\frac{\partial w_r}{\partial z}.
\label{eq:reverse-accumulation}
\end{equation}
The reverse traversal evaluates this identity by processing the tape in the opposite order from the forward computation and adding every contribution to the derivative already stored for $z$.

At a $\Par$ node,  Equations~\ref{eq:parallel-product-derivative-left} to \ref{eq:parallel-integration-derivative-left-endpoint} propagate derivatives from the parent interval values to its slope and breakpoints and then through the stored product tree to the active child slopes. No explicit leave-one-out products are required.

At a $\Ser$ node, Equations~\ref{eq:series-derivative-alpha-L} to \ref{eq:series-derivative-alpha-R-breakpoint} propagate the derivatives of the active parent piece and its breakpoint to the affine coefficients and local breakpoints of the children. Recursively reversing each child tape eventually reaches a primitive control.

\subsubsection{How an internal node contributes to a leaf derivative}
\label{app:node-to-leaf-chain}

Fix a primitive control $e$. Suppose first that $e$ lies in the subtree of child $G_i$ of the current internal node. The internal-node reverse pass does not differentiate the parent value directly with respect to $\beta_e$. Instead, it performs the following three-stage chain rule:
\begin{equation}
V_{\Network}^*
\longrightarrow
\text{parameters of the parent profile}
\longrightarrow
\text{parameters of the profile of }G_i
\longrightarrow
\beta_e.
\label{eq:node-to-leaf-pipeline}
\end{equation}
The first arrow is supplied by the reverse pass of the component containing the current node. The second arrow is supplied by the parallel or series rules below. The third arrow is obtained by recursively reversing the construction of the child profile. This recursion ends at the leaf rule in Proposition~\ref{prop:leaf-reverse-rule}.

\paragraph{Parallel node.}
Let
\begin{equation}
G=\Par(G_1,\ldots,G_n),
\end{equation}
and suppose $e$ belongs to $G_i$. On a parent interval $[x_{j-1},x_j]$, let $c_{r,j}$ be the active slope of child $G_r$. The parent slope is
\begin{equation}
d_j=\prod_{r=1}^{n}c_{r,j}.
\end{equation}
Once the integration step has produced
\begin{equation}
\frac{\partial V_{\Network}^*}{\partial d_j},
\end{equation}
the product tree passes the corresponding derivative to the active slope of child $G_i$:
\begin{equation}
\frac{\partial V_{\Network}^*}{\partial c_{i,j}}
\mathrel{+}=
\frac{\partial V_{\Network}^*}{\partial d_j}
\prod_{\substack{r=1\\r\neq i}}^{n}c_{r,j}.
\label{eq:parallel-child-slope-derivative}
\end{equation}
The segment tree evaluates Equation~\ref{eq:parallel-child-slope-derivative} by reversing its stored multiplications; it does not explicitly form a separate product for every child.

The interval endpoints $x_{j-1}$ and $x_j$ may be breakpoints inherited from child $G_i$. When this occurs, the endpoint derivatives from Equation~\ref{eq:parallel-integration-derivative-right-endpoint} and Equation~\ref{eq:parallel-integration-derivative-left-endpoint} are routed to the corresponding child breakpoint. Thus the parallel node returns derivatives with respect to two types of child-profile data:
\begin{equation}
\frac{\partial V_{\Network}^*}{\partial c_{i,j}}
\qquad\text{and}\qquad
\frac{\partial V_{\Network}^*}{\partial t_{i,k}},
\label{eq:parallel-child-profile-derivatives}
\end{equation}
where $t_{i,k}$ ranges over the breakpoints of $\widehat{\Psi}_{G_i}$.

Let $\mathcal I_i$ index the affine pieces of $\widehat{\Psi}_{G_i}$ and let $\mathcal B_i$ index its internal breakpoints. The complete dependence of the network value on $\beta_e$ through this child is
\begin{equation}
\frac{\partial V_{\Network}^*}{\partial\beta_e}
=
\sum_{j\in\mathcal I_i}
\frac{\partial V_{\Network}^*}{\partial c_{i,j}}
\frac{\partial c_{i,j}}{\partial\beta_e}
+
\sum_{k\in\mathcal B_i}
\frac{\partial V_{\Network}^*}{\partial t_{i,k}}
\frac{\partial t_{i,k}}{\partial\beta_e}.
\label{eq:parallel-to-beta-chain}
\end{equation}
Equation~\ref{eq:parallel-to-beta-chain} is not evaluated as two explicit sums. Reversing the tape of $G_i$ computes the same chain rule recursively. If $G_i$ is an internal component, its own parallel or series rule propagates the derivatives further. If $G_i=\Ctrl(e)$, the leaf rule converts them into $\partial V_{\Network}^*/\partial\beta_e$.

\paragraph{Series node.}
Now let
\begin{equation}
G=\Ser(G_1,\ldots,G_n),
\end{equation}
and again suppose $e$ belongs to $G_i$. During one parent interval, child $G_i$ contributes an active affine piece
\begin{equation}
\widehat{\Psi}_{G_i}(s)=a_{i,j}s+b_{i,j},
\end{equation}
together with its next local breakpoint $\eta_{i,j}$. The series segment tree combines these quantities through Equation~\ref{eq:series-coefficient-combine} and Equation~\ref{eq:mapped-upstream-breakpoint}. Reversing the tree therefore returns
\begin{equation}
\frac{\partial V_{\Network}^*}{\partial a_{i,j}},
\qquad
\frac{\partial V_{\Network}^*}{\partial b_{i,j}},
\qquad
\frac{\partial V_{\Network}^*}{\partial\eta_{i,j}}.
\label{eq:series-child-profile-derivatives}
\end{equation}

Let $\mathcal I_i$ index the active affine pieces of $\widehat{\Psi}_{G_i}$ encountered during the sweep, and let $\mathcal B_i$ index its local breakpoints. The corresponding chain rule is
\begin{equation}
\frac{\partial V_{\Network}^*}{\partial\beta_e}
=
\sum_{j\in\mathcal I_i}
\frac{\partial V_{\Network}^*}{\partial a_{i,j}}
\frac{\partial a_{i,j}}{\partial\beta_e}
+
\sum_{j\in\mathcal I_i}
\frac{\partial V_{\Network}^*}{\partial b_{i,j}}
\frac{\partial b_{i,j}}{\partial\beta_e}
+
\sum_{k\in\mathcal B_i}
\frac{\partial V_{\Network}^*}{\partial\eta_{i,k}}
\frac{\partial\eta_{i,k}}{\partial\beta_e}.
\label{eq:series-to-beta-chain}
\end{equation}
As in the parallel case, the implementation does not form these sums explicitly. It reverses the tape of $G_i$, which applies the chain rule operation by operation until the derivative reaches the primitive parameter $\beta_e$.

\paragraph{Interpretation.}
Equations \eqref{eq:parallel-to-beta-chain} and \eqref{eq:series-to-beta-chain} express the same principle. A parallel node communicates with a child through the slopes and breakpoints of the child profile. A series node communicates with a child through the slopes, intercepts, and local breakpoints of the child profile. In both cases, the child tape translates derivatives with respect to profile data into derivatives with respect to the primitive controls in that child. Therefore, for a fixed control $e$, only the unique root-to-leaf path containing $e$ contributes to
\begin{equation}
\frac{\partial V_{\Network}^*}{\partial\beta_e}.
\label{eq:unique-root-leaf-derivative-path}
\end{equation}

\paragraph{Leaf construction.}
For a primitive control $e$, let
\begin{equation}
L_h(s)=A_h+B_hs
\end{equation}
be the payoff from allowing at most $h$ attempts before taking outside option $s$. Set
\begin{equation}
A_0=0,\qquad B_0=1,
\end{equation}
and write $p_h=p_e(h)$. Appending attempt $h$ gives
\begin{equation}
A_{h+1}=A_h+\beta_ep_hB_h,
\qquad
B_{h+1}=\beta_e(1-p_h)B_h.
\label{eq:leaf-coefficient-recurrence}
\end{equation}
The calibrated leaf profile is the upper envelope
\begin{equation}
\widehat{\Psi}_e(s)=
\max_{0\leq h\leq q_e}
\{A_h+B_hs\}.
\label{eq:leaf-upper-envelope}
\end{equation}
The leaf tape records the recurrence, the line intersections used to form the upper hull, and the active line on every profile interval. Reversing the hull construction first converts incoming profile derivatives into derivatives of $A_h$ and $B_h$.

\begin{appendixprop}[Leaf reverse rule]
\label{prop:leaf-reverse-rule}
For $h=q_e-1,\ldots,0$, reversing Equation~\ref{eq:leaf-coefficient-recurrence} gives the following updates.

The derivative of $A_h$ receives
\begin{equation}
\frac{\partial V_{\Network}^*}{\partial A_h}
\mathrel{+}=
\frac{\partial V_{\Network}^*}{\partial A_{h+1}}.
\label{eq:leaf-derivative-A}
\end{equation}
The derivative of $B_h$ receives two contributions:
\begin{equation}
\frac{\partial V_{\Network}^*}{\partial B_h}
\mathrel{+}=
\beta_ep_h
\frac{\partial V_{\Network}^*}{\partial A_{h+1}},
\label{eq:leaf-derivative-B-from-A}
\end{equation}
\begin{equation}
\frac{\partial V_{\Network}^*}{\partial B_h}
\mathrel{+}=
\beta_e(1-p_h)
\frac{\partial V_{\Network}^*}{\partial B_{h+1}}.
\label{eq:leaf-derivative-B-from-B}
\end{equation}
The derivative of $\beta_e$ also receives two contributions:
\begin{equation}
\frac{\partial V_{\Network}^*}{\partial\beta_e}
\mathrel{+}=
p_hB_h
\frac{\partial V_{\Network}^*}{\partial A_{h+1}},
\label{eq:leaf-derivative-beta-from-A}
\end{equation}
\begin{equation}
\frac{\partial V_{\Network}^*}{\partial\beta_e}
\mathrel{+}=
(1-p_h)B_h
\frac{\partial V_{\Network}^*}{\partial B_{h+1}}.
\label{eq:leaf-derivative-beta-from-B}
\end{equation}
After all attempts have been reversed, the accumulated quantity is
\begin{equation}
\frac{\partial V_{\Network}^*}{\partial\beta_e}.
\end{equation}
Since $\beta_e=e^{-\lambda\ell_e}$,
\begin{equation}
\frac{\partial V_{\Network}^*}{\partial\ell_e}
=
-\lambda\beta_e
\frac{\partial V_{\Network}^*}{\partial\beta_e}.
\label{eq:leaf-delay-gradient}
\end{equation}
\end{appendixprop}

\begin{proof}
Each update follows directly from the two scalar assignments in Equation~\ref{eq:leaf-coefficient-recurrence}. Processing the attempts in reverse order accumulates every direct and indirect dependence of the leaf profile on $\beta_e$. The final conversion follows from
\begin{equation}
\frac{\partial\beta_e}{\partial\ell_e}=-\lambda\beta_e.
\end{equation}
\end{proof}

Combining the internal-node rules with Proposition~\ref{prop:leaf-reverse-rule} gives the complete derivative path
\begin{equation}
\frac{\partial V_{\Network}^*}{\partial(\text{root-profile data})}
\longrightarrow
\frac{\partial V_{\Network}^*}{\partial(\text{child-profile data})}
\longrightarrow
\frac{\partial V_{\Network}^*}{\partial\beta_e}
\longrightarrow
\frac{\partial V_{\Network}^*}{\partial\ell_e}.
\label{eq:complete-derivative-path}
\end{equation}
At a regular parameter vector, the complete reverse traversal returns the exact defender gradient. At a tie, the same traversal under the fixed deterministic tie rule returns the corresponding limiting subgradient.

\begin{example}
The technical intuition behind the $n$-RT algorithm can be overwhelming, so we give the following toy example of how the algorithm works. Suppose we have the following equations:
\begin{equation*}
    a = 2b \qquad c = a^2 \qquad V=5c
\end{equation*}

We want to compute $\frac{\partial V}{\partial b}$ evaluated at $b = 3$.
We first compute the profile value (forward pass) as such:
\begin{enumerate}
    \item $a = 2b = 2(3) = 6$ (Store input $3$ and local derivative $2$)
    \item $c = a^2 = 6^2 = 36$ (Store input $6$ and local derivative $2a = 12$)
    \item $V = 5c =5(36) = 180$ (Store input $36$ and local derivative $5$)
\end{enumerate} Now we compute the gradients using the stored inputs and local derivatives in the \textbf{reverse way (backward pass)}. Start with $V = 1$.
\begin{enumerate}
    \item $\frac{\partial V}{\partial c} += \mathbf{5}*V = 5(1) = 5$
    \item $\frac{\partial V}{\partial a} += \mathbf{2a}*\frac{\partial V}{\partial c} = (12)(5) = 60 $
    \item $\frac{\partial V}{\partial b} += \mathbf{2} * \frac{\partial V}{\partial a} = 2(60) = 120$
\end{enumerate}
One can verify via simple calculus that $\frac{\partial V}{\partial b} = 120$ when $b=3$. Furthermore, we get $\frac{\partial V}{\partial a} = 60, \frac{\partial V}{\partial c} = 5$ for free as well when $b = 3$.

\end{example}

\subsection{Complete value-and-gradient $n$-RT algorithm}
\label{app:complete-reverse-tape}

The preceding constructions describe how to combine the calibrated profiles of the child subcomponents at an \(n\)-ary $\Par$ or $\Ser$ node. We now assemble these local procedures into a complete algorithm for computing the value of the root component and its derivatives with respect to all defender delays.

\paragraph{Leaf initialisation.}
For each leaf control \(e\), the algorithm first computes the discount factor

\begin{equation}
\beta_e=\exp(-\lambda\ell_e).
\end{equation}

The calibrated leaf profile is then constructed from the finite collection of affine policies corresponding to different stopping times. The upper hull of these affine policies gives the piecewise-linear leaf profile. The recurrence operations used to form the affine policies, the line intersections used to determine the hull breakpoints, and the identity of the line active on each interval are recorded on a local reverse-mode tape.

\paragraph{Bottom-up profile construction.}
After all leaf profiles have been constructed, the internal nodes of the parse tree are processed in postorder, so that the profiles of all child subcomponents are available before their parent is processed. If \(G=\Par(G_1,\cdots,G_n)\), the profile of \(G\) is constructed using Algorithm~\ref{alg:nary-parallel}. If \(G=\Ser(G_1,\cdots,G_n)\), it is constructed using Algorithm~\ref{alg:nary-series}. Each node appends its segment-tree operations, breakpoint computations, and profile-piece computations to the global reverse-mode tape \(\mathcal{A}\).

\paragraph{Root evaluation.}
Let \(\Network\) denote the root component. Once its calibrated profile has been constructed, the attacker value of the complete network is obtained by evaluating the root profile at outside option 0: $V^\star_\Network(\ell)=\widehat{\Psi}_N(0)$. The evaluation operation is appended to the tape, and the derivative of \(V^\star\) with respect to itself is initialised to one.

\paragraph{Reverse traversal.}
The reverse-mode tape is traversed in the opposite order from the forward profile construction. The traversal is initialized by setting the derivative of the root value with respect to itself equal to one. At each recorded operation, the corresponding local chain-rule derivatives are applied, and all contributions to a common input variable are accumulated. The reverse pass first propagates the sensitivity of the root value through the profile of the root component, then through the profiles of its child subcomponents, and continues recursively until it reaches the leaf-profile constructions. As shown in Subsection~\ref{app:local-to-leaf-derivatives}, reversing the tape associated with leaf control $e$ accumulates all paths through which its discount factor affects the root value. The derivative stored at the end of this leaf tape is therefore $\partial V^\star_\Network/\partial\beta_e$. After which optaining the full subgradient is easy. The full procedure is summarised in Algorithm~\ref{alg:complete-value-gradient}.

\begin{algorithm}
\caption{$n$-RT}
\label{alg:complete-value-gradient}
\begin{algorithmic}[1]
\Require Parse tree with root \(\Network\), defender delays \((\ell_e)_{e\in \mathcal{E}}\), and discount rate \(\lambda\)
\Ensure Network value \(V^\star\) and derivatives \((\partial V^\star_\Network/\partial\ell_e)_{e\in \mathcal{E}}\)
\State Initialize an empty reverse-mode tape \(\mathcal{A}\)
\For{each leaf control \(e\in \mathcal{E}\)}
    \State Compute \(\beta_e\gets\exp(-\lambda\ell_e)\)
    \State Construct \(\widehat{\Psi}_{\Ctrl(e)}\) and append the leaf-profile operations to \(\mathcal{A}\)
\EndFor
\For{each internal node \(G\) in postorder}
    \If{\(G=\Par(G_1,\cdots,G_n)\)}
        \State Construct \(\widehat{\Psi}_G\) using Algorithm~\ref{alg:nary-parallel}
    \Else
        \State Construct \(\widehat{\Psi}_G\) using Algorithm~\ref{alg:nary-series}
    \EndIf
    \State Append the node-level operations to \(\mathcal{A}\)
\EndFor
\State Evaluate \(V^\star\gets\widehat{\Psi}_\Network(0)\) and set \(\partial V^\star_\Network/\partial V^\star_\Network\gets 1\)
\For{each operation in \(\mathcal{A}\) in reverse order}
    \State Apply its local chain-rule derivatives and accumulate the resulting contributions
\EndFor
\For{each leaf control \(e\in \mathcal{E}\)}
    \State Set \(\partial V^\star_\Network/\partial\ell_e\gets-\lambda\beta_e\,\partial V^\star_\Network/\partial\beta_e\)
\EndFor
\State \Return \(V^\star\) and \(g(\ell) = (\partial V^\star_\Network/\partial\ell_e)_{e\in \mathcal{E}}\)
\end{algorithmic}
\end{algorithm}

At parameter values for which all active affine pieces and breakpoint orderings are locally unique, Algorithm~\ref{alg:complete-value-gradient} returns the exact gradient. The treatment of coincident breakpoints and non-differentiable parameter values is given in the subsequent subsection.

\subsection{$n$-RT Complexity Analysis}
\label{app:whole-tree-complexity}

The parallel and series constructions have the same asymptotic cost at an internal component. Let $\mathcal{I}$ denote the set of internal components in the parse tree. For each $G\in\mathcal{I}$, let $m_G$ be the number of child subcomponents of $G$, let $E(G)$ be the set of leaf controls below $G$, and let $q_e$ be the lockout limit of leaf control $e$. Define the subtree lockout budget

\begin{equation}
Q_G=\sum_{e\in E(G)}q_e.
\end{equation}

If $M(\widehat{\Psi}_e)$ denotes the number of pieces in the calibrated profile of leaf $e$, then $M(\widehat{\Psi}_e)\leq q_e+1$. More generally, the profile-size bound gives $M(\widehat{\Psi}_G)\leq Q_G+1$ for every component $G$.

The number of events processed while constructing the profile of $G$ is at most proportional to $Q_G$. Each event performs a constant number of queries and updates in a balanced tree of size $m_G$, and therefore costs $\Oh{\log m_G}$. Consequently, the forward construction and reverse traversal at $G$ each require \(\Oh{Q_G\log m_G}\) work. Summing over all internal components gives the whole-tree complexity \(\Oh{\sum_{G\in\mathcal{I}}Q_G\log m_G}.\)

To express this bound in terms of root-to-leaf structure, write $G \succ e$ to denote that $G$ is an ancestor of $e$, including its parent and the root component. Define the log-arity depth of the parse tree by

\begin{equation}
\delta=\max_{e\in \mathcal{E}}\sum_{G\succ e}\log m_G.
\end{equation}

Since each leaf-profile piece contributes to $Q_G$ precisely when $G$ is an ancestor of that leaf, the node-wise costs can be rearranged as

\begin{equation}
\sum_{G\in\mathcal{I}}Q_G\log m_G=\sum_{e\in \mathcal{E}}q_e\sum_{G\succ e}\log m_G.
\end{equation}

Letting $Q=\sum_{e\in \mathcal{E}}q_e$ denote the total number of leaf-profile pieces, it follows that

\begin{equation}
\sum_{G\in\mathcal{I}}Q_G\log m_G\leq Q\delta.
\end{equation}

The complete value-and-gradient algorithm therefore requires \(\Oh{Q\delta}\) work for both the bottom-up profile construction and the reverse-mode derivative computation. The reverse pass has the same asymptotic complexity as the forward pass because each recorded operation is processed once and applies only a constant number of local derivative updates.

If the parse tree has depth at most $d$ and every internal component has arity at most $M$, then $\delta\leq d\log M$, giving the simpler bound \(\Oh{Qd\log M}\). In particular, a single flat component with $Q$ leaf-profile pieces and arity proportional to $Q$ is processed in $\Oh{Q\log Q}$ time.

\subsection{Ties and subgradients}
\label{app:ties-and-subgradients}

The preceding derivative formulas assume that the discrete decisions made during the forward profile construction remain unchanged under a sufficiently small perturbation of the defender delays. This condition may fail when several affine policies are simultaneously optimal, several breakpoint events coincide, or a minimum operation has more than one owner. This subsection explains how such ties are handled and why the resulting reverse pass returns a valid subgradient.

\paragraph{Regular parameter vectors.}
Call a defender-delay vector regular if the active affine policy at every leaf is locally unique, the ordering of all copied and mapped breakpoints is locally unchanged, and every minimum event has a unique owner. On a neighbourhood of a regular vector, the same affine pieces, breakpoint orderings, and event owners are selected. The recorded profile construction is therefore an ordinary differentiable arithmetic circuit, and reversing its tape returns the exact gradient of $V^\star_\Network$.

\paragraph{Simultaneous breakpoint events.}
At a parallel or series component, several children may reach their next breakpoint at the same outside-option value. All such events must be processed as one batch. The parent profile is terminated at the common breakpoint, every child involved in the event is advanced to its next active piece, and all affected paths in the balanced event tree are updated before the next parent interval is constructed. This prevents the creation of artificial zero-width intervals and ensures that the forward profile remains independent of the order in which equal events are detected.

\paragraph{No-nunique selections.}
A tie may also occur when two affine policies attain the same value at a leaf or when two candidate event locations are equal in a segment-tree minimum. At a strict comparison, the reverse derivative follows the uniquely selected argument. At a tie, the implementation uses a fixed symbolic or lexicographic perturbation rule to select one of the neighbouring regular regions. The same rule must be used consistently throughout the forward construction and the reverse traversal so that the tape represents the limiting computation associated with that region.

\paragraph{Limiting subgradient.}
Let $\ell$ be a tied parameter vector, and let $(\ell^{(r)})_{r\geq 1}$ be a sequence of regular parameter vectors approaching $\ell$ while preserving the selections specified by the perturbation rule. The reverse pass at $\ell^{(r)}$ returns the exact gradient $\nabla V^\star_\Network(\ell^{(r)})$. The vector returned at the tie is the corresponding limiting gradient

\begin{equation}
g(\ell)=\lim_{r\to\infty}\nabla V^\star_\Network(\ell^{(r)}).
\end{equation}

Because $V^\star_\Network$ is finite and convex in the defender delays, every such limiting gradient satisfies

\begin{equation}
V^\star_\Network(y)\geq V^\star_\Network(\ell)+g^\top(y-\ell)
\end{equation}

for every feasible delay vector $y$. Hence $g(\ell)\in\partial V^\star_\Network(\ell)$, so the reverse-tape algorithm returns a valid subgradient even when the objective is not differentiable at $\ell$.

\paragraph{Numerical event grouping.}
In a floating-point implementation, a numerical tolerance may be used to decide whether breakpoint events should be processed in the same batch. The tolerance should be used only to group values that are numerically indistinguishable; it should not reverse the mathematical ordering of visibly separated breakpoints. Within a batch, the fixed perturbation rule determines the derivative selection and makes the computed subgradient reproducible.

\section{Omitted Proof of Theorem 2}
Recall that in our model $\Par(G_1, G_2, \dots G_n)$ means that as long as \textit{any} of $G_1$ is compromised, then the entire sub-component is as well. Another natural interpretation is that \textit{all} of $G_1,\dots,G_n$ need to be compromised. This would be similar to our $\Ser$ formulation, but does \textit{not} impose sequential constraints (and is hence commutative). Let us distinguish these two notions of ``parallel'' by $\ParOR$ and $\ParAND$ respectively. We have the following negative result.
\begin{restatetheorem}{2}[Informal] 
Extending the grammar of our SP-networks to allow both $\ParOR$ and $\ParAND$ makes finding the optimal solution $\pi^*$ for arbitrary $\ell$ weakly NP-hard. 
\end{restatetheorem}
 This is an unfortunate result, since $\ParAND$ is commonly seen in real attack graphs. Nonetheless, it can be viewed both a boon and bane to the defender. While this makes finding the true optimal defence difficult to compute, such networks are also \textit{computationally} difficult for the attacker to exploit. Finding fixed-parameter tractable or constant-factor approximations is left for future work. We give the formal writing of the theorem in the next subsection.

\subsection{Weak NP-Hardness for Mixed $\ParOR$--$\ParAND$ Networks}
\label{appendix:mixed-hardness}

\begin{appendixtheorem}[(Formal) Weak NP-hardness of mixed parallel semantics]
\label{thm:mixed-and-or-hardness}
Consider the value-decision problem that asks, for a rationally parameterised attack network $\Network$ and a rational threshold $\Theta$, whether $V_N^*\geq\Theta$. Extending the SP grammar to permit both $\ParOR$ and $\ParAND$ makes this problem weakly NP-hard. The hardness already holds for networks with one top-level $\ParOR$ node whose non-trivial children are pairwise disjoint two-control $\ParAND$ bundles (a bundle is a set of controls in which all controls in the set must be breached to compromise the bundles). Each bundle control has at most two unit-duration attempts, and its second-attempt success probability is strictly smaller than its first-attempt success probability. The construction uses only one additional one-shot pivot child and one deterministic fallback child. Consequently, exact computation of the optimal value is weakly NP-hard.
\end{appendixtheorem}

The proof is a reduction from \textsc{Equal-Cardinality Partition}. Its main difficulty is not the numerical encoding but the adaptive policy space: after a failed attempt, the attacker may leave a bundle, attack elsewhere, and later return. We first build a robust bundle gadget whose residual states are separated by continuation thresholds. We then prove an atomisation lemma showing that some optimal adaptive policy treats each bundle as one of two indivisible modes. The remaining optimisation is a finite ordering problem over affine maps, where a product construction encodes the partition instance.

\subsection{Restricted instances and the source problem}
\label{appendix:H-intro}

Fix the per-unit-time discount factor to
\begin{equation}
\beta=e^{-\lambda}=\frac12,
\end{equation}
so $\lambda=\log 2$. The constructed network has one top-level $\ParOR$ node. It contains $2k$ bundle children, one pivot child, and one fallback child. Bundle $B_i$ is a $\ParAND$ of two prerequisite controls $U_i$ and $V_i$. The top-level node completes when any child completes, whereas a bundle completes only when both of its controls have been breached.

\begin{figure}[h!]
\centering
\begin{tikzpicture}[>=Stealth,
  term/.style={circle,draw,fill=black,inner sep=1.1pt},
  lbl/.style={font=\small,inner sep=1.5pt}]
\begin{scope}
  \node[term] (s1) at (0,0) {};
  \node[term] (t1) at (2.2,0) {};
  \draw[->] (s1) to[bend left=30] node[lbl,above]{$G_1$} (t1);
  \draw[->] (s1) to[bend right=30] node[lbl,below]{$G_2$} (t1);
  \node[font=\small] at (1.1,-1.15) {(a) $\ParOR(G_1,G_2)$};
\end{scope}
\begin{scope}[xshift=3.4cm]
  \node[term] (s2) at (0,0) {};
  \node[term] (t2) at (2.2,0) {};
  \draw[->] (s2) to[bend left=30] node[lbl,above]{$G_1$} (t2);
  \draw[->] (s2) to[bend right=30] node[lbl,below]{$G_2$} (t2);
  \draw[thick] ([shift=(148:0.70)]t2) arc[start angle=148,end angle=212,radius=0.70];
  \node[font=\footnotesize] at ($(t2)+(-0.95,0)$) {$\wedge$};
  \node[font=\small] at (1.1,-1.15) {(b) $\ParAND(G_1,G_2)$};
\end{scope}
\end{tikzpicture}
\caption{The two parallel constructors on a two-terminal block. Both branches share a source and a sink. In (a) the block is cleared as soon as \textit{either} branch is cleared; in (b) the arc marks the merge as conjunctive, so the block is cleared only once \textit{both} branches are cleared.}
\label{fig:sp-and-or}
\end{figure}
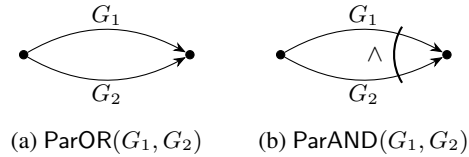

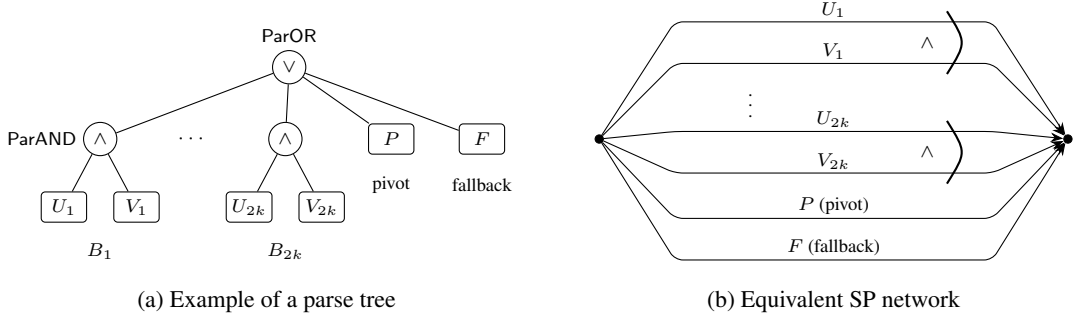
\begin{figure*}[h!]
\centering
\begin{tabular}{@{}c@{\hspace{2.5em}}c@{}}
\begin{tikzpicture}[>=Stealth,
  gate/.style={circle,draw,inner sep=0pt,minimum size=4.4mm,font=\scriptsize},
  leaf/.style={draw,rounded corners=1.5pt,inner xsep=2pt,inner ysep=1.2pt,minimum width=6mm,minimum height=4mm,font=\scriptsize},
  every path/.style={thin}]
  \node[gate] (root) at (2.50,1.85) {$\vee$};
  \node[font=\scriptsize,anchor=south] at ($(root)+(0,0.2)$) {$\ParOR$};
  \foreach \x/\n/\lab in {0/a/1, 2.45/c/{2k}}{
    \node[gate] (\n) at (\x,0.9) {$\wedge$};
    \node[leaf] (\n u) at (\x-0.48,0) {$U_{\lab}$};
    \node[leaf] (\n v) at (\x+0.48,0) {$V_{\lab}$};
    \draw (\n) -- (\n u); \draw (\n) -- (\n v);
    \draw (root) -- (\n);
    \node[font=\scriptsize] at (\x,-0.58) {$B_{\lab}$};
  }
  \node[font=\scriptsize,anchor=east] at ($(a)+(-0.18,0)$) {$\ParAND$};
  \node[font=\scriptsize] at (1.225,0.9) {$\cdots$};
  \node[leaf] (pv) at (3.85,0.9) {$P$};
  \node[leaf] (fb) at (5.05,0.9) {$F$};
  \draw (root) -- (pv); \draw (root) -- (fb);
  \node[font=\scriptsize] at (3.85,0.3) {pivot};
  \node[font=\scriptsize] at (5.05,0.3) {fallback};
  \node[font=\small] at (2.20,-1.25) {(a) Example of a parse tree};
\end{tikzpicture}
&
\begin{tikzpicture}[>=Stealth,font=\scriptsize,
  term/.style={circle,draw,fill=black,inner sep=1.1pt},
  ed/.style={->,rounded corners=3pt,thin},
  elab/.style={above,font=\scriptsize,inner sep=1.5pt}]
  \node[term] (s) at (0,0) {};
  \node[term] (t) at (6.2,0) {};
  \foreach \y/\lab in {1.55/{$U_1$}, 1.00/{$V_1$}, 0.10/{$U_{2k}$}, -0.45/{$V_{2k}$},
                       -1.05/{$P$ (pivot)}, -1.60/{$F$ (fallback)}}{
    \draw[ed] (s) -- (1.0,\y) -- node[elab]{\lab} (5.2,\y) -- (t);}
  \foreach \ya/\yb in {1.55/1.00, 0.10/-0.45}{
    \draw[thick] (4.60,\ya+0.14) .. controls (4.86,{(\ya+\yb)/2}) .. (4.60,\yb-0.14);
    \node at (4.34,{(\ya+\yb)/2}) {$\wedge$};}
  \node at (2.0,0.55) {$\vdots$};
  \node[font=\small] at (3.1,-2.15) {(b) Equivalent SP network};
\end{tikzpicture}
\end{tabular}
\caption{Topology of a constructed instance, as an AND/OR parse tree (a) and as the two-terminal SP network it denotes (b). A single $\ParOR$ root races the $2k$ pairwise disjoint bundle children $B_1,\ldots,B_{2k}$ against the one-shot pivot child $P$ and the deterministic fallback child $F$. The run ends as soon as any child is cleared, and all children share the source and sink of the whole network. Each bundle is a $\ParAND$ of two prerequisite controls $U_i,V_i$, marked in (b) by an arc across the two conjunctive branches, and is cleared only when both controls are breached. Each bundle prerequisite admits at most two attempts. The pivot and the fallback comprise single controls, $P=\Ctrl(e_{\mathrm{piv}})$ and $F=\Ctrl(e_{\mathrm{fb}})$; their parameters are fixed later in the construction.}
\label{fig:hard-instance}
\end{figure*}

We reduce from \textsc{Equal-Cardinality Partition}. An instance consists of $2k$ positive integers $w_1,\ldots,w_{2k}$ satisfying
\begin{equation}
\sum_{i=1}^{2k}w_i=2Z,
\end{equation}
and asks whether there is a subset $S$ of the $2k$ positive integers such that
\begin{equation}
|S|=k,\qquad \sum_{i\in S}w_i=Z.
\end{equation}
This problem is NP-hard. Indeed, from a \textsc{Partition} instance $a_1,\ldots,a_n$ with total sum $2Z$, choose $C>2Z$ and form the $2n$ numbers $C+a_1,\ldots,C+a_n$ together with $n$ copies of $C$. A set of exactly $n$ new numbers has sum $nC+Z$ if and only if the selected original increments have sum $Z$.

We may assume without loss of generality that $k \geq k_0$ for a sufficiently large absolute constant $k_0$. If the original instance already satisfies $k \geq k_0$, no modification is needed. Otherwise, let $r \eqdef k_0 - k > 0$, and choose $C > 2Z$. Furthermore, we add $2r$ copies of $C$. Replace the target cardinality $k$ by $k + r = k_0$ and target sum $Z$ to $Z +rC$. Any feasible solution to the new instance must select exactly $r$ of the added copies as selecting fewer would require the original items to contribute more than $2Z$, while selecting more would make their required contribution negative. Hence, the new instance is feasible if and only if the original instance contains exactly $k$ items summing to $Z$.

\subsection{Proof architecture}
\label{appendix:H-overview}
The entire proof is really trying to do one thing. We want to turn each \textsc{Equal-Cardinality Partition} item $w_i$ into one attacker choice between 2 behaviours (in a bundle, either attack the first or the second control first. We have a name for each behaviour (or mode) to be defined later), and then engineer the network so that choosing which $k$ bundles fall under each behaviour encodes whether those $k$ weights sum to $Z$.

\paragraph{Step 1: Comparing Pieces of Attacker Behaviour.}
Consider a finite sequence of actions that operate only within one child of the top-level $\ParOR$ node and then either completes that child or switches away from it. We call such a sequence a \textit{policy fragment}. If switching away leads to a continuation of value $W$, each fragment induces an affine continuation $f(W) =\mu_f + \nu_fW$, where $\mu_f$ is its expected discounted contribution when it complete its child, while $\nu_f W$ is the discounted value fo continuing elsewhere when it does not. Its continuation threshold $\chi(f) \eqdef \frac{\mu_f}{1 - \nu_f}$ is the value of $W$ at which executing the fragment and skipping directly to $W$ are equally good. \Cref{appendix:affine-fragments} shows that fragments acting on disjoint children can be ordered by decreasing $\chi$, independently of the eventual continuation value.

\paragraph{Step 2: One 2-mode bundle for each source weight $w_i$}
For each integer $w_i$, construct a bundle $B_i = \ParAND(U_i, V_i)$, so completing $B_i$ requires breaching both controls. A fresh bundle has two intended behaviours, but we will call them \textit{modes} from now on: $H_i$ and $L_i$ modes. $H_i$ mode is when the attacker attacks $U_i$ first, whereas $L_i$ attacks $V_i$ first. After this first attempt, the unfinished bundle is in a \textit{residual state}, meaning the local state that remains after part of the bundle has been attempt. We design a \textit{gadget} so that a failed first attempt leaves a \textit{low residual state}, whereas a successful first attempt leaves a \textit{high residual state}, whose remaining actions have thresholds above all fresh bundles and auxiliary children. Thus failure makes the bundle unattractive to re-visit, while success makes finishing it immediately attractive.

\paragraph{Step 3: Eliminating arbitrary switching}
Steps 1 and 2 can already let us simplify the attacker's adaptive behaviour. After a \textit{successful} first attempt (bundle in a \textit{high residual state}), the attacker continues attacking the remaining prerequisite control until it either succeeds or exhausts all available attempts before moving to another child. If the first attempt in a fresh bundle fails (bundle in a \textit{low residual state}), the attacker permanently gives up on that bundle. It makes no further attempts on either prerequisite control and instead continues with other bundles, pivots, or fallback. Therefore some optimal policy treats every fresh bundle as \textbf{one indivisible choice}: execute either $H_i$ or $L_i$ mode, and then move on. We call this the \textit{macro} problem: instead of reasoning about individual attempts and arbitrary returns to partially completed bundles, we reason only about an ordered sequence of the affine maps \textit{due to} the mode the attacker adopts for each bundle. This is the \textit{atomisation} result proved in \Cref{appendix:atomisation}.

\paragraph{Step 4: Encoding the partition choice} Recall that each bundle mode has the form $f_i(W) = \mu_i + \nu_iW$. The coefficient $\nu_i$ determines how strongly that later continuation $W$ contributes. If two modes are executed consecutively, their maps are substituted into one another and their continuation coefficients multiply. For e.g. if $f_i(W) = \mu_i + \nu_iW$ for $i \in \{1, 2\}$, then $f_1(f_2(W)) = \mu_1 + \nu_1\mu_2 + \boxed{\nu_1\nu_2W}$. We exploit this multiplication to encode the weight integers $w_i$. For each $w_i$, set $x_i \eqdef 1 + \epsilon w_i$ and choose the two modes of bundle $i$ so that their continuation coefficients contain $x_i$. Thus after several bundles are executed in the same mode, the attack payoff contains a product of the corresponding $x_i's$. We then add a \textbf{pivot} $P$, consisting of a control with one available attempt, and a deterministic \textbf{fallback} $F$. Their thresholds are chosen so that the every bundle in $H_i$ mode will be attacked before $P$ (pivot), followed by the every bundle in $L_i$ mode, and then the fallback $F$ last. Hence, an optimal policy has the form:
\[
(H_i)_{i \in S} \longrightarrow P \longrightarrow (L_i)_{i \notin S} \longrightarrow F
\]
The remaining choice is therefore the set $S$ of bundles assigned to $H$-mode. The construction first forces $|S| = k$, and then makes the attack payoff largest \textit{precisely} when. theweights $w_i$ indexed by $S$ sum to $Z$.

\paragraph{Step 5: Obtaining a polynomial-size rational instance.} The affine maps used above are first realised by nearby real-valued success probabilities. To obtain a valid complexity-theoretic reduction, these probabilities must have finite polynomial-size encodings. We therefore approximate them by polynomial-bit rational probabilities. The local threshold inequalities are robust to sufficiently small changes, and a coupling argument bounds the resulting change in the optimal value below the YES--NO gap. Hence, the rationalised network preserves the answer to the \textsc{Equal-Cardinality Parition} instance.

\subsection{Step 1: Affine fragments and exchange rules}
\label{appendix:affine-fragments}

Recall from \Cref{appendix:H-overview} Step 1 that a policy \textit{fragment} is a finite policy segment that operates only inside one child of the top-level $\ParOR$ node. It ends either when that child completes or when the policy switches away. For a realised trajectory $\zeta$ through a fragment $f$, let $t(\zeta)$ be its duration and define
\begin{equation}
\mu_f=\mathbb{E}\!\left[\beta^{t(\zeta)}\mathbf{1}\{f\text{ completes its child}\}\right],
\end{equation}
\begin{equation}
\nu_f=\mathbb{E}\!\left[\beta^{t(\zeta)}\mathbf{1}\{f\text{ exits without completion}\}\right].
\end{equation}
If every non-completion exit is followed by a continuation of value $W$, then the fragment has value
\begin{equation}
f(W)=\mu_f+\nu_fW.
\label{eq:fragment-map}
\end{equation}
Every nontrivial fragment consumes positive time, so $\nu_f<1$. Its \emph{continuation threshold} is
\begin{equation}
\chi(f)=\frac{\mu_f}{1-\nu_f}.
\label{eq:fragment-threshold}
\end{equation}
The identity
\begin{equation}
f(W)-W=(1-\nu_f)\bigl(\chi(f)-W\bigr)
\label{eq:fragment-vs-continuation}
\end{equation}
shows that executing $f$ is preferable to taking $W$ immediately exactly when $W<\chi(f)$.

\begin{appendixlemma}[First Exchange Rule]
\label{lem:exchange-rules-1}
Consider policy fragments acting. on different children of the top-level $\ParOR$ node, so executing one policy fragment does not change the local state of the other child.

First, let \begin{equation}
    f_i(W) = \mu_i + \nu_i W, \qquad i \in \{1,2\}
\end{equation} be two fragments, where $W$ is the continuation value obtained after the policy fragment terminates without completing its child. Then we have \begin{equation}\label{eq:two-fragment-exchange}
    f_1(f_2(W)) - f_2(f_1(W)) = (1 - \nu_1)(1 - \nu_2)(\chi(f_1) - \chi(f_2))
\end{equation}. Hence, the better order is determined only by their continuation thresholds. A policy fragment with larger $\chi$ mayber be executed before one with smaller $\chi$, independently of the continuation value $W$.
\end{appendixlemma}
\begin{proof}
    Expanding the two compositions gives
\begin{equation}
f_1(f_2(W))=\mu_1+\nu_1\mu_2+\nu_1\nu_2W,
\end{equation}
\begin{equation}
f_2(f_1(W))=\mu_2+\nu_2\mu_1+\nu_1\nu_2W.
\end{equation}
Subtracting and using $\mu_i=(1-\nu_i)\chi(f_i)$ (by definition of $\chi$) yields Equation~\ref{eq:two-fragment-exchange}.
\end{proof}

\begin{appendixlemma}[Second Exchange Rule]\label{lem:exchange-rules-2}
Consider a one-attempt action $a(W) = \mu_a + \nu_aW$ whose success completes its child, and a policy fragment $f$ acting on another child. If $f$ exits without completing its child, different outcomes due to $f$ may leave that child in different \textbf{residual states}, meaning different unfinished local states at the moment the attacker switches away. Let $\Omega$ be the set of these possible residual states. For each residual state $\omega \in \Omega$, let $\nu_\omega$ be the discounted probability weight that $f$ exits in state $\omega$. Since these possibilities partition all non-completion termination of $f$, we have \begin{equation}
    \nu_f = \sum_{\omega \in \Omega} \nu_\omega
\end{equation}
Let $W_\omega$ denote the continuation value after $f$ terminates in state $\omega$ and a subsequent attempt of $a$ fails. Then we have: 
\begin{equation}
\operatorname{Val}(a\text{ then }f)-\operatorname{Val}(f\text{ then }a)=(1-\nu_a)(1-\nu_f)\bigl(\chi(a)-\chi(f)\bigr).
\label{eq:multi-exit-exchange}
\end{equation}
Thus, even when different terminations of $f$ lead to different continuation values $W_\omega$, the preferred order of $a$ and $f$ still depends only on their continuation threshold $\chi$.
\end{appendixlemma}

\begin{proof}

Suppose $a$ is the policy fragment that is executed first. If it does not breach its child, the attacker executes $f$. Hence
\begin{equation}
    \operatorname{Val}(a \,\text{then}\, f) = \mu_a + \nu_a\left(\mu_f + \sum_{\omega \in \Omega} \nu_\omega W_\omega\right)
\end{equation}

On the other hand, if policy fragment $f$ is executed first, then whenever it exits in a residual state $\omega$, the attacker executes $a$. If that attempt also fails, the subsequent continuation has value $W_\omega$. Therefore:

\begin{equation}
    \operatorname{Val}(f \,\text{then}\, a) = \mu_f + \sum_{\omega \in \Omega} \nu_\omega (\mu_a + \nu_aW_\omega)
\end{equation}

We then arrive at:
\begin{align}
    \operatorname{Val}(a \,\text{then}\, f) -  \operatorname{Val}(f \,\text{then}\, a)&= \mu_a + \nu_a\left(\mu_f + \sum_{\omega \in \Omega} \nu_\omega W_\omega\right) - \left(\mu_f + \sum_{\omega \in \Omega} \nu_\omega (\mu_a + \nu_aW_\omega)\right) \\
    &=\mu_a + \nu_a \mu_f + \nu_a \sum_{\omega \in \Omega} \nu_\omega W_\omega - \mu_f - \mu_a \sum_{\omega\in \Omega} \nu_\omega -  \nu_a \sum_{\omega \in \Omega} \nu_\omega W_\omega\\
    &=\mu_a + \nu_a\mu_f - \mu_f - \mu_a \sum_{\omega \in \Omega} \nu_\omega\\
    &= \mu_a(1 - \nu_f) - \mu_f(1 - \nu_a) & \because\sum_{\omega \in \Omega}\nu_\omega = \nu_f
\end{align}

Finally, substituting $\mu_a = (1 - \nu_a) \chi(a)$ and $\mu_f = (1 - \nu_f)\chi(f)$ yields Equation~\ref{eq:multi-exit-exchange}.
\end{proof}

\subsection{Step 2.1: The Two-Mode Bundle}
\label{appendix:two-mode-bundle}
For each source weight $w$, we will eventually construct one bundle $B = \ParAND (U, V)$. Breaching $B$ requires breaching both controls $U$ and $V$. Each control has two available attempts, each taking one unit of time ($\ell_U = \ell_V = 1$). Recall that the discount factor remains as $\beta = \frac{1}{2}$. Let $u_1 > u_2$ be the success probabilities of the first and second attempts on $U$. Define $v_1 > v_2$ analogously for $V$.

First, consider attacking $U$ repeatedly until it succeeds or both attempts are exhausted. Let $\mu_U$ denote the expected discounted contribution from the cases in which $U$ succeeds, and let $\nu_V$ denote the discounted probability weight of exhausting $U$ without success, then then have:
\begin{align}
    \mu_U &= \frac{1}{2}u_1 + \frac{1}{4}(1 - u_1)u_2\\
    \nu_U &= \frac{1}{4}(1 - u_1)(1 - u_2)
\end{align} and define $\mu_V, \nu_V$ analagously. Thus, if failure to breach $U$ is followed by continuation value $W$, this two-attempt policy fragment has value $\mu_U + \nu_U W$.

We define (again) two ways of attacking a fresh bundle. In $\mathbf H$-\textbf{mode}, the attacker attacks $U$ once first. If this attempt fails, it abandons the bundle and continues elsewhere with value $W$. If the attack is. a success, the attacker has breached $U$ and attacks $V$ until $V$ succeeds or exhausts its attempts. Hence

\begin{equation}\label{eq:H-map}
    H(W)  = \frac{1}{2}u_1(\mu_V + \nu_VW) + \frac{1}{2}(1 - u_1)W = \mu_H + \nu_HW
\end{equation} where $\mu_H = \frac{1}{2}u_1\mu_V$ and $\nu_H = \frac{1}{2}(1 - u_1) + \frac{1}{2}u_1\nu_V$. In $\mathbf{L}$\textbf{-mode}, the roles of $U$ and $V$ are reversed:

\begin{equation}\label{eq:L-map}
    L(W) = \frac{1}{2}v_1(\mu_U + \nu_U W) + \frac{1}{2}(1 - v_1)W = \mu_L + \nu_LW
\end{equation} where $\mu_L = \frac{1}{2}v_1 \mu_U$ and $\nu_L = \frac{1}{2}(1 - v_1) + \frac{1}{2}v_1\nu_U$.

\begin{figure}[h!]
\centering
\begin{tikzpicture}[
  >={Stealth[length=3.5pt]}, font=\scriptsize,
  box/.style={draw,rounded corners=2pt,align=center,inner sep=3pt,minimum height=7.5mm},
  fin/.style={box,double,double distance=0.6pt},
  esc/.style={box,dashed},
  mode/.style={font=\scriptsize\bfseries,anchor=west},
  every path/.style={thin}
]
\node[box] (hu) at (0,0)    {Attack $U$};
\node[box] (hv) at (2.75,0) {Attack $V$ to\\success or\\exhaustion};
\node[fin] (hc) at (5.3,0)  {bundle\\completes};
\node[esc] (he) at (2.75,-1.75) {exit to\\continuation $W$};
\node[mode] at (-0.85,0.95) {$H$-mode};
\draw[->] (hu) -- node[above]{success} (hv);
\draw[->] (hv) -- (hc);
\draw[->] (hu.south) |- node[below,pos=0.75]{failure} (he.west);
\draw[->] (hv.south) -- node[right]{exhaustion} (he.north);

\node[box] (lv) at (0,-3.5)    {Attack $V$};
\node[box] (lu) at (2.75,-3.5) {Attack $U$ to\\success or\\exhaustion};
\node[fin] (lc) at (5.3,-3.5)  {bundle\\completes};
\node[esc] (le) at (2.75,-5.25) {exit to\\continuation $W$};
\node[mode] at (-0.85,-2.55) {$L$-mode};
\draw[->] (lv) -- node[above]{success} (lu);
\draw[->] (lu) -- (lc);
\draw[->] (lv.south) |- node[below,pos=0.75]{failure} (le.west);
\draw[->] (lu.south) -- node[right]{exhaustion} (le.north);
\end{tikzpicture}
\caption{The two intended local modes. A failed first attempt leaves a low residual state. A successful first attempt creates a high residual state, after which the remaining prerequisite is processed to success or exhaustion. Lemma~\ref{lem:atomisation} proves that these fragments may be treated atomically.}
\label{fig:two-modes}
\end{figure}
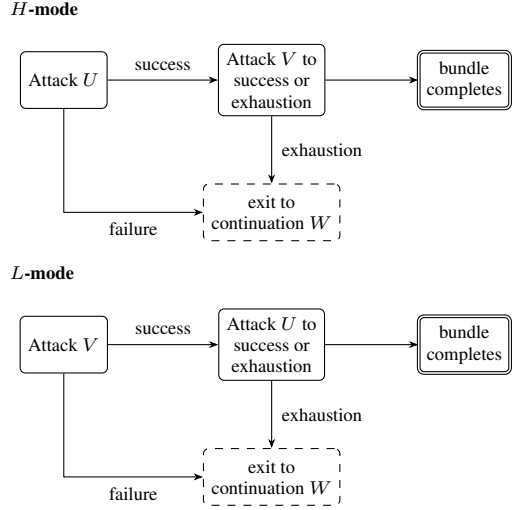

Thus, a fresh bundle has two candidate affine maps $L$ and $H$, determined by which prerequisite control is attempted first. The next sections choose the probabilities so that these two modes have carefully separated residual states and later show that some optimal adaptive policy can indeed treat each fresh bundle as one of these two modes.

\subsection{Step 2.2: A Robust Two-Mode Gadget}
\label{appendix:robust-gadget}

This subsection is where we prove that \textit{there really exist} control probabilities making the bundle behave the way Step 2 promised in the previous section (\Cref{appendix:two-mode-bundle}). Furthermore, we define three different thresholds, ($m, \theta, \iota$), where $m$ is the lower benchmark which separates unattractive residual states created by a failed first attempt, the middle level $\theta$ is the common threshold of the two fresh-bundle modes $H$ and $L$, and the upper benchmark $\iota$ separates high residual states created after a successful first attempt.

\begin{appendixlemma}[Robust two-mode base gadget]
\label{lem:robust-base-gadget}
There exist rational success probabilities $u_1 > u_2$ and $v_1 > v_2$ and three threshold levels $m < \theta < \iota$, for which the two-control bundle $B = \ParAND(U, V)$ from \Cref{appendix:two-mode-bundle} has the following properties:
\begin{enumerate}
    \item \textbf{Fresh State Bundle.} The two modes $H$ and $L$ have the same continuation threshold: \begin{equation}
        \chi(H) = \chi(L) = \theta
    \end{equation} while their continuation coefficients are different. Thus a \textbf{fresh bundle} admits two distinct affine behaviours that are equally attractive in terms of threshold.
    \item \textbf{After a failed first attempt.} If the first attempt on the chosen prerequisite fails, the resulting unfinished bundle is a \textbf{low residual state}. Every fragment that continues from such a state has continuation threshold strictly below $m$.
    \item \textbf{After a successful first attempt.} If the first prerequisite succeeds, the remaining unfinished bundle is a \textbf{high residual state}: one prerequisite has already been breached and only the other remains. Every action used to continue the remaining prerequisite has continuation threshold strictly above $\iota$.
\end{enumerate}

Moreover, all of these inequalities are strict and therefore continue to hold for every probability vector $(u_1, u_2, v_1, v_2)$ in some open neighbourhood $\mathcal{C}$ of the chosen probabilities.
\end{appendixlemma}

Basically, what we are saying is that we can give the thresholds $m, \theta, \iota$ and also $u_1, u_2, v_1, v_2$ such that:
\begin{align*}
    \boxed{\text{low residual threshold values}} &< m \\
    &< \boxed{\text{fresh state thresholds @ H, L-modes}} = \theta \\
    &<\iota \\
    &< \boxed{\text{high residual threshold values}}
\end{align*}

The purpose of this separation is so that we can force the following situation to happen: after a failure, the attack will later prefer the fallback to returning to that bundle and continuing the attack, whereas after a successful attempt, it will prefer to finish the remaining prerequisite control before moving somewhere else.

\begin{proof}
We give one explicit choice of probabilities and verify the three threshold regimes in the lemma. The particular fractions below are only a certificate that such a gadget exists; what matters is the strict separation between the low, fresh, and high states. The following values are from ChatGPT 5.6.

Set
\begin{equation}
m=\frac18,\qquad \theta=\frac{4001}{32000},\qquad \iota=\frac14,
\label{eq:base-thresholds}
\end{equation}
and choose
\begin{equation}
u_1=\frac9{20},\qquad u_2=\frac{17}{40},\qquad v_1=\frac{12803200}{17233053},\qquad v_2=\frac{154659283157}{372094362441}.
\label{eq:base-probabilities}
\end{equation}
These probabilities satisfy $u_1 > u_2$ and $v_1 > v_2$.

Substituting these probabilities into Equation~\ref{eq:H-map} and Equation~\ref{eq:L-map} gives
\begin{equation}
H(W)=\theta(1-b)+bW,\qquad L(W)=\theta(1-c)+cW,
\label{eq:base-mode-maps}
\end{equation}
where
\begin{equation}
b=\frac{45588892249}{160836083649},\qquad c=\frac{2721053}{17233053}.
\label{eq:base-slopes}
\end{equation}

Therefore, 
\begin{equation}
    \chi(H) = \frac{\theta(1 - b)}{1 -b} = \theta, \qquad \chi(L) = \frac{\theta(1 - c)}{1 -c} = \theta
\end{equation}

Thus both ways of entering a fresh bundle have the same continuation threshold $\theta$, while $b \neq c$, so they remain distinct affine policies.

\paragraph{After a failed first attempt.} We next show that a failed attack makes the bundle unattractive. Represent a local bundle state by $(r_U, r_V)$, where $r_U, r_V \in \{0, 1\}$ record whether $U$ or $V$ has already incurred one failed attempt. For such a state $x$, let $V_x(z)$ be the best value achievable by continuing to attack only inside this bundle when leaving the bundle gives continuation value $z$. Define its threshold by 
\begin{equation}
\chi(x)=\sup\{z:V_x(z)>z\}.
\label{eq:state-threshold}
\end{equation}

The fresh bundle is state $(0,0)$. A failed first attack moves it to either $(1, 0)$ or $(0, 1)$ (depending on the mode). $(1,1)$ is another possible low state in which both controls have already suffered a failed attack. Since this local state graph is finite and acyclic, these thresholds can be evaluated exactly by working backwards through the remaining attempts. For the chosen probabilities,
\begin{equation}
    \chi(1,0)=\frac{2305520063957}{19196426311760}<m \qquad
    \chi(0,1)=\frac{140275969823399}{1646482867274879}<m \qquad
    \chi(1,1)=\frac{2629207813669}{38583128135229}<m
\end{equation}

Hence, every low residual state has threshold $< m$, the lower benchmark. In particular, once. thecontinuation value is at least $m$, there is no advantage in returning to one of these states.

\paragraph{After a successful first attempt.} Conversely, suppose the first attempted prerequisite succeeds. One prerequisite of the $\ParAND$ bundle is now complete, so succeeding on the remaining prerequisite completes the bundle. Consider any available attempt on that remaining prerequisite, with success probability $p \in \{u_1, u_2, v_1, v_2\}$.

Since the attempt takes one unit of time, its affine map is 
\begin{equation}
a(z)=\frac p2+\frac{1-p}{2}z,
\end{equation}

Its continuation threshold is therefore
\begin{equation}
    \chi(a) = \frac{p}{1 + p}
\end{equation}

Note that every one of the four chosen probabilities satisfies $p > \frac 13$, and hence we have that 

\begin{equation}
    \chi(a) > \frac{1}{4} = \iota
\end{equation}

Thus, every action available after a successful first attempt lies above the upper benchmark $\iota$.

Finally, combining the three calculations gives the desired separation:
\begin{equation}
    \chi(\text{low state}) < m < \chi(H) = \chi(L) = \theta < \iota < \chi(\text{high state})
\end{equation}

It remains to show that this behavior is robust. All of the inequalities above are strict, and there are only finitely many local deterministic fragments. Their affine coefficients, and hence their relevant threshold comparisons, vary continuously with with $(u_1, u_2, v_1, v_2)$. Therefore, there is an open neighbourhood $\mathcal C$ of the chosen probability vector in which the same inequalities and the strict orderings $u_1 > u_2, v_1 > v_2$ continue to hold.

\end{proof}

\paragraph{Additional base-point properties.} We will later use several further numerical properties of the two continuation coefficients. Direct calculation gives the following: 
\begin{equation}
0<c<b<1,\qquad c<\frac13,\qquad b>\frac{2c}{1+c},\qquad b<\frac{1+c}{2}.
\label{eq:base-slope-inequalities}
\end{equation}
We also define
\begin{equation}
\upsilon=\frac{1-b}{b-c}=\frac{576235957}{100966523},
\label{eq:upsilon}
\end{equation}
for which
\begin{equation}
\frac{1-2\theta}{2(1-\theta)}<\frac3\upsilon<1.
\label{eq:upsilon-window}
\end{equation}

\subsection{Step 2.3: From Desired Affine Maps to Actual Probabilities}
\label{appendix:realisation}

In the reduction, it will be convenient to specify the two affine mode maps $H_i$ and $L_i$ that we want each bundle $B_i$ to have. However, the coefficients of these maps \textit{are not parameters that we can choose directly.} An actual bundle is specified by its four attempt probabilities $(u_1, u_2, v_1, v_2)$. We therefore need to know that small desired changes to the two maps can indeed be implemented by small changes to these probabilities.

For a bundle with probability vector $x = (u_1, u_2, v_1, v_2)$, define
\begin{equation}\label{eq:Phi-definition}
    \Phi(u_1,u_2,v_1,v_2)=(\mu_H,\nu_H,\mu_L,\nu_L)
\end{equation}
where \begin{equation}
    H(W) = \mu_H + \nu_HW, \qquad L(W) = \mu_L + \nu_LW
\end{equation}

We can see that this map $\Phi$, maps the four primitive attempt probabilities to the four coefficients describing the bundle's two modes.

\begin{appendixlemma}[Local Realisation of Mode Maps]
\label{lem:effective-realization}
Let $x_0$ be the base probability vector constructed in Lemma ~\ref{lem:robust-base-gadget}, and let $y_0 \eqdef \Phi(x_0)$ be its mode-coefficient vector. There exist neighbourhoods $\mathcal C_0 \subseteq \mathcal C$ of $x_0$ and $\mathcal{Y}_0$ of $y_0$ such that $\Phi : \mathcal C_0 \to \mathcal{Y}_0$ is bijective.

Moreover, if a target coefficient vector $y \in \mathcal Y_0$ is rational, then for any precision $L$, we can compute in polynomial time, a rational probability vector $\tilde{x}$ satisfying

\begin{equation}\label{eq:effective-inverse-error}
    ||\tilde{x} - \Phi^{-1}(y)||_{\infty} \leq 2^{-L}
\end{equation}
\end{appendixlemma}

\begin{proof}
We first show that the four mode coefficients can be varied independently near the base gadget. A direct calculation at $x_0$ gives
\begin{equation}
\det D\Phi(x_0)\neq0.
\label{eq:Phi-determinant}
\end{equation}

Hence, the Jacobian of $\Phi$ is invertible at $x_0$. By the inverse function theore, after restricting to sufficiently small neighbourhoods $\mathcal C_0$ and $\mathcal{Y}_0$, the map $\Phi$ is a continuously differentiable bijection with a continuously differentiable inverse.The inclusion $\mathcal C_0 \subseteq \mathcal C$ is important because every probability vector in $\mathcal C$ retains the threshold separation established in Lemma ~\ref{lem:robust-base-gadget}. Thus, we may adjust the maps locally without losing the low-, fresh-, high-state behaviour of the base gadget.

Now it remains to show that the inverse probabilities can be approximated efficiently by rationals. Let $J_0 \eqdef (D \Phi(x_0))^{-1}$. For a target coefficient vector $y$ close to $y_0$, consider
\begin{equation}
\phi_y(x)=x+J_0\bigl(y-\Phi(x)\bigr).
\label{eq:contraction-map}
\end{equation}

A fixed point $x^*$ of this map satisfies $y = \Phi(x^*)$, so $x^* = \Phi^{-1}(y)$. In the neighbourhood of the base probability vector, this iteration is a contraction. Hence, starting from $x_0$ and iterating $\Oh{L}$ times gives an approximation with error below $2^{-L}$. Using rational arithmetic and sufficient precision keeps the accumulated rounding error within the same bound.
\end{proof}

\subsection{Step 3: Adaptive Policies to Atomic Bundle Choices}
\label{appendix:atomisation}

\paragraph{Motivation.} So far, we have described each fresh bundle by two intended modes $H_i$ and $L_i$. A fully adaptive attacker, however, is not required to respect these modes as it could leave a partially attempted bundle, attack another child, and return later. We now show that this extra freedom is unnecessary. Using the threshold separation from Lemma~\ref{lem:robust-base-gadget} and the exchange rules from Lemma~\ref{lem:exchange-rules-1} and Lemma~\ref{lem:exchange-rules-2}, some optimal policy can be chosen so that each selected fresh bundle $B_j$ is treated as one atomic $H_j$- or $L_j$-mode.

\begin{appendixlemma}[Atomisation]
\label{lem:atomisation}
Consider any finite collection of disjoint two-mode bundles whose probability vectors lie in the robust neighbourhood $\mathcal{C}$ from Lemma~\ref{lem:robust-base-gadget}. Suppose a fallback of value $m$ is always available, and every additional one-attempt child has continuation threshold below $\iota$. Then some optimal fully adaptive policy can be chosen with the following behaviour:
\begin{itemize}
    \item If the first attempt in a fresh bundle fails, the attacker never returns to that bundle;
    \item If the first attempt succeeds, the attacker attacks the remaining control, either succeeding or exhausting its available attempts, before moving to another fresh bundle or additional child.
\end{itemize}
\end{appendixlemma}

Consequently, whenever a fresh bundle $B_i$ is selected, the attacker executes only one of its two modes, $H_i$ or $L_i$. We may therefore replace the primitive adaptive problem by a \textbf{macro problem} in which each $H_i$, or $L_i$, is treated as a single action, without changing the optimal value. The proof of the above lemma is long, so it will be presented in various parts.

\begin{proof}[Preamble + Part 1 of Proof: High Residual States Come First]
Every primitive attempt consumes one of finitely many available attempts, so the state graph is finite and acyclic. We prove the result by induction on the number of primitive attempts remaining. The argument has 3 parts: first, whenever a high residual state exists, its remaining control can be attacked beofre lower-threshold alternatives; second, a low residual state is never worth revisiting; and finally the above two facts then imply that every fresh bundle is atomic.

Suppose a high residual state is present. Recall that this means one control of a bundle has already been successfully breached, so an attempt on the remaining control can complete that bundle. By Lemma~\ref{lem:robust-base-gadget}, every such action has continuation threshold greater than $\iota$.

Among the currently available high-state actions, let $a$ have the largest continuation threshold. We show that some optimal policy may take $a$ before any lower-threshold action. So, if another high-state action $a'$ is taken first, the exchange rules allows $a$ and $a'$ to be reordered by decreasing continuation thresholds without decreasing the value. Similarly, $a$ should precede the fallback, whose value is $m < \iota$, and any additional one-attempt child, whose threshold is below $\iota$.

The only remaining possibility is that the policy first works inside a fresh or low-state bundle $X$. Follow the local mode until just before $a$ would be first used, and call this initial policy fragment $f$. By the induction hypothesis, after every possible termination of $f$, action $a$ is to be executed next. Thus, every non-completion termination from $f$ eventually reaches $a$ before any lower-priority action.

Since $f$ begins in either a fresh or low state, its continuation threshold is below $\iota$, whereas $\chi(a) > \iota$. The exchange rules again move $a$ before $f$ without decreasing the value. Hence, an optimal policy can always process the highest-priority available high-state action first.
\end{proof}
\begin{proof}[Part 2 of Proof: Low states are never revisited.]

Now suppose no high residual state is present, and consider a low residual-state bundle $X$. We show that an optimal policy never needs to return to $X$.

Suppose instead that the attacker is attacking a control inside bundle $X$. Follow its local mode until the attacker leaves $X$ permanently, and call this resulting policy fragment $f$. If an attack inside $X$ succeeds and creates a high state, then Part 1 of the proof says that the remaining control is completed immediately before anything else is attempted. If it fails, $X$ remains low or becomes dead (or invalid).

Once $f$ terminates, the remaining children are exactly the same regardless of how $f$ exited, because $X$ will never be attacked again. Let $W$ be the optimal value of this common continuation. This fallback is always available so $W \geq m$. However, recall that Lemma~\ref{lem:robust-base-gadget} gives $\chi(f)\leq\chi(X)<m$.

Thus, we then have: \begin{equation}
f(W)-W=(1-\nu_f)\bigl(\chi(f)-W\bigr)<0.
\label{eq:low-state-dominated}
\end{equation}

In other words, starting with $f$ is strictly worse than skippng $X$ altogether and taking the continuation value $W$. This contradicts optimality and hence a bundle in a low-state is never revisted.
\end{proof}

Now the conclusion is almost trivial:

\begin{proof}[Part 3 of Proof: Fresh Bundles are therefore atomic]
Finally, consider what happens when an optimal policy selects a fresh bundle $B_i$. By Part 1 of the proof, no unresolved high-state work is being postponed.

If the first attempt in $B_i$ fails, $B_i$ enters a low residual state, and by Part 2, is never revisited. If the first attempt on $B_i$ succeeds, $B_i$ enters a high residual state and by Part 1, the attacker attacks the remaining control to success or exhaustion before moving elsewhere.
Therefore, choosing $U_i$ first produces exactly mode $H_i$, whereas choosing to attack $V_i$ first produces exactly mode $L_i$. There is no beneficial interleaving of attacks with outside bundles.

To conclude, we may represent the attack at a \textbf{coarser level}. Rather than recording every primitive attack and residual bundle state, we treat each mode ($H_i$, $L_i$) as one \textbf{macro-action}. Every such macro policy can be implemented by the corresponding primitive policy, and an optimal primitive policy can be chosen in this form. The primitive and macro problems, therefore, have the same optimal value.
\end{proof}

\subsection{Step 4.1 Encoding the Input Integers in the Bundle Modes}
\label{appendix:reduction-construction}
Lemma~\ref{lem:atomisation} lets us replace the fully adaptive attack problem by a macro problem in which each selected bundle $B_i$ is executed in exactly one of two modes, $H_i, L_i$. We now use this binary choice to encode the \textsc{Equal-Cardinality Partition}. We interpret $H_i$ as $i \in S$ and $L_i$ as $i \notin S$, where $S \subseteq[2k]$ is the subset represented by the attacker's mode choices.

Our eventual goal is to make the attack value largest precisely when $S$ contains exactly $k$ items and their input integers sum to $Z$. The first task is to therefore to encode each input integer $w_i$ into the affine map of bundle $B_i$.

Recall that the continuation coefficients of affine mode maps multiply when several modes are executed one after another. We therefore encode each input integer $w_i$ by a number close to 1, so that multiplying the numbers associated with a subset approximately records the sum of its weights.

Let $K_0$ be a fixed constant, chosen sufficiently large later, and define \begin{equation}
    \epsilon = \frac{1}{K_0(1 + Z)^4}
\end{equation} The only role of $K_0$ here is to make $\epsilon$ sufficiently small.

For each input integer $w_i$, define \begin{equation}
    x_i = 1 + \epsilon w_i
\end{equation}

If $S$ is a candidate subset, then 

\begin{equation}
    \prod_{i \in S} x_i = \prod_{i \in S} (1 + \epsilon w_i) \approx 1 + \epsilon \sum_{i \in S} w_i \qquad (\because \epsilon \, \text{is small})
\end{equation}

Therefore, if the weights in $S$ sum to the target $Z$, this product should be close to \begin{equation}
    z_0 \eqdef 1 + \epsilon Z
\end{equation}

We therefore use $z_0$ as the target value that the product associated with the chosen subset should approach. Also, let us define \begin{equation}
    X \eqdef \prod_{i = 1}^{2k} x_i
\end{equation}

The quantity $X$ is simply the product over all input iterms. If $S$ is the set assigned mode $H$, then the complement is assigned mode $L$, and we have
\begin{equation}
    \prod_{i \notin S} x_i = \frac{X}{\prod_{i \in S} x_i}
\end{equation}

This will later let us express both the $H$- and $L$-mode contributions using the same subset product. We now encode $x_i$ into the two mode maps. Recall the constants $b$ and $c$ which are the continuation coefficients of the two modes. We slightly separate the thresholds of the two modes by defining: 

\begin{equation}
    h = \theta + \eta
\end{equation} where $\eta = \frac{3(\theta - m)c^k}{\upsilon}$. For ease of notation, we let $\Delta = \theta - m$ for now. The exact choice of $\eta$ will matter only later. For now, the important point is simply that $\eta > 0$ is small, so $h$ is only slightly greater than $\theta$.

For bundle $B_i$, prescribe:
\begin{align*}
    H_i(W) &= h + b_i(W - h) &b_i = bx_i\\
    L_i(W) &= \theta + c_i(W - \theta) &c_i = cx_i^2
\end{align*}

By doing so, every $H_i$ has continuation threshold $h$, while every $L_i$ has continuation threshold $\theta$. In the next subsection, we will attempt to place a pivot whose threshold lies strictly between these two thresholds. The exchange rule will then force all $H$-modes before the pivot and all $L$-modes after it.

The factors $x_i$ and $x_i^2$ are chosen deliberately. Let \begin{equation}
    z_S \eqdef \prod_{i \in S} x_i.
\end{equation} Then the product of the continuation coefficients of all $H$-mode bundles is \begin{equation}
    \prod_{i \in S} b_i = b^{|S|}z_S
\end{equation}

Meanwhile, the $L$-mode bundles are precisely those outside $S$, so \begin{equation}
    \prod_{i \notin S} c_i = c^{2k - |S|} \frac{X^2}{z^2_S}
\end{equation}

The reason for using $x_i$ for bundles in $H$-mode, but $x^2$ for bundles in $L$-mode only becomes visible later. After we force exactly $k$ bundles into each mode, these two products combine so that the loss depends on \begin{equation}
    z_S + \frac{z_0^2}{z_S}
\end{equation}

The above expression is smallest exactly when $z_S = z_0$.

Thus, the construction is being engineered so that the best subset is the one whose product $z_S$ lies closest to the target $z_0$. In the next few sections, we will then show that this distinguishes whether the corresponding weights sum to $Z$.

Finally, we must check that the mode maps we have just prescribed can actually be produced by real two-control bundles. Lemma~\ref{lem:robust-base-gadget} only guarantees this for maps sufficiently close to the base gadget. This is therefore precisely why $\eta$ and $\epsilon$ were chosen to be small. Since the total input weight is $2Z$, then \begin{equation}
    \log X \leq 2 \epsilon Z,
\end{equation}
so choosing $K_0$ sufficiently large makes $X$, and therefore every $x_i$ close to 1. Consequently $b_i = bx_i$ stays close to $b$, and $c_i = cx_i^2$ stays close to $c$. Similarly, because $k \geq k_0$, we have that \begin{equation}
    0 < \eta \le \frac{3\Delta c^{k_0}}{\upsilon}.
\end{equation} Thus, choosing the fixed padding constant $k_0$ sufficiently large makes $h = \theta + \eta \approx \theta$.

We therefore choose $k_0$ first, and then $K_0$, sufficiently large so that all the desired $H_i$ and $L_i$ maps lie inside the neighbourhood covered by Lemma~\ref{lem:effective-realization}. That lemma then gives an actual nearby probability vector for every bundle. Since these probability vector remain inside the robust neighbourhood from Lemma~\ref{lem:robust-base-gadget}, all of the low-, fresh-, high-state threshold properties continue to hold.

\subsection{Step 4.2: Pivot and Fallback}
\label{appendix:pivot-fallback}
We now add two simple top-level children (controls): a fallback $F$, and a pivot $P$. They serve different purposes. The fallback provides an always-available continuation of value $m$, which is what makes low residual states unattractive. The pivot is placed between the thresholds of $H$- and $L$-modes, so that the exchange rule will later force all $H$-modes before the pivot and all $L$-modes after it.

The fallback is a deterministic control of duration 3. Since each unit of time contributes a discount factor of $\frac 12$, its value is $m = (\frac{1}{2})^3 = \frac 18$. Thus, the lower benchmark $m$ is now realised by an actual child of the constructed network.

We next construct the pivot. Recall that every $H_i$ has threshold $h$, while every $L_i$ has threshold $\theta$, with $h > \theta$. We want the pivot threshold to lie strictly between them (recall Lemma ~\ref{lem:robust-base-gadget}):

\begin{align}\label{eq:h-chi-theta}
    h > \chi(P) > \theta
\end{align}

To achieve this, define
\begin{equation}
\kappa \eqdef\Delta c^k\frac{X^2}{z_0^2}.
\label{eq:kappa-definition}
\end{equation}

The pivot consists of a single control with one available attempt, taking one unit of time. Choose its success probability as 

\begin{equation}
\vartheta=\frac{\theta+2\eta-\kappa}{1-\theta-\kappa}.
\label{eq:pivot-probability}
\end{equation}

If the pivot is breached, the top-level $\ParOR$ structure collapses and the attacker succeeds. However, if it fails, the attacker continues with whatever comes next. Herce, for continuation value $W$, we have the affine map of the pivot:\begin{equation}\label{eq:pivot-map}
    P(W) = \frac \theta 2 + \nu_FW
\end{equation} where $\nu_F = \frac{1 - \vartheta}{2}$. Its continuation threshold is therefore \begin{equation}
    \chi(P) = \frac{\vartheta}{1 + \vartheta}
\end{equation}

What is left to prove is that this threshold really lies between $\theta$ and $h$.

From the definitions of $\eta$ and $\kappa$, we have that: \begin{equation}\label{eq:eta-kappa-ratio}
    \frac{\eta}{\kappa} = \left(\frac{3}{\upsilon} \right) \left(\frac{z_0^2}{X^2} \right)
\end{equation}

Because  $K_0$ is chosen large, every $x_i = 1 + \epsilon w_i$ is close to 1. Hence both $X$ and $z_0$ are close to 1, so $z_0^2/X^2$ is also close to 1. The base gadget was chosen so that $3/\upsilon$ lies inside the interval needed here. Therefore, by choosing the fixed constants sufficiently large, via Equation~\eqref{eq:upsilon-window}, we can ensure
\begin{equation}
\frac{1-2\theta}{2(1-\theta)}<\frac{\eta}{\kappa}<1.
\label{eq:eta-kappa-window}
\end{equation}

The lower bounds gives \begin{equation}
\chi(P)-\theta=\frac{2\eta(1-\theta)-\kappa(1-2\theta)}{1+2\eta-2\kappa}>0 \implies \chi(P) > \theta
\label{eq:pivot-above-theta}
\end{equation} whereas the upper bound gives \begin{equation}
    h-\chi(P)=\frac{(\kappa-\eta)(1-2h)}{1+2\eta-2\kappa}>0 \implies\chi(P) <h
\label{eq:pivot-below-h}
\end{equation} which yields Equation~\ref{eq:h-chi-theta}.

For sufficiently large fixed $k_0$ and $K_0$, we also have
\begin{equation}
    \theta > m, \quad h < \iota. \quad \chi(P) < \iota
\end{equation}
and the chosen pivot probability satisfies $ 0< \vartheta < 1$, which now means that the pivot and fallback remain compatible with the threshold separation required by Lemma ~\ref{lem:robust-base-gadget}.

Finally, define
\begin{equation}
\xi=h-P(\theta).
\label{eq:xi-definition}
\end{equation}

This quantity measures how far below the $H$-mode level $h$ the pivot output lies when the continuation after the pivot has value exactly $\theta$.

Substitution gives
\begin{equation}
\xi=\nu_P\kappa=\nu_P\Delta c^k\frac{X^2}{z_0^2}.
\label{eq:xi-identity}
\end{equation}

We record $\xi$ because it will make the value of the canonical macro policy much cleaner in the next subsection. In particular, it is through the choice of $\kappa$ that the target quantity $z_0$ eventually enters the final subset objective.

At this point, the threshold ordering is \begin{equation*}
    h > \chi(P) > \theta > m
\end{equation*}

The next subsection will use the exchange rules to turn this threshold ordering into the canonical macro-policy order

\begin{equation*}
    H\text{-modes} \longrightarrow P \longrightarrow L\text{-modes} \longrightarrow F
\end{equation*}

\subsection{Step 4.3: The Canonical Macro Policy}
\label{appendix:canonical-macro}
The remaining question we now have is which mode should each bundle use and in what order these macro-actions should be executed. As per previous sections, let $S \subseteq [2k]$ be the set of bundles assigned to mode $H$ and bundles outside of $S$ are assigned mode $L$. We have the following lemma.

\begin{appendixlemma}[Canonical Macro Form]
\label{lem:canonical-macro}
Every optimal macro policy uses every bundle, the pivot, and the fallback. For a fixed set $S$, then an optimal order is
\begin{equation}
(H_i)_{i\in S}\;\longrightarrow\;P\;\longrightarrow\;(L_i)_{i\notin S}\;\longrightarrow\;F.
\label{eq:canonical-macro-order}
\end{equation}
The order within the $H$-group of bundles, and within the $L$-group of bundles is immaterial for now. 

Define the total continuation coefficients of the two groups of bundles by
\begin{equation}
\nu_H(S)=\prod_{i\in S}b_i,\qquad \nu_L(S)=\prod_{i\notin S}c_i.
\label{eq:macro-slopes}
\end{equation}

Then the value of this policy is
\begin{equation}
V(S)=h-\nu_H(S)\bigl(\xi+\nu_P\Delta\nu_L(S)\bigr),
\label{eq:macro-value}
\end{equation}
\end{appendixlemma}

\begin{proof}
Firstly the order is justified as explained many times in the previous subsections. However, now we also need to show that none of these bundles/fallback/pivot are omitted.

Suppose bundle $B_i$ were omitted. The continuation provided by the rest of the policy is below $h$. Since $H_i$ has threshold $h$, inserting $H_i$ at the beginning would strictly increase the attack value. Hence every bundle is used, either in $H$- or $L$-mode.

The pivot is also used. If it were omitted, then after the $H$-modes the remaining policy would consist of only $L$-modes and the fallback. That continuation has value below $\theta$. Since the pivot has a threshold strictly above $\theta$, inserting it improves the value.

Finally, the fallback is always used. Without it, the terminal continuation would be 0; adding the fallback replaces this by $m > 0$. Since every preceding mode has a positive continuation coefficient, this increases the value.

Hence, every optimal macro policy has that stated form. We now compute its value, and it is easiest to work \textbf{backwards from the end.} Start with the $L$-group of bundles followed by the fallback. Recall that every $L_i$ map has fixed point $\theta$, and that the continuation coefficients multiply. Hence the entire $L$-group of bundles has a continuation coefficient:
\begin{equation}
    \nu_L(S)=\prod_{i\notin S}c_i
\end{equation}

Starting from the fallback value $m$, the value immediately before the $L$-group of bundles is therefore
\begin{equation}
W_L=\theta+\nu_L(S)(m-\theta)=\theta-\Delta\nu_L(S).
\label{eq:L-suffix-value}
\end{equation}

Next, apply the pivot. Since its continuation coefficient is $\nu_P$,
\begin{equation}
P(W_L)=P(\theta)-\nu_P\Delta\nu_L(S)=h-\xi-\nu_P\Delta\nu_L(S).
\label{eq:pivot-after-L}
\end{equation} by recalling the definitions of $\xi$.

Finally, apply the $H$-group of bundles. Again, recall that every $H_i$ map has fixed point $h$, and their continuation coefficient multiplies as well to:
\begin{equation}
    \nu_H(S)=\prod_{i\in S}b_i
\end{equation} and hence we have:
\begin{equation}
    V(S) = h + \nu_H(S)(P(W_L) - h)
\end{equation} which simplifies to the desired expression \begin{equation*}
    V(S)=h-\nu_H(S)\bigl(\xi+\nu_P\Delta\nu_L(S)\bigr)
\end{equation*}
\end{proof}

At this point, we now that once $S$ is chosen, the optimal order is fixed as in Equation ~\ref{eq:canonical-macro-order}. So the remaining decision is the subset $S$ itself. The next subsection shows that any optimal choice must satisfy $|S| =k$ which is exactly the equal-cardinality part of our source problem.

\subsection{Step 4.4: Forcing Exactly $k$ H-modes}
\label{appendix:equal-cardinality}

The canonical form of the Lemma~\ref{lem:canonical-macro} leaves only one choice, that is the subset $S$ of bundles assigned mode $H$. The source problem, however, requires $|S| =k$. We now show that the construction enforces this automatically.

The proof has two steps. First, we \textit{temporarily} set $\epsilon = 0$, removing the dependence on the input weights, and show that the loss is uniquely minimised when exactly $k$ bundles adopt each mode. We then return to the actual construction with small $\epsilon > 0$ and show that the weight-dependent changes are too small to overturn this strict preference for balance $|S| = |\bar{S}| = k$.

\begin{appendixlemma}[Every optimal assignment is balanced]
\label{lem:balanced-assignment}
Every optimal macro policy assigns exactly $k$ bundles to mode $H$ and exactly $k$ bundles to mode $L$.
\end{appendixlemma}

\begin{proof}
Let $S$ be the set of bundles assigned to mode $H$ and write \begin{equation*}
    |S| = k - j
\end{equation*}
This means that $k+j$ bundles use mode $L$. Clearly, $j = 0$ is the balanced case, $j > 0$ means that there are fewer that $k$ $H$-mode bundles, and $j < 0$ means there are more than $k$ $H$-modes. Define
\begin{equation}
z=\prod_{i\in S}x_i.
\end{equation}

Recall that $b_i = bx_i$ and $c_i = cx_i^2$. Hence, the total continuation coefficient of the $H$-group of bundles is 
\begin{equation}\label{nu-H}
    \nu_H(S)=b^{k-j}z
\end{equation}

Since $X = \prod_i x_i$, the complementary continuation of the $L$-group of bundles is therefore derived by
\begin{equation}\label{nu-L}
    \prod_{i \notin S}x_i^2 = \frac{X^2}{z^2} \implies \nu_L(S)=c^{k+j}\frac{X^2}{z^2}
\end{equation}

Recall that the value of the policy from Lemma~\ref{lem:canonical-macro} is \begin{equation*}
    V(S) = h - \nu_H(S)(\xi + \nu_P\Delta \nu_L(S)).
\end{equation*} Since maximising $V(S)$ is equivalent to minimising the positive loss $h - V(S)$, we shall work with the latter expression instead. We simply substitute Equations~\ref{nu-H} and ~\ref{nu-L}, and also Equation~\ref{eq:xi-identity} to get the following simplification:
\begin{equation}
    h - V(S) = \nu_P\Delta(bc)^kX^2\Lambda_j(z)
\end{equation} where
\begin{equation}
\Lambda_j(z)=b^{-j}\left(\frac{z}{z_0^2}+\frac{c^j}{z}\right).
\label{eq:normalised-loss}
\end{equation}
Note that the front divisors are positive and independent of $S$, so it actually suffices to minimise $\Lambda_j(z)$, which we define as the normalised loss.

We first isolate the effect of cardinality by considering the baseline case $\epsilon = 0$. Then every $x_i = 1$ and hence $z = z_0 = 1$. Hence, we have that

\begin{equation*}\label{eq:yj-definition}
    y_j \eqdef \Lambda_j(z=1) = b^{-j}(1 + c^j)
\end{equation*}

At the balanced ($j=0$) assignment, we have that $\Lambda_0(1) = 2$. We show that the loss increases strictly in either direction of $j = 0$. For $j \ge 0$, we have that
\begin{align}
    \frac{y_{j+1}}{y_j} &=\frac{\Lambda_{j+1}(1)}{\Lambda_{j}(1)} \\
    &= \left(\frac{1}{b}\right)\left(\frac{1 + c^{j+1}}{1+c^j}\right)\\
    &\geq \left(\frac{1}{b}\right)\left(\frac{1+c}{2}\right) &\because c \in (0,1)\\
    &>\left(\frac{2}{1+c}\right)\left(\frac{1+c}{2}\right) &\because b < \frac{1+c}{2} \,\,\text{from Equation~\ref{eq:base-slope-inequalities}} \\
    &= 1 \implies y_{j+1} > y_j
\end{align}

Therefore, the loss strictly increases when the number of $H$-mode bundles fall below $k$. Similarly, for negative $j = -u$ for $u \geq 0$:

\begin{align}
    \frac{y_{-(u+1)}}{y_{-u}} &= \frac{\Lambda_{-(u+1)}(1)}{\Lambda_{-u}(1)}\\
    &= \left(\frac{b}{c}\right)\left(\frac{1 + c^{j+1}}{1+c^j}\right)\\
    &\ge \left(\frac{b}{c}\right)\left(\frac{1 + c}{2}\right)&\because c \in (0,1)\\
    &> \left(\frac{2}{1+c}\right)\left(\frac{1 + c}{2}\right)&\because b < \frac{1+c}{2} \,\,\text{from Equation~\ref{eq:base-slope-inequalities}} \\
    &=1 \implies y_{-(u+1)} > y_{-u}
\end{align}
Hence, the loss also strictly increase when the number of $H$-mode bundles exceeds $k$. Therefore $j=0$ \textbf{is the unique minimiser of the normalised loss}.

Let
\begin{equation}
\varpi_0=\min\{y_1,y_{-1}\}-2>0.
\label{eq:cardinality-gap-constant}
\end{equation}

Since the sequence increases (in either direction) from $j = 0$, every unbalanced assignment has loss of at least $2 + \varpi_0$, whereas the balanced assignment has loss of 2.

We now return to the actual construction, where $\epsilon > 0$. The weight encoding slighly changes $z$ and $z_0$, but both remain close to 1. Indeed, every feasible subset would have satisfied

\begin{equation}
1\leq z\leq X\leq e^{2\varepsilon Z},\qquad 1\leq z_0\leq e^{\varepsilon Z},
\end{equation}
and therefore
\begin{equation}
e^{-4\varepsilon Z}y_j\leq\Lambda_j(z)\leq e^{4\varepsilon Z}y_j.
\label{eq:loss-perturbation-bound}
\end{equation}

Thus, the actual loss differs from its baseline value by only a small multiplicative factor. Finally, since $\varepsilon Z\leq1/K_0$, we may therefore choose the fixed constant $K_0$ sufficiently large such that
\begin{equation}
e^{-4\varepsilon Z}(2+\varpi_0)>2e^{4\varepsilon Z}
\label{eq:balanced-separation}
\end{equation}
for every positive integer $Z$. The left-hand side expression is a lower bound on the normalised loss of every unbalanced assignment, while the right-hand side is an upper bound on the loss of every balanced assignment. Hence, every assignment with $j \neq 0$ has strictly greater loss than every assignment with $j=0$. Therefore, every optimal macro policy satisfies $|S| =k$, meaning exactly $k$ bundles use mode $H$ and exactly $k$ bundles use mode $L$.
\end{proof}

Lemma~\ref{lem:balanced-assignment} enforces the cardinality condition. ofthe \textsc{Equal-Cardinality Partition}. It remains to enforce the sum condition $\sum_{i \in S} w_i = Z$, which we. donext by comparing the subset product $z_S$ with the target $z_0$.

\subsection{Step 4.5: Encoding the target sum $Z$}
\label{appendix:partition-gap}
By Lemma~\ref{lem:balanced-assignment}, we only need to consider assignments in which exactly $k$ bundles adopt $H$-mode. It remains to distinguish whether the corresponding $k$ input integers sum to the target $Z$. For each subset $S$, define
\begin{equation*}
    z_S \eqdef \prod_{i \in S} x_i.
\end{equation*} Recall that $x_i = 1 + \epsilon w_i$, while the target product was chosen as $z_0 = 1 + \epsilon Z$. The purpose of this construction is that $z_S$ approximately records the sum of the weights selected by $S$. We will show that the attack value is highest precisely when $z_S$ is close to $z_0$, and that the YES and NO cases are separated by a non-zero gap.

\begin{appendixlemma}[Gap for balanced assignments]
\label{lem:partition-gap}
There is a rational threshold $\Theta$ and a positive rational gap $\varpi_1$ such that the ideal construction has value at least $\Theta+\varpi_1$ in every YES-instance and at most $\Theta-\varpi_1$ in every NO-instance.
\end{appendixlemma}

\begin{proof}
Since $|S| = k$, the continuation coefficients from Lemma~\ref{lem:balanced-assignment} become \begin{equation}
    \nu_H(S) = b^kz_S \qquad \nu_L(S) = c^k\frac{X^2}{z_S^2}
\end{equation}

Using the definition $\xi = \nu_P\frac{\Delta c^k X^2}{z_0^2}$, the loss formula from Lemma~\ref{lem:balanced-assignment} simplifies to 

\begin{align}
h-V(S)&=\xi b^k\left(z_S+\frac{z_0^2}{z_S}\right)\\
&=\xi b^k\left(2z_0 + \frac{(z_S - z_0)^2}{z_S}\right)
\label{eq:balanced-loss}
\end{align}

The first term, $2z_0$, is the same for every balanced subset. Therefore, among balanced subsets, maximising the attack value is equivalent to minimising $\frac{(z_S - z_0)^2}{z_S}$. In particular, the loss is smallest when $z_S \approx z_0$.

We now relate this product distance to the original subset-sum condition. Let
\begin{equation}
    s_S \eqdef \sum_{i \in S} w_i.
\end{equation} Expanding the product $z_S = \prod_{i \in S} (1 + \epsilon w_i)$ gives 
\begin{equation}
z_S=1+\varepsilon s_S+R_S.
\label{eq:product-expansion}
\end{equation} where $R_S\geq0$ contains all terms of second and higher order in $\epsilon$. Recall $\epsilon$ is chosen sufficiently small, these higher-order terms are uniformly tiny. More precisely,
\begin{equation}
    R_s \leq e^{\epsilon s_S} - 1 - \epsilon s_S \leq 4 \epsilon^2Z^2
\end{equation} after increasing the fixed constant $K_0$ if necessary. To obtain the above, we simply apply the well-known inequality $1 + t \leq e^t$ and from the fact that $\epsilon s_S \leq 2 \epsilon Z$.

We now compare the YES and NO cases. If $S$ satisfies the partition constraint, then $s_S = Z$. Since $z_0 = 1 + \epsilon Z$, we then get $z_S - z_0 = R_S$. Thus, the only difference between $z_S$ and $z_0$ comes from the small higher-order terms in the product expansion. By choosing $K_0$ sufficiently large (recall how $K_0$ relates to $\epsilon$ in the earlier subsections), we have \begin{equation}
    \frac{(z_S - z_0)^2}{z_S} \leq \frac{\epsilon^2}{32}
\end{equation} So every YES-instance contains a balanced assignment whose product lies very close to the target.

Now suppose $S$ \textbf{does not} satisfy the partition constraint. Since all $w_i$ and $Z$ are integers, then 
\begin{equation}
    s_S \neq Z \implies |s_S - Z| \geq 1
\end{equation}

If $s_S > Z$, then \begin{equation}
    z_S - z_0 = \epsilon(s_S - Z) + R_S \geq \epsilon
\end{equation}.

If $s_S < Z$, then \begin{equation}
    z_0 - z_S = \epsilon(Z - s_S) - R_s
\end{equation}

The first term is at least $\epsilon$, while $R_s$ is of smaller order. Choosing $K_0$ sufficiently large, therefore, gives \begin{equation}
    z_0 - z_S \geq \frac{\epsilon}{2}
\end{equation}

Hence, every NO assignment satisfies
\begin{equation}
    |z_S - z_0| \geq \frac \epsilon 2
\end{equation}

We may also choose $K_0$ so that $z_S \leq X \leq 2$. Therefore, \begin{equation}
    \frac{(z_S - z_0)^2}{z_S} \geq \frac{\epsilon^2}{8}
\end{equation}

We have therefore obtained a strict separation:
\begin{align*}
    \text{YES}: &\qquad\frac{(z_S - z_0)^2}{z_S} \leq \frac{\epsilon^2}{32}\\
    \text{NO}: &\qquad\frac{(z_S - z_0)^2}{z_S} \geq \frac{\epsilon^2}{8}
\end{align*}

The remaining step is simply to place the decision threshold between these two values. Since $\frac{5\epsilon^2}{64}$ is the midpoint of $\frac{\epsilon^2}{32}$ and $\frac{\epsilon^2}{8}$ which leaves a margin of $\frac{3\epsilon^2}{64}$ on either side, we can define the following expressions:
\begin{align*}
    \Theta &= h - \xi b^k \left(2z_0 + \frac{5\epsilon^2}{64}\right)\\
    \varpi_1 &= \frac{3}{64}\xi b^k \epsilon^2
\end{align*}

Now, we have that
\begin{align*}
    \text{YES}: &\qquad V(S) \geq \Theta  + \varpi_1 > \Theta\\
    \text{NO}: &\qquad V(S) \leq \Theta - \varpi_1 < \Theta
\end{align*}

This proves the claimed gap.
\end{proof}

\subsection{Step 5: Replacing the Ideal Probabilities by Rationals}
\label{appendix:rationalisation}
The construction so far is very idealised. Lemma~\ref{lem:effective-realization} guarantees bundle probabilities that realise the desired $H_i$- and $L_i$-maps, but those probabilities need not be rational. A polynomial-time reduction, however, must output numbers with finite, polynomial encoding length.

We therefore replace each ideal bundle probability by a sufficiently close rational approximation. There are two things to check. First, the approximation must be small enough that all structural properties used earlier (especially the threshold separation from Lemma~\ref{lem:robust-base-gadget} remain valid. Second, it must be small enough that the attack value changes by much less than the YES/NO gap $\varpi_1$ established in Lemma~\ref{lem:partition-gap}.

\begin{appendixlemma}[Polynomial-bit rationalisation]
\label{lem:rationalisation}
The ideal real probability vectors can be replaced in polynomial time by rational probability vectors of polynomial encoding length while preserving whether the optimal value lies above or below the threshold $\Theta$.
\end{appendixlemma}

We split the proof into two parts. The first part is to show that we still preserve the YES/NO gap even with a rational approximation, and the second is to show we have an polynomial encoding length.

\begin{proof}[Proof of preserving the YES/NO gap]
Let $I^\star$ be the ideal instance constructed above and $\widetilde I$ denote the instance obtained after rational approximation. For each bundle $B_i$, let  $x_i^\star$ denote its probability vector (referring to $(u_1, u_2, v_1, v_2)$). By Lemma~\ref{lem:effective-realization}, for any precision parameter $L$, we can compute $\widetilde{x}_i$ satisfying \begin{equation}
    ||\widetilde{x}_i - x_i^\star||_\infty \leq 2^{-L}
\end{equation} Therefore, by increasing $L$, we can approximate every primitive bundle probability as accurately as needed.

We first determine how accurate the approximation must be. There are at most $H_{max} = 8k + 2$ primitive attempts in any execution ($2k$ bundles with 2 attempts each per internal control, the pivot and the fallback). We also need the approximated bundles to remain inside the robust neighbourhood $\mathcal C$ from Lemma~\ref{lem:robust-base-gadget}. Recall that all ideal bundle vectors were chosen inside a fixed compact subset $\mathcal C_b$ lying strictly inside $\mathcal C$. Hence, there is some fixed rational distance $\delta_C > 0$ such that moving any ideal bundle probability vector by at most $\delta_C$ still leaves it inside $\mathcal C$. This condition preserves all the strict gadget inequalities used in Lemma~\ref{lem:robust-base-gadget}.

We now choose an approximation tolerance that satisfies both requirements, $\delta_p$ defined as \begin{equation}
    \delta_p=\min\left\{\frac{\varpi_1}{4H_{\max}},\frac{\delta_{\mathcal C}}2\right\},
\end{equation} where $\frac{\varpi_1}{4H_{max}}$ will control the change in value and $\frac{\delta_C}{2}$ keeps every approximated bundle safely inside $\mathcal{C}$. We can then pick $L$, our precision parameter to be
\begin{equation}
L=\left\lceil\log_2\frac1{\delta_p}\right\rceil+1.
\label{eq:precision-choice}
\end{equation}

Now, it remains to bound how much these small probability changes can alter the attack value. Fix any deterministic history-dependent attacker policy, and imagine running it once in the ideal instance $I^*$ and once using the rational instance $\widetilde{I}$. Use the same independent uniform random number for the corresponding primitive attempt in both runs.

As long as the two runs have agreed (in stochastic processes literature, it is a coupling) up to a given attempt, their outcomes can differ at that attempt only because its success probabilities changed. Since that change is at most $\delta_p$, the probability of disagreement at any particular attempt is at most $\delta_p$. There are at most $H_{max}$ attempts, so a union bound therefore gives the probability that the two runs disagreeing to be at most $H_{max} \cdot \delta_p \leq \frac{\varpi_1}{4}$. 

Every discounted terminal reward lies between 0 and 1. If the two executions never disagree, then its straightforward that they receive exactly the same (expected) reward. Hence, the expected value of any fixed policy changes by at most the probability that the executions disagree, which is $\frac{\varpi_1}{4}$. Moreover, this bound holds for every policy, and thus it also bounds the change in the optimal value:

\begin{equation}
\left|V(\widetilde I)-V(I^\star)\right|\leq\frac{\varpi_1}{4}.
\label{eq:optimal-value-coupling}
\end{equation}

Now, recall that in the previous section, we obtained a positive gap of $\varpi_1$. So changing the value by at most $\frac{\varpi}{4}$ therefore cannot cross the threshold of $\Theta$. The rational instance remains above $\Theta$ in every YES-instance, and below $\Theta$ in every NO-instance.
\end{proof}

\begin{proof}[Proof of polynomial encoding length] It remains only to verify that the required rational numbers use polynomially many bits. For this, we make the dependence of the gap $\varpi_1$ on $k$ and $Z$ explicit. Recall that $\varpi_1 = \frac{3}{64}\xi b^k\epsilon^2$ and $\xi = \nu_P \Delta c^k \frac{X^2}{z_0^2}$. Combining the two expressions gives
\begin{equation}
    \varpi_1 = \frac{3}{64}\nu_P\Delta (bc)^k \frac{X^2}{z_0^2}\epsilon^2
\end{equation}

Thus, to understand how small $\varpi_1$ can become, we only need bounds on $\nu_P, \frac{X^2}{z_0^2}$ and $\epsilon$.

First, consider the pivot coefficient. Since $\chi(P) = \frac{\vartheta}{1 + \vartheta} < \frac{1}{4}$, we have that $\vartheta < \frac{1}{3}$ and hence $\nu_P = \frac{1 - \vartheta}{2} > \frac{1}{3}$. Also, since $\vartheta$ is a probability, we have that $\vartheta \in (0, 1)$, so $\nu_P < \frac{1}{2}$. Hence, we have a bound on $\nu_P$: \begin{equation}
    \frac{1}{3} < \nu_P < \frac{1}{2}
\end{equation}

Next, consider $\frac{X^2}{z_0^2}$. Recall that \begin{equation*}
    x_i = 1 + \epsilon w_i, \quad X = \prod_{i=1}^{2k}x_i, \quad z_0 = 1 + \epsilon Z
\end{equation*} and that the total input weight is $2Z$. Since all the terms in the product are positive,
\begin{align}
    X &= \prod_{i}(1 + \epsilon w_i)\nonumber\\
    &\geq 1 + \epsilon \sum_i w_i\\
    &=1 + 2 \epsilon Z > z_0 \implies \frac{X^2}{z_0^2} \geq 1
\end{align}

On the other hand, we already showed that $X \leq e^{2\epsilon Z}$. Since $\epsilon Z \leq \frac{1}{K_0}$, then we have \begin{equation}
    \frac{X^2}{z_0^2} \leq X^2 \leq e^{4\epsilon Z} \leq e^{\frac{4}{K_0}}
\end{equation}
We may enlarge the fixed constant $K_0$ (if necessary), so that $K_0 \geq 4$. Then, for every source instance:
\begin{equation}
    1 \leq \frac{X^2}{z_0^2} \leq e
\end{equation}

We can now lower-bound the YES/NO gap. Using $\nu_P > \frac{1}{3}$ and $\frac{X^2}{z_0^2} \geq 1$, we have
\begin{equation}
    \varpi_1 \geq \frac{\Delta}{64}(bc)^k \epsilon^2
\end{equation}
Recall from the very start that $\epsilon = \frac{1}{K_0(1 + Z)^4}$. Therefore
\begin{equation}
    \varpi_1 \geq \frac{\Delta}{64K_0^2}(bc)^k(1 + Z)^{-8}
\end{equation}
Taking logarithms, we have 
\begin{equation}
    \log \left(\frac{1}{\varpi_1}\right) \leq \log \frac{64K_0^2}{\Delta} + k \log \frac{1}{bc} + 8 \log(1 + Z)
\end{equation}

Here, $K_0, \Delta, b, c$ are all fixed constants chosen once for the reduction. None of the depends on the \textsc{Equal-Cardinality Partition} instance. Hence, the first term is a constant, $\log \frac{1}{bc}$ is also a constant. Consequently
\begin{equation}
    \log \frac{1}{\varpi_1} = \Oh{k + \log Z}
\end{equation}

Finally, recall that $\delta_p = \min \{\frac{\varpi_1}{4 H_{max}}, \frac{\delta_C}{2}\}$ and $H_{max} = 8k + 2$. The quantity $\delta_C > 0$ is a fixed constant, while $\log H_{max} = \Oh{\log k}$. Therefore, \begin{align}
    \log \frac{1}{\delta_p} = \Oh{k + \log Z} \implies L = \left \lceil \log_2 \frac{1}{\delta_p} \right\rceil + 1 = \Oh{k + \log Z}
\end{align}

Therefore, each bundle probability can be approximated using only polynomially many bits and these rational approximations can be computed in polynomial time.

\end{proof}

At this point, we have converted the ideal construction into a polynomial-size rational instance without changing the YES/NO answer. All that remains is to collect the preceding lemmas and conclude the reduction.

\subsection{Completion of the proof}

\begin{proof}[Proof of Theorem~\ref{thm:mixed-and-or-hardness}]
Given an \textsc{Equal-Cardinality Partition} instance, first apply the padding reduction so that $k\geq k_0$. Construct the $2k$ target bundle maps, the pivot, and the fallback as above. Lemma~\ref{lem:effective-realization} realises the desired bundle maps using nearby probability vectors, and Lemma~\ref{lem:rationalisation} replaces these ideal probabilities by polynomial-bit rational probabilities in polynomial time.

It remains to collect what the preceding lemmas establish. \begin{itemize}
    \item Lemma~\ref{lem:atomisation} shows that the attacker gains nothing from arbitrary switching and resumption. An optimal adaptive policy can be represented by the macro problem in which each bundle is either executed in $H$-mode or $L$-mode.
    \item Lemma~\ref{lem:canonical-macro} then gives the canonical order\begin{equation*}
        H\text{-modes} \longrightarrow P \longrightarrow L\text{-modes}  \longrightarrow F
    \end{equation*}
    \item Lemma~\ref{lem:balanced-assignment} forces exactly $k$ of the $2k$ bundles to be in $H$-mode. Hence, an optimal macro policy determines a $k$-element subset $S$ of the original \textsc{Equal-Cardinality Partition} instance.
    \item Finally, Lemma~\ref{lem:partition-gap} shows that the value distinguishes exactly the remaining partition condition. In the ideal construction: \begin{align*}
        \text{YES}&\implies V^* \geq \Theta + \varpi_1\\
        \text{NO}&\implies V^* \leq \Theta - \varpi_1
    \end{align*}. Thus, the ideal network has value above $\Theta$ when there exists a $k$-element subset whose weights sum to $Z$.
    \item Lemma~\ref{lem:rationalisation} changes the optimal value by at most $\frac{\varpi_1}{4}$ when the ideal probabilities are replaced by their rational approximations. This is smaller than the gap on either side of $\Theta$, so the rational output instance still satisfies
    \begin{align*}
        V^* > \Theta &\qquad\text{in every YES-instance}\\
        V^* < \Theta &\qquad\text{in every NO-instance}
    \end{align*}
\end{itemize}

The construction is computable in polynomial time and has polynomial encoding length. Therefore, deciding whether the optimal value exceeds a given rational threshold is weakly NP-Hard. Exact computation of the optimal value is consequently weakly NP-Hard as well. Moreover, any algorithm that outputs an optimal policy together with a polynomial-time evaluable certificate of its value would solve the value-decision problem and is therefore weakly NP-hard under the same reduction. This completes the proof.
\end{proof}

\endgroup

\end{document}